\def\thxPACS{Funded by the Deutsche Forschungsgemeinschaft (DFG, German
	Research Foundation), project ``Parameterized Algorithmics in Computational Sustainability (PACS)''---FL~1247/1-1, 522475669.}

\def\HUaffil{Humboldt-Universität zu Berlin,
	Department of Computer Science, Algorithm Engineering Group, Germany}
\def\TForcid{0000-0003-2203-4386}

\def\thxMPMS{Funded by the Deutsche Forschungsgemeinschaft (DFG, German Research Foundation), project COMSOC-MPMS (grant agreement No.\ 465371386).}
\def\TUCaffil{Technische Universität Clausthal,
	Department of Computer Science, Foundations of Algorithms and Artificial Intelligence, Germany}
\def\EDorcid{0000-0003-0431-2733}
\def\RBorcid{0000-0002-6303-6276}

\documentclass[a4paper,UKenglish,cleveref, autoref, thm-restate]{lipics-v2021}

\pdfoutput=1 %
\hideLIPIcs  %

\newcommand{\mytitle}{Sustainable and Envy-free Multi-Period Allocation}

\title{\mytitle} %

\author{Robert Bredereck}{\TUCaffil}{robert.bredereck@tu-clausthal.de}{https://orcid.org/\RBorcid}{}

\author{Eva Deltl}{\TUCaffil}{eva.deltl@tu-clausthal.de}{https://orcid.org/\EDorcid}{\thxMPMS}

\author{Till Fluschnik}{\HUaffil}{till.fluschnik@hu-berlin.de}{https://orcid.org/\TForcid}{\thxPACS}%

\authorrunning{R.~Bredereck, E.~Deltl and T.~Fluschnik} %

\Copyright{Robert Bredereck, Eva Deltl and Till Fluschnik} %

\ccsdesc[500]{Theory of computation~Parameterized complexity and exact algorithms}
\ccsdesc[500]{Theory of computation~Algorithmic game theory and mechanism design}

\category{} %

\relatedversion{} %

\nolinenumbers %

\usepackage{mathrsfs}
\usepackage{enumitem}
\usepackage{amsthm}
\usepackage[super]{nth}
\usepackage[toc,page]{appendix}
\usepackage{listings}
\usepackage{mathtools}
\usepackage{color}
\usepackage{colortbl}
\usepackage{tabularx}
\usepackage{amsmath}
\usepackage{csquotes}
\usepackage{multicol}
\usepackage{fancyvrb}
\usepackage{makecell}
\usepackage{dsfont}
\usepackage{xargs}
\newcommandx{\set}[2][1=]{\if\relax#1\relax\ensuremath{[#2]}\else\ensuremath{[#1,#2]}\fi}
\newcommand{\xset}[3]{\ensuremath{\{#1_{#2},\dots,#1_{#3}\}}}
\newcommand{\xtpl}[3]{\ensuremath{(#1_{#2},\dots,#1_{#3})}}

\usepackage{mdframed}
\usepackage[toc,page]{appendix}
\usepackage{mathdots}
\usepackage{xspace}
\usepackage{empheq}
\usepackage{arydshln}
\usepackage{nicefrac}
\usepackage{dsfont}
\usepackage{paralist}
\usepackage{booktabs}

\usepackage[ruled, vlined, linesnumbered]{algorithm2e}

\usepackage{array}
\newcolumntype{L}[1]{>{\raggedright\let\newline\\\arraybackslash\hspace{0pt}}m{#1}}
\newcolumntype{C}[1]{>{\centering\let\newline\\\arraybackslash\hspace{0pt}}m{#1}}
\newcolumntype{R}[1]{>{\raggedleft\let\newline\\\arraybackslash\hspace{0pt}}m{#1}}

\usepackage{algpseudocode}

\usepackage{tikz}
\usetikzlibrary{automata, positioning, calc, shapes, arrows, fit, patterns.meta}
\usetikzlibrary{decorations.pathreplacing}

\DeclareMathOperator{\poly}{poly}

\makeatletter
\def\mk@font#1#2{%
	\expandafter\def\csname mk#1\endcsname##1{%
		\expandafter\gdef\csname #1##1\endcsname{\ensuremath{#2{##1}}}%
	}%
}
\mk@font{cal}{\mathcal} %
\mk@font{bf}{\mathbf} %
\mk@font{rm}{\mathrm} %
\@tfor\L:=A,B,C,D,E,F,G,H,I,J,K,L,M,N,O,P,Q,R,S,T,U,V,W,X,Y,Z\do{%
	\expandafter\mkcal\L %
	\expandafter\mkbf\L %
	\expandafter\mkrm\L %
}
\@tfor\L:=a,b,c,d,e,f,g,h,i,j,k,l,m,n,o,p,q,r,s,t,u,v,w,x,y,z\do{%
	\expandafter\mkbf\L %
	\expandafter\mkrm\L %
}
\makeatother

\newcommand{\N}{\mathbb{N}}
\newcommand{\Q}{\mathbb{Q}}
\newcommand{\Z}{\mathbb{Z}}
\newcommand{\Qnn}{\Q_{\geq 0}}
\newcommand{\Qpos}{\Q_{+}}

\newcommand{\Nzero}{\N_0}

\newcommand{\cocl}[1]{\textrm{#1}}
\newcommand{\classP}{\cocl{P}}
\newcommand{\NP}{\cocl{NP}}
\newcommand{\W}[1]{\cocl{W}[#1]}
\newcommand{\XP}{\cocl{XP}}
\newcommand{\FPT}{\cocl{FPT}}
\newcommand{\coNP}{\cocl{coNP}}
\newcommand{\cpoly}{\cocl{poly}}
\newcommand{\NPincoNPslashpoly}{\ensuremath{\NP\subseteq\coNP/\cpoly}}

\newcommand{\UnlessPK}{Unless \NPincoNPslashpoly}

\newcommand{\prob}[1]{\textnormal{\textsc{#1}}}
\newcommand{\wpb}{when parameterized by}

\newcommand{\false}{\bot}
\newcommand{\true}{\top}
\newcommand{\yes}{\texttt{yes}}
\newcommand{\no}{\texttt{no}}
\newcommand{\RD}{$(\Rightarrow)\quad$}
\newcommand{\LD}{$(\Leftarrow)\quad$}

\newcommand{\cif}{\text{if~}}

\newcommand{\etal}{et al.}

\newcommand{\eps}{\varepsilon}
\newcommand{\ceq}{\coloneqq}
\newcommand{\Indic}{\ensuremath{\mathds{1}}}

\newcommand{\fish}{\textsc{Sustainable Envy-Free Multi-Period Allocation}}

\newcommandx{\xassignment}[1][1=]{\if\relax#1\relax period assignment\else#1-period assignment\fi}

\newcommandx{\decprob}[6][3=Input,5=Question]{\vspace{0.125em}
	\begin{samepage}
		\begingroup
		\label{DEACTIVATEDprob:#2}{
			{\noindent \textsc{#1}}}
		\nopagebreak[4]\nopagebreak[4]
		\par\noindent\hangindent=\parindent\textbf{#3}:  #4\nopagebreak[4]
		\par\noindent\hangindent=\parindent\textbf{#5}:  #6
		\par\medskip\endgroup
	\end{samepage}
}

\usepackage{thmtools}

\crefname{theorem}{Theorem}{Theorems}
\Crefname{theorem}{Thm.}{Thms.}
\newtheorem{rrule}{Reduction Rule}
\crefname{observation}{Observation}{Observations}

\declaretheorem[style=definition,name=Construction,qed=$\diamond$]{construction}
\declaretheorem[style=definition,name=Construction,sibling=construction]{construction*}
\crefname{construction}{Construction}{Constructions}
\crefname{construction*}{Construction}{Constructions}

\crefname{algocf}{Algorithm}{Algorithms}

\colorlet{colorFPT}{green!50}
\colorlet{colorPNPH}{red!50!white}
\colorlet{colorXP}{yellow!90!black}
\colorlet{colorWH}{orange!50!white}

\newcommand{\tikzpreamble}{%
    \tikzstyle{param}=[
        minimum width=5*\xr em,
        minimum height=1.75*\yr em,
        rounded corners,
        draw
    ]

    \tikzstyle{colPNPH}=[fill=colorPNPH]
    \tikzstyle{colXP}=[fill=colorXP]
    \tikzstyle{colWH}=[fill=colorWH]
    \tikzstyle{colFPT}=[fill=colorFPT]

    \tikzstyle{colPNPHXPtr}=[
        preaction={colPNPH},
        pattern={Lines[
            angle=45,
            distance=6pt,
            line width=2pt
        ]},
        pattern color=colorXP
    ]

    \tikzstyle{colXPstar}=[
        preaction={colXP},
        pattern={Lines[
            angle=45,
            distance=6pt,
            line width=3pt
        ]},
        pattern color=white
    ]

    \tikzstyle{colXPWH}=[
        preaction={colXP},
        pattern={Lines[
            angle=90,
            distance=6pt,
            line width=3pt
        ]},
        pattern color=colorWH
    ]

    \tikzstyle{colFPTPK}=[
        preaction={colFPT},
        pattern={Lines[
            angle=90,
            distance=6pt,
            line width=1.5pt
        ]},
        pattern color=green!20!white
    ]

    \tikzstyle{colFPTNOPK}=[
        preaction={colFPT},
        pattern={Lines[
            angle=0,
            distance=6pt,
            line width=1.5pt
        ]},
        pattern color=green!60!black
    ]

    \tikzstyle{colNOPK}=[
        preaction={fill=white},
        pattern={Lines[
            angle=0,
            distance=6pt,
            line width=3pt
        ]},
        pattern color=green!60!black
    ]
}

\newcommand{\appsymb}{$\bigstar$}
\newcommand{\appref}[1]{{\hyperref[proof:#1]{\appsymb}}}

\usepackage{etoolbox}

\newcommand{\appendixproof}[2]{%
	\gappto{\appendixProofText}
		{
			\subsection{Proof of \cref{#1}}\label{proof:#1}
			#2
		}
}

\keywords{computational sustainability, fair allocation, envy-freeness, renewable resources, multi-period allocation, parameterized complexity}

\newcommand{\BibTeX}{\rm B\kern-.05em{\sc i\kern-.025em b}\kern-.08em\TeX}

\renewcommand{\r}{r}
\newcommand{\rSet}{R}
\newcommand{\iq}{s}
\newcommand{\lb}{\ell}
\newcommand{\ub}{u}
\newcommand{\rf}{\rho}
\newcommand{\q}{q}
\newcommand{\val}{v}

\newcommand{\nonwasteful}{approval-restricted}
\newcommand{\wasteful}{approval-unrestricted}
\newcommand{\NonWasteful}{Approval-Restricted}

\EventEditors{John Q. Open and Joan R. Access}
\EventNoEds{2}
\EventLongTitle{42nd Conference on Very Important Topics (CVIT 2016)}
\EventShortTitle{CVIT 2016}
\EventAcronym{CVIT}
\EventYear{2016}
\EventDate{December 24--27, 2016}
\EventLocation{Little Whinging, United Kingdom}
\EventLogo{}
\SeriesVolume{42}
\ArticleNo{23}

\begin{document}

\maketitle

\begin{abstract}
	Allocating renewable resources fairly requires accounting for how present decisions affect future availability. We introduce \emph{Sustainable and Envy-Free Multi-Period Allocation}, in which agents are assigned over multiple periods to renewable resources whose quantities evolve over time according to renewal functions. An allocation must respect viability and capacity constraints while being envy-free over the whole time horizon. We study the computational and parameterized complexity of finding such allocations, with a focus on the three dimensions of the problem: the numbers of agents, resources, and periods.
	We show that finding such an allocation is NP-hard already for unary encoding and any single one of them being fixed to a constant.
	Whether two small dimensions suffice for tractability depends on the pair,
	the class of renewal functions,
	and the numerical encoding.
	For few resources and periods,
	binary-encoded instances remain hard in general,
	whereas unary encoding yields XP-time algorithms.
	For few resources and agents,
	binary-encoded instances also remain hard in general
	but become XP-time solvable when all renewal functions are translations.
	Few agents and periods yield fixed-parameter tractability,
	but no polynomial kernel under standard complexity assumptions.
\end{abstract}

\section{Introduction}

Fair division classically assumes a fixed pool of goods:
the items are given in advance,
and an allocation only decides who receives what~\cite{AmanatidisEtAl2023FairDivisionSurvey}.
This remains true in recent models of multi-period fair division~\cite{IgarashiEtAl2024RepeatedFairAllocation,ElkindEtAl2025TemporalFairDivision,CooksonEbadianShah2025TemporalFairDivision,CaragiannisNarang2024RepeatedMatching},
where the goods of each round are exogenous.
Many resources,
however,
are \emph{renewable}:
they form a stock that regenerates over time,
so that what is available tomorrow depends on how much is used today.
When such resources are allocated over multiple periods,
the decisions of different periods are coupled through the dynamics of the stocks,
and the supply of later periods is shaped by the allocation itself.
We ask how this single feature---renewal---changes the algorithmics of fair allocation.%

To isolate this question,
we introduce a deliberately minimal model,
\fish{},
with $m$~resources,
$n$~agents,
and $\tau$~periods.
Each resource has a starting quantity,
a lower viability bound and an upper capacity bound,
and a renewal function;
we focus on affine renewal functions~$x\mapsto\alpha x+\beta$.
In every period,
each agent may be assigned to one resource or remain unassigned.
Each agent has a weight,
capturing heterogeneous resource usage,
for example due to differences
in demand or harvesting capacity.
An assignment to a resource reduces its current quantity according to the weights of the assigned agents,
after which the resource evolves according to its renewal function.
An allocation is \emph{sustainable} if
it never extracts more than the currently available quantity (\emph{no over-extraction}),
never lets a quantity exceed its upper bound (\emph{upper capacity}),
does not harvest a resource whose quantity is below its lower bound (\emph{viability protection}),
and leaves every resource viable at the end (\emph{terminal viability}).
At the same time,
we require fairness across periods:
agents approve subsets of the resources,
each agent evaluates the sequence of resources assigned to her,
and an allocation is \emph{envy-free} if no agent would prefer the sequence of assignments received by another agent.
The resulting computational problem asks whether there exists a multi-period assignment
that is simultaneously sustainable and envy-free
(see \cref{sec:prelims:model} for the formal definition).

For instance,
in the fisheries interpretation,
an agent represents a fisher and an assignment gives the fisher access to a particular fish stock in a given period;
the fisher's weight reflects her harvesting capacity.
Harvesting reduces the stock,
which subsequently renews before the next period.
A stock that falls below its viability threshold is temporarily closed to harvesting,
while the upper bound captures capacity constraints such as overpopulation that may destabilize the ecosystem.
Fairness, in contrast, is evaluated over the complete sequence of access decisions rather than separately in every period.
The same framework applies equally to managed stocks such as fish farms
and to other renewable resources such as grazing land or managed wildlife populations.

\subparagraph{A minimal model of fair allocation with renewable resources.}
Our guiding question is a technical one:
which component of the model is responsible for which algorithmic effect?
To this end,
we deliberately keep the model simple,
aiming to isolate the algorithmic interaction between fair allocation,
sustainability constraints,
and renewal dynamics.
On the side of fair allocation,
we keep the simplest setting:
agents merely approve or disapprove resources,
and fairness is envy-freeness without any relaxation;
hence,
none of our hardness results relies on rich preferences.
On the side of renewal,
the dynamics are kept equally simple.
For a single period,
renewal acts only once,
through the bounds that the next quantity must respect,
so that the problem is a one-shot fair allocation problem with bounds on the harvest from each resource;
over multiple periods,
already translations~$x\mapsto x+\beta$,
arguably the simplest non-trivial renewal functions,
suffice for hardness. %
Finally,
the upper capacity bounds are essential:
since harvesting never increases future quantities,
without them the empty allocation---which is trivially envy-free---would be sustainable
whenever any allocation is.
With them,
sustainability may enforce harvesting,
envy-freeness propagates enforced assignments to further agents,
and whether an assignment is feasible depends,
through the renewal functions,
on the entire allocation history.
We therefore regard \fish{} as the basic problem of fair allocation with renewable resources,
rather than as a model of a particular application.
Within this model,
we vary the agent weights,
the restriction to approved resources,
the number of periods,
the encoding of the numbers,
and the class of renewal functions,
together with the three dimensions---agents, resources, and periods---and their combinations.

\subsection{Our Contributions}
We provide a parameterized complexity analysis of our problem regarding
the number~$n$ of agents,
$m$ of resources,
and $\tau$ of periods,
as well as their combinations.
Thereby,
we distinguish between unary and binary encoding,
as well as different classes of affine renewal functions,
with a focus on translations.
\Cref{fig:Overview} summarizes our results.
\begin{figure}
	\centering
	\begin{tikzpicture}
		\def\xr{1.5}
		\def\yr{1}
		\tikzpreamble{}

		\newcommand{\resultnode}[7]{%
			\node (#1) at (#2*\xr ,#3)
				[very thick, #4, draw,
				rounded corners=2pt,
				minimum width=1.8*\xr cm,
				minimum height=1.15cm,
				inner sep=3pt] {};

			\node at (#1.center)
				[#5, draw,
				rounded corners=1pt,
				minimum width=1.3*\xr cm,
				minimum height=0.7cm,
				align=center,
				inner sep=2pt]
				{{\boldmath \(#6\)}\\[-1mm]
				{\scriptsize \Cref{#7}}};
		}

		\newcommand{\splitnode}[9]{%
			\node (#1) at (#2*\xr ,#3)
			[draw, very thick,
			rounded corners=2pt,
			minimum width=1.8*\xr cm,
			minimum height=1.15cm,
			inner sep=0pt] {};

			\begin{scope}
				\clip[rounded corners=2pt] (#1.south west) rectangle (#1.north east);
				\fill[#4] (#1.west) rectangle (#1.north east);
				\fill[#5] (#1.south west) rectangle (#1.east);
			\end{scope}

			\draw[black, very thick, rounded corners=2pt]
				(#1.south west) rectangle (#1.north east);

			\node (x#1) at (#1.center)
				[draw=none,
				rounded corners=1pt,
				minimum width=1.3*\xr cm,
				minimum height=0.7cm,
				inner sep=2pt] {};

			\begin{scope}
				\clip[rounded corners=1pt] (x#1.south west) rectangle (x#1.north east);
				\fill[#6] (x#1.west) rectangle (x#1.north east);
				\fill[#7] (x#1.south west) rectangle (x#1.east);
			\end{scope}

			\draw[black, rounded corners=1pt] (x#1.south west) rectangle (x#1.north east);

			\node at (#1.center)
				[align=center, inner sep=0pt]
				{{\boldmath #8}\\[-1mm]
				{\scriptsize \Cref{#9}}};
		}

		\resultnode{n}{0}{0.5}{colPNPH}{colPNPH}{n}{thm:n:bin:pnph,thm:n:unary:pnph}
		\resultnode{m}{-2.25}{0.5}{colPNPH}{colPNPH}{m}{thm:mtau:weak,thm:whardmtau}
		\resultnode{tau}{2.25}{0.5}{colPNPH}{colPNPH}{\tau}{thm:mtau:weak,thm:whardmtau}

		\splitnode{nm}{-2.25}{2}{colPNPHXPtr}{white}{colPNPHXPtr}{colXPstar}{$n+m$}{thm:nm:translations,thm:nm:pnph,thm:nm:integralslope}
		\resultnode{ntau}{2.25}{2}{colNOPK}{colFPT}{n+\tau}{thm:ntau:nopk,thm:ntau-fpt}
		\splitnode{mtau}{0}{2}{colPNPH}{colWH}{colPNPH}{colXP}{$m+\tau$}{thm:mtau:weak,thm:whardmtau,thm:mtau:unary:XP}

		\draw[->, thick] (n) -- (nm);
		\draw[->, thick] (m) -- (nm);
		\draw[->, thick] (m) -- (mtau);
		\draw[->, thick] (tau) -- (mtau);
		\draw[->, thick] (n) -- (ntau);
		\draw[->, thick] (tau) -- (ntau);

		\node at (3.4*\xr,2.4*\yr+0.05*\yr)[colPNPH,inner sep=4pt,label=0:{para-\NP-hard}]{};
		\node at (3.4*\xr ,2.*\yr+0.05*\yr)[colPNPHXPtr,inner sep=4pt,label=0:{para-\NP-hard, $\XP_{\text{tr}}$}]{};
		\node at (3.4*\xr ,1.6*\yr+0.05*\yr)[colWH,inner sep=4pt,label=0:{\W{1}-hard}]{};
		\node at (3.4*\xr ,1.2*\yr+0.05*\yr)[colXP,inner sep=4pt,label=0:{XP}]{};
		\node at (3.4*\xr ,0.8*\yr+0.05*\yr)[colXPstar,inner sep=4pt,label=0:{$\XP_{\alpha\in\N}$}]{};
	  \node at (3.4*\xr ,0.4*\yr+0.05*\yr)[colFPT,inner sep=4pt,label=0:{FPT}]{};
	  \node at (3.4*\xr ,0*\yr+0.05*\yr)[colNOPK,inner sep=4pt,label=0:{No PK}]{};
	\end{tikzpicture}
	\caption{Overview of our parameterized complexity results.
	For each parameter,
	the inner and outer boxes correspond to the best complexity upper and lower bound
	we established
	(white means that the classification is unknown).
	For nodes split horizontally,
	the upper color corresponds to binary
	encoding and the lower color to unary encoding.
	$\XP_{\text{tr}}$ means \XP{} when all renewal functions are translations,
	and $\XP_{\alpha\in\N}$ means \XP{} when all renewal functions have integral slope.}
	\label{fig:Overview}
\end{figure}
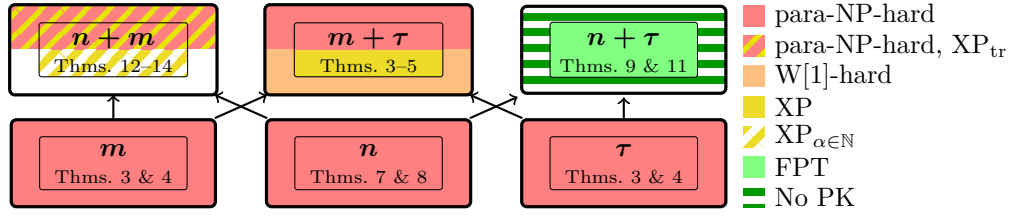
Our results give an almost complete classification and show that the problem is remarkably intractable.
Even for a single period,
the problem is \NP-hard with agents of only two different weights (\cref{thm:t1:hardness}).
This result is tight since with unit weights,
the problem becomes polynomial-time solvable (\cref{thm:poly-t1w1}),
as long as agents are assigned only to resources they approve.
Over several periods,
however,
unit weights no longer help (\cref{thm:unaryweight}).
Among all single dimensions and their pairs,
only for~$n+\tau$ we establish~fixed-parameter tractability
(\cref{thm:ntau-fpt}),
which we complement with an incompressibility result~(\cref{thm:ntau:nopk}).

We particularly highlight the following results.
First, the problem is \NP-hard already for three agents,
even with unary encoding and all renewal functions being translations (\cref{thm:n:unary:pnph}).
Second,
under binary encoding, hardness holds even for one agent
(\cref{thm:n:bin:pnph}),
so sustainability alone is already computationally hard.
The problem remains \NP-hard for five agents and three resources
(\cref{thm:nm:pnph}),
and for one period and two resources or one resource and two periods
(\cref{thm:mtau:weak}).
We identify two different complexity boundaries for~$m+\tau$ and~$n+m$:
\begin{compactitem}
 \item For~$m+\tau$,
the binary encoding is crucial for the hardness,
as with unary encoding the problem is in \XP{} (\cref{thm:mtau:unary:XP})
and hence becomes polynomial-time solvable whenever~$m+\tau$ is constant.
In other words,
to encode a hard problem into ours with few resources and periods,
we need numbers that are large relative to the instance size.
Moreover,
the dependence of the polynomial degree on~$m+\tau$ is unavoidable
under standard assumptions,
as the problem is \W{1}-hard (\cref{thm:whardmtau}).
 \item For~$n+m$,
the hardness breaks as soon as all renewal functions are translations:
then,
the problem becomes polynomial-time solvable whenever~$n+m$ is constant,
even under binary encoding (\cref{thm:nm:translations});
under unary encoding,
this extends to all renewal functions with integral slope (\cref{thm:nm:integralslope}).
The crucial property is that,
for a translation,
the quantity at any period depends only on the cumulative harvest up to this period,
whereas for slope~$\alpha\neq1$ it also depends on \emph{when} the harvests took place.
Our hardness proof exploits exactly this by encoding the digits of large numbers in the timing of harvests.
Thus, to encode a hard problem into ours with few agents and resources,
there must exist at least one resource whose quantity depends on its period-by-period harvesting history.
In contrast,
when considering the number of agents alone,
the problem remains \NP-hard already for three agents,
even under unary encoding and with all renewal functions being translations
(\cref{thm:n:unary:pnph}). There, the construction uses their order-independence to time all events of a resource solely by the amount harvested from it.
\end{compactitem}
The only cases left open concern~$n+m$ under unary encoding:
whether the problem is \FPT{} or \W{1}-hard,
and whether the \XP-algorithm extends
to renewal functions with rational~slopes. For results marked~\appsymb,
we give proof ideas in the main text and full proofs in the appendix.

\subsection{Related Work}

\subparagraph{Motivational background.}
Computational sustainability~\cite{GomesEtAl2019ComputationalSustainability,Gomes2010FrontiersComputationalSustainability,LaessigKerstingMorik2016ComputationalSustainability} is a branch of computer and information science devoted to problems arising from sustainability sciences.
Allocating renewable resources in a fair way is one of those problems,
studied already in 1979 by
Sobel~\cite{Sobel1979FairAllocations},
nowadays linked to bioeconomics~\cite{Clark2010MathematicalBioeconomics}.
Fisheries are a recurring concrete application and motivation for fair allocation of renewable resources.
Such tasks are studied through stochastic processes~\cite{KormanEtAl2020ReplenishingResources},
including Markov decision processes~\cite{ErmonEtAl2010PlayingGames,ErmonEtAl2011RiskSensitive}.

Allee effects~\cite{CourchampBerecGascoigne2008AlleeEffects} describe criticality for species at small population size.
Hence,
viability becomes crucial for conservation programs.
In part,
such criticality depends on the population's renewal function,
often modeled such that the current stock determines the population size of the subsequent period.
Such functions are often difficult to determine~\cite{SubbeyEtAl2014StockRecruitment},
and are mostly more complex than affine functions.
Yet,
under several assumptions,
conservation programs may handle even a priori unknown renewal functions~\cite{Clarke1984UnknownRenewalFunction}.
Still,
the rapid pace of human harvesting
poses a severe threat to fish populations and species~\cite{DulvySadovyReynolds2003}. Beyond biological constraints, the management of renewable resources also
depends on human behavior.
Threshold effects can shape cooperation in collective-resource dilemmas
\cite{cooperation}, while altruistic behavior has also been studied in the
context of environmental conservation~\cite{altruism}.

\subparagraph{Related models.}
As part of bioeconomics,
fairness is crucial when it comes to controlling and managing harvesting of natural resources by human entities.
Among many,
envy-freeness is one established concept of fairness~\cite{Varian1974EquityEnvyEfficiency}.
As envy-freeness is often difficult to achieve,
relaxations of it are studied,
such as proportionality and envy-freeness up to one item;
Amanatidis~\etal~\cite{AmanatidisEtAl2023FairDivisionSurvey} survey such relaxations in the context of fair division of indivisible items.
Since we restrict to binary agent valuations,
related to our work is the stream on envy-free matchings in graphs
\cite{AIGNERHOREV2022164}.
Envy-free matching problems with lower quotas~\cite{Yokoi2017EnvyFreeLowerQuotas,krishnaa2022}
model the requirement that some objects be assigned a minimum number of agents;
this resembles our capacity constraints. However, these models consider justified envy, which takes preferences on both sides into account.
In our model, resources have no preferences over agents,
and envy depends only on the agents' approvals.
Notably, finding a maximum-cardinality matching with lower and upper quotas,
the latter resembling our no-over-extraction constraint,
is \NP-hard when objects may remain unused and lower quotas
apply only to used objects~\cite{ArulselvanEtAl2018LowerQuotas}.
Boehmer and Heeger~\cite{BoehmerH22} study stable house allocations
with lower and upper quotas and one-sided preferences.
They show that this problem is \FPT{} in the number of houses.

There is a vast literature on single-period matching and fair allocation,
and in the recent years,
these problems are also studied over multiple periods.
Igarashi~\etal~\cite{IgarashiEtAl2024RepeatedFairAllocation} studied multi-period fair allocation of indivisible items.
Among other results,
they show that an envy-free allocation always exists if the number of periods is a multiple of the number of agents,
but may fail to exist otherwise.
Elkind~\etal~\cite{ElkindEtAl2025TemporalFairDivision} study multi-period fair division of indivisible items
and introduce the notion of \emph{temporal} envy-freeness up to one item.
Cookson~\etal~\cite{CooksonEbadianShah2025TemporalFairDivision} study multi-period fair division
with two different fairness angles:
one fairness criterion per period,
and one over all periods.
Caragiannis and Narang~\cite{CaragiannisNarang2024RepeatedMatching} study multi-period one-to-one matchings
with envy-freeness up to one item.
While we study static binary valuations,
their valuations may change over time.
They prove \NP-hardness results.

To the best of our knowledge,
the computational and parameterized complexity of fair multi-period allocation
of renewable resources, as initiated by Sobel~\cite{Sobel1979FairAllocations},
or a closely related model has not previously been studied.

\subparagraph{Complexity comparisons.}
Bliem~\etal~\cite{BliemBN16} study the problem of allocating resources to agents
in a Pareto-efficient and envy-free way.
They study different preference profiles and show that for additive utilities when parameterized by the number of agents,
their problem is para-\NP-hard under binary encoding but \XP{} under unary encoding.
Our para-\NP-hardness for~$n$ holds even under unary encoding.
Elkind~\etal~\cite{ElkindEtAl2025TemporalFairDivision}
show that deciding whether a temporal envy-free allocation up to one item exists
is \NP-hard already for three agents,
whereas such an allocation always exists and can be computed in polynomial time for two agents or two periods.
In contrast,
our problem is \NP-hard already for one period.
Moreover,
they show that when additionally demanding Pareto-optimality,
the problem becomes \NP-hard for two agents.
Their hardness proofs require binary encoding,
whereas our problem is \NP-hard for three agents already under unary encoding (\cref{thm:n:unary:pnph}).

Bredereck~\etal~\cite{BredereckFK22} study a temporal voting problem,
called multistage voting,
where consecutive committees should either agree or disagree on common candidates.
They prove that already for two voters,
their problem is \NP-hard.
However,
when combining any two of the number of agents, of candidates, and of time steps,
the problem becomes fixed-parameter tractable.
Our problem,
in contrast,
remains para-\NP-hard for~$n+m$ and~$m+\tau$ under binary encoding,
and is \W{1}-hard for~$m+\tau$ under unary encoding.

For temporal graphs, the number of vertices is most naturally compared to~$n+m$ in our model
and
often yields \FPT{} or \XP{} classifications.
We discuss two relevant exceptions~\cite{FroeseKZ24,ErlebachMSW24}.
Froese~\etal~\cite{FroeseKZ24} study the problem of covering
all temporal edges with at most~$k$ size-$\ell$ intervals for each vertex.
They show that the problem is \W{1}-hard yet \XP{} when parameterized by the number of vertices.
For our problem under unary encoding,
we did not find a parameterized lower or upper bound for~$n+m$.
Erlebach~\etal~\cite{ErlebachMSW24} study the cops and robber game on edge-periodic graphs,
which are temporal graphs that allow for a compact representation due to periodicity even for large lifetimes.
They prove that deciding whether a cop-winning strategy exists is \W{1}-hard when parameterized by the number of vertices.
For their proof,
they exploit that the periodic representation allows lifetimes exponential in the input size.
Our problem is even para-\NP-hard with respect to~$n+m$ under binary encoding.

\section{Preliminaries}\label{sec:prelims:model}

We denote by~$\N$ and~$\Nzero$ the natural numbers excluding and including zero,
respectively.
We denote by~$\Q$ the rational numbers,
and by~$\Qnn$ all nonnegative and by~$\Qpos$ all positive rational numbers.
For~$m,n \in \N$ with $m\leq n$,
we denote by~$\set[m]{n}$ the set~$\{m,\ldots,n\}$ and $\set{n}$ for short when~$m=1$.

\smallskip\noindent
\textbf{Our Model and Problem.}
We denote by $\rSet=\xset{\r}{1}{m}$ the set of resources.
A resource~$\r_i=(\iq_i,\lb_i,\ub_i,\rf_i)$ consists of a starting quantity $\iq_i \in \Qnn$,
a lower viability bound $\lb_i \in \Qnn$,
an upper capacity bound $\ub_i \in \Qnn$,
and a renewal function $\rf_i \colon \Q \to \Q$.
In this work,
we focus on affine renewal functions,
i.e.,
$\rf_i(x)=\alpha_i\cdot x+\beta_i$,
with~$\alpha_i\in\Qpos$ and~$\beta_i\in\Q$.
Affine functions capture,
in a simplified way,
two key features of renewal functions:
a quantity dependence through~$\alpha$
(e.g., birthrate),
and an external influence~$\beta$
(e.g., migration).
We say that an affine function $\alpha\cdot x+ \beta$ has integral slope if and only if~$\alpha\in\N$.
The special case of $x+\beta$,
i.e., with $\alpha=1$,
is called a \emph{translation}.
Translations naturally model resources with fixed quantity-independent change.
An example is a water reservoir with approximately fixed inflow,
where the lower bound represents a minimum operational water level and
the upper bound the reservoir's capacity.

We denote by~$A = \xset{a}{1}{n}$ the set of \emph{agents},
where each agent is associated with a weight $w_a \in \N$
and
a valuation function $v_a \colon \rSet \to \Nzero$ extended by $v_a(\emptyset)\ceq 0$.
In this work,
we focus on \emph{binary} valuations,
i.e.,
$\val_a(\r) \in \{0,1\}$,
and say that $a$ \emph{approves} resource $\r$ if $v_a(\r) = 1$.
Thus,
we can represent the agents' approvals by an \emph{approval} graph~$G=(X_\rSet\uplus X_A,E)$
which is
a bipartite graph with parts~$X_\rSet=\{x_\r\mid \r\in \rSet\}$ and~$X_A=\{x_a\mid a\in A\}$
such that it contains the edge~$\{x_a,x_\r\}$ if and only if~$a$ approves~$\r$.

For a number $\tau\in\N$ of periods,
we define a \xassignment[$\tau$] $\pi = \xtpl{\pi}{1}{\tau}$ with assignments~$\pi_t\colon A \to \rSet\cup\{\emptyset\}$ for all $t\in \set{\tau}$.
An assignment is called \emph{\nonwasteful}
if $\pi_t(a)=\r_j$ implies $v_a(\r_j)=1$
for all agents, resources, and periods.
In the \wasteful{} setting, assignments to unapproved resources
are also allowed.
Unless stated otherwise, assignments are \nonwasteful.
Given a \xassignment[$\tau$] $\pi$,
we define the harvest of resource~$\r_i$ at period~$t$ as
$h_i^\pi(t)\coloneqq \sum_{a\in \pi_t^{-1}(\r_i)} w_a$.
We define the quantity~$\q^{\pi}_{i} \colon \set{\tau+1} \to \Q$
of each resource $\r_i$ recursively by
\[
\q^{\pi}_{i}(t) =
\begin{cases}
\iq_i & \text{if } t = 1, \\
\rf_i\left( \q^{\pi}_{i}(t - 1) - h_i^\pi(t-1) \right) & \text{if } t > 1.
\end{cases}
\]

We say that a \xassignment[$\tau$] $\pi = (\pi_1,\dots,\pi_\tau)$ is \emph{sustainable}
if the following constraints are satisfied for every resource~$\r_i\in \rSet$:
\smallskip
\begin{compactenum}[(I)]
		\item \textbf{No over-extraction:} For every period~$t\in\set{\tau}$, we have $h_i^\pi(t)\le \q_i^\pi(t)$.\label{model:noe}
        \item \textbf{Upper capacity:}  For every period~$t \in \set{\tau+1}$, we have $\q_i^{\pi}(t)\leq \ub_i$.\label{model:uc}
        \item \textbf{Viability protection:} For every period~$t\in\set{\tau}$,
        we have $\q_i^\pi(t) < \lb_i\implies \pi_t^{-1}(\r_i) = \emptyset$.\label{model:vp}
        \item \textbf{Terminal viability:} We have $\q_i^{\pi}(\tau+1)\geq \lb_i$.\label{model:tv}
\end{compactenum}
\smallskip
We allow a resource quantity to drop to zero.
We interpret this as temporary exhaustion of the quantity available within the managed system,
rather than irreversible depletion of the underlying resource.
Thus,
a resource may recover from zero.
Before applying any of our algorithms,
we reject any instance with~$\iq_j>\ub_j$ for some resource~$\r_j$, since this already violates the capacity constraint in period~$1$.
Given an \xassignment[$\tau$]~$\pi$,
we say that a resource~$\r_i$ is \emph{urgent} in period~$t\in\set{\tau}$ if its quantity would exceed its upper bound in period~$t+1$ if not harvested in period~$t$,
i.e.,
$\rf_i(q_i^\pi(t))>\ub_i$.
We call a resource~$\r_i\in \rSet$ \emph{uncritical} if for every \xassignment[$\tau$]~$\pi$
that assigns no agent to~$\r_i$,
it holds that
$0\leq \q_{i}^{\pi}(t) \leq \ub_i$ for all $t\in\set{\tau+1}$
and~$\lb_i\leq \q_{i}^{\pi}(\tau+1)$;
otherwise, we call the resource \emph{critical}.

A sustainable \xassignment[$\tau$] is \emph{envy-free} if
for all pairs of agents~$a, a' \in A$ we have that~$\sum_{t=1}^\tau v_a(\pi_t(a)) \geq \sum_{t=1}^\tau v_{a}(\pi_t(a'))$.
We introduce the following decision problem:

\smallskip
\decprob{\fish}{find}{
	A set $\rSet$ of resources, a set $A$ of agents, and a number $\tau$ of periods.
}{Is there a \xassignment[$\tau$] $\pi$ of $A$ to $\rSet$ that is sustainable and envy-free?}

As a convention throughout the paper,
the number $\tau$ of periods is always encoded in unary.
For affine renewal functions, this implies that \fish{} is in \NP{},
since a~$\tau$-period assignment has polynomial size
and all constraints can be verified in polynomial time.
In particular,
exact rational quantities under affine renewal functions have polynomial binary-encoding length for a unary horizon.
Thus,
we refer to unary- and binary-encoded instances
when all agent- and resource-specific numerical parameters are encoded in unary
and binary,
respectively.
We assume rational numbers to be represented as~$p/q$ with~$p\in \Z$, $q\in\N$.
Unless stated otherwise,
these numerical parameters are encoded in binary.

\smallskip\noindent
\textbf{Parameterized Complexity.} We adopt standard notation from parameterized complexity and algorithm design~\cite{CFKL15,Nie06,FG04}.
A \emph{parameterized problem} is a subset $L\subseteq\Sigma^*\times\N$, where $\Sigma$ is a finite alphabet. The integer $\kappa$ is called the \emph{parameter} of an instance $(x,\kappa)$. The problem $L$ is \emph{fixed-parameter tractable} (FPT) if there exists an algorithm that decides whether $(x,\kappa)\in L$ in time $f(\kappa)\cdot |x|^{O(1)}$, where $f$ is a computable function depending only on $\kappa$. The class XP consists of parameterized problems that can be solved in time $|x|^{f(\kappa)}$ for some computable function $f$. \W{1}-hard problems are believed not to be in FPT. A problem is \emph{para-NP-hard} with respect to a parameter $\kappa$ if it remains NP-hard when $\kappa$ is fixed to a constant.
A \emph{kernelization} for~$L$ is an algorithm transforming an instance $(x,\kappa)$ in polynomial time into an instance $(x',\kappa')$ such that $(x,\kappa)\in L \iff (x',\kappa')\in L$ and $|x'|+\kappa'\leq f(\kappa)$ for some computable function~$f$ of~$\kappa$.
A decidable problem is \FPT{} if and only if it admits a kernelization.
If $f$ is polynomial,
we say that the problem admits a \emph{polynomial kernel}.

\section{Single-Period Allocation: Weights and Approval Restrictions Matter}

We first study the single-period case with~$\tau=1$.
This isolates the interaction between sustainability and envy-freeness from renewals over multiple periods.
First,
we show that our problem is polynomial-time solvable for unit-weight agents.
Then,
we show that
agents of different weights or allowing \wasteful{} assignments
make the problem~\NP-hard.

\begin{theorem}[\appref{thm:poly-t1w1}]
	\label{thm:poly-t1w1}
	\fish{} is polynomial-time solvable when every agent has weight one and $\tau=1$.
\end{theorem}

\begin{proof}[Proof idea of \cref{thm:poly-t1w1}]
Since all weights are one, sustainability gives a lower and
an upper bound on the number of agents assigned to each resource.
We compute these intervals and reject if any is empty.
We then replace each resource by a number of copies equal to its upper bound on the number of assigned agents, capped at~$|A|$, and connect agents to copies of approved resources.
Assignments now correspond to matchings.

The envy-free matching decomposition of Aigner-Horev and Segal-Halevi~\cite{AIGNERHOREV2022164}
identifies a subgraph containing every envy-free matching.
The retained subgraph admits a matching assigning all its agents,
and every such matching is envy-free in the original graph.
Note that any envy-free matching can be extended to such a matching
by only adding edges: take a matching assigning all retained agents and consider
its edges incident to currently unassigned agents.
By envy-freeness,
their approved resources are also unassigned,
so these edges can be added.
Adding edges preserves the lower bounds on the number of agents assigned to each resource, while the number of copies enforces the corresponding upper bounds.
Thus, we may require every remaining agent to be matched.

All copies of a resource have the same neighborhood, so the decomposition keeps either all of them or none.
If an excluded resource requires an assignment, we reject.
Otherwise, we use a flow network with lower and upper bounds to
determine whether all remaining agents can be assigned to approved resources while respecting the feasible intervals.
Such an assignment is sustainable and envy-free,
and the above argument ensures that no solution is lost.
All steps take polynomial time.
\end{proof}

\appendixproof{thm:poly-t1w1}{
\begin{proof} Given an instance $I'=(\rSet,A,1)$ with $w_a=1$ for all $a\in A$, observe that a sustainable envy-free assignment corresponds exactly to an envy-free matching respecting all lower and upper bounds. Indeed, since $\tau=1$ and valuations are binary, under a \nonwasteful{} assignment an agent~$a$ envies another agent~$a'$ if and only if $a$ is unassigned and $a'$ is assigned to a resource approved by~$a$. This is precisely the envy-free matching condition in the corresponding graph~\cite{AIGNERHOREV2022164}.

	Our algorithm proceeds in three steps.

	We first compute, for every resource $\r_j$, the feasible range for the
	number of assigned agents. For an assignment $\pi$ to be sustainable, it must hold that $\lb_j \le \rf_j(\iq_j-h^\pi_j(1)) \le \ub_j$.
	Moreover, $0\le h^\pi_j(1)\le \min\{|A|,\lfloor\iq_j\rfloor\}$. If $\iq_j<\lb_j$, then $h^\pi_j(1)=0$ (viability in period~1). As $\rf_j$ is affine, the set of integers $h^\pi_j(1)$ satisfying these inequalities is an interval. We can
	therefore compute in polynomial time its smallest and largest elements, denoted by $\underline{m}_j$ and $\overline{m}_j$. If this interval is
	empty for some $\r_j$, the instance is infeasible.

	We then construct an uncapacitated bipartite graph $G=(X\cup Y,E)$. The left side corresponds to the set of agents, that is, $X=A$. The right side is constructed as follows: For every resource $\r_j$, we create $\overline{m}_j$ copies of $\r_j$. If $v_a(\r_j)=1$, we add an edge between $a$ and every copy of $\r_j$.

	Let $(X_S,X_L,Y_S,Y_L)$ denote the EFM decomposition of $G$, as defined by~\cite{AIGNERHOREV2022164}. It can be computed in polynomial time by first computing a maximum matching and then considering alternating reachability from unmatched vertices in $X$. We use the following properties of the decomposition:

	\smallskip\noindent
	\begin{compactenum}
		\item every envy-free matching is contained in $G[X_L,Y_L]$;
		\item $G[X_L,Y_L]$ admits a matching saturating $X_L$;
		\item every matching in $G[X_L,Y_L]$ that saturates $X_L$ is envy-free.
	\end{compactenum}

	\smallskip
	For every resource $\r_j$, let $C_j$ denote its set of $\overline{m}_j$ copies. Since all copies of $\r_j$ have identical neighborhoods, they all belong to the same part of the EFM decomposition. Indeed, suppose that some copy of $\r_j$ belongs to $Y_S$. Then it has a neighbor in $X_S$. Since every copy of $\r_j$ has the same neighborhood and there are no edges between $X_S$ and $Y_L$, all copies of $\r_j$ must belong to $Y_S$. Hence,	$C_j\subseteq Y_S$
	or $C_j\subseteq Y_L$.
	Consequently, if $\underline{m}_j>0$ and $C_j\subseteq Y_S$ for some $\r_j$, then no feasible envy-free matching exists, since every envy-free matching is contained in $G[X_L,Y_L]$ and therefore cannot use any clone of $\r_j$.

	Assume now that $C_j\subseteq Y_L$ for every resource $\r_j$ with $\underline{m}_j>0$. We construct a feasible-circulation instance on the agents
	in $X_L$ and the corresponding original resources, with arcs
	\begin{align}
		s &\to a && [1,1]
		&& \forall a\in X_L, \\
		a &\to \r_j && [0,1]
		&& \forall a\in X_L \text{ with } v_a(\r_j)=1
		\text{ and } C_j\subseteq Y_L, \\
		\r_j &\to t && [\underline{m}_j,\overline{m}_j]
		&& \forall \r_j \text{ with } C_j\subseteq Y_L, \\
		t &\to s && [|X_L|,|X_L|]. &&
	\end{align}

	By flow conservation, the flow on the arc $\r_j\to t$ equals the number
	of agents assigned to $\r_j$. Hence, the bounds on this arc enforce $\underline{m}_j\le h^\pi_j(1) \le \overline{m}_j$.

	If the circulation is feasible, integrality yields an integral flow and
	thus an assignment that matches every agent in $X_L$ to an approved resource whose copies lie in $Y_L$. Equivalently, this induces a matching in
	$G[X_L,Y_L]$ that saturates $X_L$. By the properties of the EFM
	decomposition, this matching is envy-free, and the bounds on the arcs
	$\r_j\to t$ ensure that all lower and upper bounds are satisfied.

	Conversely, suppose that a sustainable envy-free assignment exists. Since $\tau=1$, this corresponds to an envy-free matching that respects all lower and upper bounds. By the EFM decomposition, it is contained in $G[X_L,Y_L]$. Moreover, every envy-free matching can be extended, by only adding edges, to an envy-free matching saturating $X_L$~\cite[Lemma~6.1]{AIGNERHOREV2022164}.
	Since this extension only adds assignments, all lower bounds remain satisfied, while the upper bounds remain satisfied because each resource $\r_j$ has only $\overline{m}_j$ copies. Hence, there exists an envy-free matching saturating $X_L$ respecting the bounds, which induces a feasible circulation in the network above.
\end{proof}}

Next we show that the restriction to
unit weights and the assumption of \nonwasteful{} assignments
are essential here:
dropping either yields \NP-hardness.

\begin{theorem}[\appref{thm:t1:hardness}]
\label{thm:t1:hardness}
\fish{} is \NP-hard
even if
$\tau=1$,
each agent approves at most two resources,
and
\begin{inparaenum}[(i)]
    \item weights are in~$\{1,3\}$;
    \item all agents have weight one and assignments may be \wasteful.
\end{inparaenum}
\end{theorem}

For~(i), we give a polynomial-time many-one reduction from
\prob{Independent Set}, which remains \NP-complete
on cubic graphs~\cite{DBLP:books/fm/GareyJ79,MOHAR2001102}.
Given a cubic graph~$G=(V,E)$ and an integer~$k$,
the problem asks whether $G$ has~$k$ pairwise nonadjacent vertices.

\begin{proof}[Proof idea of \cref{thm:t1:hardness} (i)]
For each vertex, construct an agent of weight~$3$
and a resource with initial quantity~$3$.
For each edge, construct an agent of weight~$1$
approving the resources corresponding to its endpoints.
Add one critical resource,
approved by all vertex agents,
to which exactly~$|V|-k$ vertex agents must be assigned for sustainability.
To avoid envy,
the remaining~$k$ vertex agents
must be assigned to their corresponding resources.
These resources cannot receive any further agents.
An edge between two such vertices would hence have
both its approved resources fully harvested,
causing unavoidable envy.
Thus,
these~$k$ vertices form an independent set.
\end{proof}

For~(ii),
we give a polynomial-time many-one reduction from
\prob{Clique}~\cite{DBLP:books/fm/GareyJ79}, where
given a graph~$G=(V,E)$ and an integer~$k$,
the problem asks whether there is a set of at least~$k$
pairwise adjacent vertices.
Recall that we allow \wasteful{} assignments here.

\begin{proof}[Proof idea of \cref{thm:t1:hardness} (ii)]
Construct one agent and one resource for each vertex,
and one agent for each edge.
All agents have weight one.
Each edge agent approves the resources corresponding to its endpoints.
Each vertex agent approves its corresponding resource and one additional critical resource.
The critical resource must receive exactly~$k$ agents
for sustainability.
Since every vertex agent approves this resource,
every vertex agent must receive an approved resource to avoid envy. Since at most~$k$ vertex agents can be assigned to the critical resource,
all others must be assigned to their corresponding resources.
Hence, at most~$k$ vertex resources remain unassigned.
A final critical resource, which no agent approves, must receive exactly~$\binom{k}{2}$ agents.
These can only be edge agents who are non-envious,
i.e.,
whose two approved vertex resources are unassigned.
There are enough such agents only if exactly~$k$ vertex resources are unassigned and every pair of their corresponding vertices is an edge.
Hence, these vertices form a clique.
\end{proof}

\appendixproof{thm:t1:hardness}
{
\begin{construction}\label{constr:onelevel:twoapprovals}
Let~$I=(G=(V,E),k)$ be an instance of
\prob{Independent Set} on cubic graphs,
and let~$N=|V|$. We may assume that~$1\leq k<N$.
We construct an instance~$I'\coloneqq(\rSet,A,\tau)$
with~$\tau=1$ of \fish{} as follows.

Let~$A=A_V\cup A_E$, where
$A_V=\{a_v\mid v\in V\}$
and
$A_E=\{a_e\mid e\in E\}$.
Every agent from~$A_V$ has weight~$3$,
and every agent from~$A_E$ has weight~$1$.

We have the resources $\rSet=\{\r_v\mid v\in V\}\cup\{\r_0\}$.
For every~$v\in V$, resource~$\r_v$ has
$\iq_v=\ub_v=3$,
$\lb_v=0$,
$\rf_v(x)=x$.
Resource~$\r_0$ has
$\iq_0=\ub_0=3(N-k)$,
$\lb_0=0$,
$\rf_0(x)=x+3(N-k)$.

Every vertex agent~$a_v$ approves exactly~$\r_0$
and~$\r_v$.
Every edge agent~$a_e$ with~$e=\{u,v\}$
approves exactly~$\r_u$ and~$\r_v$.
\end{construction}

\begin{proof}[Proof of \cref{thm:t1:hardness}(i)]
The construction takes polynomial time,
and all numbers are bounded by~$3N$.
Every agent approves exactly two resources.
Moreover, $\r_0$ is the only critical resource:
without harvesting, its quantity exceeds its capacity
after renewal, whereas every resource~$\r_v$
retains quantity~$3$.
We prove that~$I$ is a \yes-instance
if and only if~$I'$ is a \yes-instance.

\RD{}
Let~$S\subseteq V$ be an independent set of size~$k$.
For every~$v\in S$, assign~$a_v$ to~$\r_v$.
For every~$v\in V\setminus S$, assign~$a_v$ to~$\r_0$.
For every edge~$e=\{u,v\}$, assign~$a_e$
to a resource corresponding to an endpoint outside~$S$.
Such an endpoint exists because~$S$ is independent.

Resource~$\r_0$ receives exactly~$N-k$ vertex agents,
so its total harvest is~$3(N-k)$.
Its quantity after renewal is therefore
$\rf_0\bigl(3(N-k)-3(N-k)\bigr)
=3(N-k)=\ub_0$.

For every~$v\in S$, resource~$\r_v$ receives only~$a_v$,
so its total harvest is~$3$ and its final quantity is~$0$.
For every~$v\in V\setminus S$,
resource~$\r_v$ receives only agents corresponding
to edges incident to~$v$.
Since~$G$ is cubic, there are at most three such agents.
Thus, its total harvest is at most~$3$,
and its final quantity lies between~$0$ and~$3$.
Hence, the assignment is sustainable.
Moreover, every agent is assigned to an approved resource,
so the assignment is envy-free.

\LD{}
Let~$\pi$ be a sustainable envy-free assignment.
No over-extraction implies
$h_0^\pi(1)\leq 3(N-k)$.
The quantity of~$\r_0$ after renewal is
$\rf_0\bigl(3(N-k)-h_0^\pi(1)\bigr)
=6(N-k)-h_0^\pi(1)$.
The capacity bound therefore implies
$6(N-k)-h_0^\pi(1)\leq 3(N-k)$,
and hence~$h_0^\pi(1)=3(N-k)$.
Since only vertex agents approve~$\r_0$
and each has weight~$3$,
exactly~$N-k$ vertex agents are assigned to~$\r_0$.

Since~$N-k>0$, resource~$\r_0$ receives at least one agent.
Every vertex agent approves~$\r_0$,
so every vertex agent must be assigned to an approved resource
to avoid envy.
Consequently, each of the remaining~$k$ vertex agents
is assigned to its corresponding resource.
Let
$S\coloneqq\{v\in V\mid \pi_1(a_v)=\r_v\}$.
Then~$|S|=k$.

We claim that~$S$ is independent.
Suppose there is an edge~$e=\{u,v\}$ with~$u,v\in S$.
Resources~$\r_u$ and~$\r_v$ each receive their corresponding
vertex agent of weight~$3$.
Since both resources have initial quantity~$3$,
neither can receive another agent.
Thus, edge agent~$a_e$ cannot be assigned
to either of her approved resources.
However, she values the assignment of~$a_u$ at one,
contradicting envy-freeness.
Therefore, $S$ is an independent set of size~$k$.
\end{proof}

\begin{proof}[Proof of \cref{thm:t1:hardness}(ii)]

Let~$I=(G=(V,E),k)$ be an instance of \prob{Clique},
and let~$N=|V|$.
We may assume that~$2\leq k\leq N$,
since the remaining cases can be decided in polynomial time.

\begin{construction}
We construct an instance~$I'=(\rSet,A,\tau)$ of \fish{}
with~$\tau=1$.

The set of agents is~$A=A_V\cup A_E$, where
$A_V=\{a_v\mid v\in V\}$ and
$A_E=\{b_e\mid e\in E\}$.
Thus, there is one \emph{vertex agent} for every vertex
and one \emph{edge agent} for every edge.
Every agent has weight one.

For every vertex~$v\in V$, introduce a resource
$\r_v=(N,0,N,\rf_v)$, where~$\rf_v(x)=x$.
We additionally introduce two special resources
$\r_0$ and~$\r^\star$, defined by
$\r_0=(k,0,k,\rf_0)$,
with
$\rf_0(x)=x+k$,
and
$\r^\star=
\left(\binom{k}{2},0,\binom{k}{2},\rf_\star\right)$,
with
$\rf_\star(x)=x+\binom{k}{2}$.

Every vertex agent~$a_v$ approves~$\r_0$ and~$\r_v$.
Every edge agent~$b_{\{u,v\}}$ approves~$\r_u$ and~$\r_v$.
This completes the construction.
\end{construction}

The construction takes polynomial time,
and all numerical parameters are polynomially bounded.
Every agent approves exactly two resources and only~$\r_0$ and~$\r^\star$ are critical.

We now prove that~$I$ is a yes-instance of
\prob{Clique} if and only if~$I'$ admits
a sustainable envy-free assignment.

\smallskip
\RD{}
Suppose that~$G$ has a clique~$K$ of size~$k$.
Assign~$a_v$ to~$\r_0$ for every~$v\in K$,
and assign every remaining vertex agent
to its corresponding resource.
For every edge with both endpoints in~$K$,
assign its agent to~$\r^\star$.
Assign every remaining edge agent to a resource
corresponding to an endpoint outside~$K$.
Such an endpoint exists because the edge
does not have both endpoints in~$K$.

Resource~$\r_0$ receives exactly~$k$ agents.
Since~$K$ is a clique, exactly~$\binom{k}{2}$
agents are assigned to~$\r^\star$.
Their quantities after renewal are therefore
$\q_0^\pi(2)=\rf_0(0)=k$ and
$\q_\star^\pi(2)=\rf_\star(0)=\binom{k}{2}$.
Every resource~$\r_v$ with~$v\notin K$
is assigned its vertex agent and at most~$N-1$ edge agents,
while all resources corresponding to vertices in~$K$
remain unassigned.
Hence, no resource is over-extracted,
all lower and upper bounds are respected,
and the assignment is sustainable.

Every agent assigned to a resource other than~$\r^\star$
receives an approved resource.
An edge agent assigned to~$\r^\star$ has zero value,
but both of its approved resources are unassigned.
Such an agent therefore values every other agent's
assignment at zero.
Thus, the assignment is envy-free.

\smallskip
\LD{}
Suppose that~$\pi$ is a sustainable envy-free assignment.
No over-extraction implies~$h_0^\pi(1)\leq k$.
Moreover,
$\q_0^\pi(2)
=\rf_0(k-h_0^\pi(1))
=2k-h_0^\pi(1)$.
The capacity bound requires~$\q_0^\pi(2)\leq k$,
and hence~$h_0^\pi(1)=k$.
Analogously, sustainability forces
$h_\star^\pi(1)=\binom{k}{2}$.

Since~$k>0$, at least one agent is assigned to~$\r_0$.
Every vertex agent approves~$\r_0$,
so every vertex agent must receive an approved resource
to avoid envy.
At most~$k$ vertex agents can be assigned to~$\r_0$.
Consequently, at least~$N-k$ vertex agents are assigned
to their corresponding resources.

Let~$U$ be the set of vertices whose resources receive no agent.
Then~$|U|\leq k$.
Exactly~$\binom{k}{2}$ agents are assigned to~$\r^\star$.
These must be edge agents,
since every vertex agent receives an approved resource.
Moreover, both endpoints of each such edge must belong to~$U$.
Otherwise, its corresponding agent would envy an agent assigned
to an endpoint resource.

There are at most~$\binom{|U|}{2}\leq\binom{k}{2}$
edge agents with both endpoints in~$U$.
Since~$\r^\star$ receives exactly~$\binom{k}{2}$ agents
and~$k\geq2$, we must have~$|U|=k$,
and every pair of vertices in~$U$ must be an edge.
Thus,~$U$ is a clique of size~$k$.
\end{proof}
}

Both hardness reductions use agents approving two resources. In contrast, for \nonwasteful{} assignments with $\tau=1$, \fish{} is polynomial-time solvable if every agent approves exactly one resource. Indeed, every critical resource must be assigned at least one agent, and envy-freeness then forces all agents approving that resource to be assigned to it. All uncritical resources can remain unused, so it suffices to check whether the resulting assignment is sustainable.

\section{Few Resource Usages: Unary Versus Binary Encoding}
\begin{figure}[h!]
 \centering
 \begin{tikzpicture}
  \tikzpreamble{}
  \def\xr{1.75}
  \def\yr{0.5}

  \clip (-1.275*\xr,-1.6*\yr) rectangle (4.7875*\xr,3.2*\yr);

  \foreach \x in {0,1,2,3}{
		\foreach \y in {0,1,2,3}{
			\node (x\x\y) at (\x*\xr,\y*\yr)[inner sep=0pt]{};
		}
	}
	\node at ($(x00)!0.5!(x10)$) [anchor=north,yshift=-1*\yr pt]{$1$};
	\node at ($(x10)!0.5!(x20)$) [anchor=north,yshift=-1*\yr pt]{$2$};
	\node at ($(x20)!0.5!(x30)$) [anchor=north,yshift=-1*\yr pt]{unbd.};
	\node at ($(x00)!0.5!(x01)$) [anchor=east,xshift=-1*\yr pt]{$1$};
	\node at ($(x01)!0.5!(x02)$) [anchor=east,xshift=-1*\yr pt]{$2$};
	\node at ($(x02)!0.5!(x03)$) [anchor=east,xshift=-1*\yr pt]{unbd.};

	\draw[decorate,decoration={brace,amplitude=5pt},color=gray]
  (3*\xr,-0.8*\yr) -- (0*\xr,-0.8*\yr)
  node[midway,below=8*\yr pt,color=black] {$\tau$};

  \draw[decorate,decoration={brace,amplitude=5pt},color=gray]
  (-0.7*\xr,-0*\yr) -- (-0.7*\xr,3*\yr)
  node[midway,left=4*\xr pt,color=black] {$m$};

	\newcommand{\fullUT}[3]{%
		\pgfmathtruncatemacro{\im}{#1-1}%
		\pgfmathtruncatemacro{\jm}{#2-1}%
		\draw[gray]
			(x\im\jm.center) --
			(x#1#2.center) --
			(x\im#2.center) --
			cycle;
		\fill[#3]
			(x\im\jm.center) --
			(x#1#2.center) --
			(x\im#2.center) --
			cycle;
	}

	\newcommand{\fullLT}[3]{%
		\pgfmathtruncatemacro{\im}{#1-1}%
		\pgfmathtruncatemacro{\jm}{#2-1}%
		\draw[gray]
			(x\im\jm.center) --
			(x#1#2.center) --
			(x#1\jm.center) --
			cycle;
		\fill[#3]
			(x\im\jm.center) --
			(x#1#2.center) --
			(x#1\jm.center) --
			cycle;
	}

	\fullUT{1}{1}{colFPT}
	\fullLT{1}{1}{colFPT}

	\foreach\x/\y in {3/1,1/3,3/2,2/3,3/3}{
		\fullUT{\x}{\y}{colPNPH}
		\fullLT{\x}{\y}{colXPWH}
	}
	\foreach\x/\y in {2/1,1/2,2/2}{
		\fullUT{\x}{\y}{colPNPH}
		\fullLT{\x}{\y}{colFPT}
	}

	\draw[color=gray] (3.275*\xr,2*\yr) rectangle (4.375*\xr,3.2*\yr);
	\draw (3.275*\xr,2*\yr) -- node[midway,above=-2pt,sloped,font=\footnotesize]{binary\phantom{xxxxx}}node[midway,below,sloped,font=\footnotesize]{\phantom{xxxx}unary}(4.375*\xr,3.2*\yr);

	\node at (3.275*\xr,1.5*\yr)[colFPT,draw,label=0:{\classP}]{};
	\node at (3.275*\xr,0.75*\yr)[colXPWH,draw,label=0:{\XP, \W{1}-hard}]{};
	\node at (3.275*\xr,0.0*\yr)[colPNPH,draw,label=0:{(para-)\NP-hard}]{};

 \end{tikzpicture}
 \caption{Overview of our results for~$m+\tau$, with values 1, 2, and unbounded for each dimension.
If a dimension is unbounded, entries give the corresponding parameterized complexity.}
 \label{fig:mtau}
\end{figure}
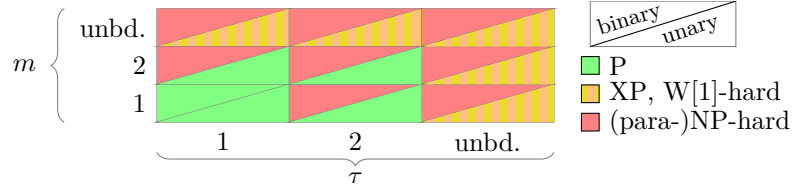
\noindent
We next investigate the upper bound~$m\cdot\tau$ on the number of resource usages,
i.e., the number of pairs~$(\r_i,t)$ such that~$\pi_t^{-1}(\r_i)\neq \emptyset$.
How many resources and periods are needed for hardness?
Our results,
summarized in~\cref{fig:mtau},
give quite different answers to this question depending on the encoding.
For a single period and a single resource,
the problem is trivially polynomial-time solvable:
if the resource is not urgent,
then no agent must be assigned;
otherwise,
to ensure envy-freeness,
all agents approving the single resource must be assigned.
In contrast, we show that
already two resources and one period
or one resource and two periods
suffice for \NP-hardness.
For unary encoding, however,
the problem is polynomial-time solvable for every constant number of periods and resources.

\begin{theorem}[\appref{thm:mtau:weak}]
  \label{thm:mtau:weak}
	\fish{} is \NP-hard,
	even with only
	\begin{inparaenum}[(i)]
		\item one period and two resources;
		\item one resource and two periods.
	\end{inparaenum}
\end{theorem}

For both statements,
we give a reduction from the
\NP-hard \prob{Partition} problem,
where, given a multiset~$X=\{x_1,\dots,x_N\}$ of numbers with~$T=\sum_{i=1}^N x_i$, the question is whether there is a subset~$J\subseteq \set{N}$ such that~$\sum_{i\in J} x_i=T/2$.
We can assume that all numbers are positive and even (multiply all by two),
including~$T$.

\begin{proof}[Proof idea of \cref{thm:mtau:weak}]
 For each number~$x_i$ in the input instance,
 construct an agent~$a_i$ of weight~$x_i$.
 For (i),
 construct two resources,
 each of which must be harvested by an amount of exactly~$T/2$.
 For (ii),
 construct a resource that must be harvested by an amount of exactly~$T/2$ in each of the two periods.
 As all agents approve all resources,
 envy-freeness
 and the total-harvest equality ensure a 2-partition with each part of sum~$T/2$
 in either case.
\end{proof}

\appendixproof{thm:mtau:weak}
{
\begin{construction}
\label{constr:1tau} Let~$I=(X)$ be an instance of \prob{Partition}, where $X=\{x_1,\ldots,x_N\}$ and $T=\sum_{x_i \in X} x_i$.
Construct an instance~$I'=(\rSet,A,\tau)$ as follows.

(i) Let~$\tau=1$ and $\rSet=\{\r_1,\r_2\}$.
Let~$A=\{a_1,\dots,a_N\}$, where~$a_i$ has weight~$x_i$.
All agents approve both~$\r_1$ and~$\r_2$.
For both resources~$\r_j$,
we set~$\iq_j=\ub_j=T/2$, $\lb_j=1$, and $\rf_j(z)=z+1$.

(ii) Let~$\tau=2$ and~$\rSet=\{\r\}$.
	Let~$A=\{a_1,\dots,a_N\}$ where~$a_i$ has weight~$x_i$.
	All agents approve~$\r$.
	Resource~$\r$ has~$\iq=\ub=T/2$,
	$\lb=1$ and~$\rf(z)=z+T/2$.
\end{construction}

\begin{proof}[Proof of \cref{thm:mtau:weak}]
(i) \RD{} Let~$J\subseteq \set{N}$ be a solution for~$I$. Assign~$a_i$ to~$\r_1$ in period~$1$ for all~$i\in J$, and assign~$a_i$ to~$\r_2$ in period~$1$ for all~$i\in \set{N}\setminus J$. The total weight assigned to each resource is~$T/2$. Hence, for both resources~$\r_j$, the quantity after the assignments is~$T/2-T/2=0$, and after renewal it is~$\rf_j(0)=1=\lb_j$. Thus, the assignment is sustainable. Moreover, all agents are assigned to an approved resource, and hence no agent envies another. Therefore, this is a solution to~$I'$.

	\LD{} Let~$\pi=(\pi_1)$ be an envy-free sustainable assignment.
	For~$j\in\{1,2\}$, let~$A_j$ be the set of agents assigned to~$\r_j$ in period~$1$. First note that~$A_j$ is nonempty for both~$j\in\{1,2\}$. Indeed, if no agent is assigned to~$\r_j$, then the quantity remains~$\iq_j=\ub_j=T/2$, and after renewal it becomes~$\rf_j(\ub_j)=\ub_j+1>\ub_j$, contradicting the capacity constraint. Since both resources receive at least one assigned agent and all agents approve both resources, envy-freeness implies that every agent is assigned. Thus, $A_1\cup A_2=A$. Let $J=\{i\in\set{N}\mid a_i\in A_1\}$. We claim that~$J$ is a solution for~$I$. Suppose first that~$\sum_{i\in J}x_i>T/2$. Then, the total weight assigned to~$\r_1$ exceeds~$\iq_1=T/2$, contradicting no over-extraction. The same argument applied to~$\r_2$ shows that the agents assigned to~$\r_2$ have total weight at most~$T/2$. Since every agent is assigned and the total weight of all agents is~$T$, it follows that $\sum_{i\in J}x_i=T/2$. Therefore, $J$ is a solution for~$I$.

(ii)
\RD{}
	 Let~$J\subseteq \set{N}$ be a solution for~$I$. Assign~$a_i$ to~$\r$ in period~$1$ for all~$i\in J$, and assign~$a_i$ to~$\r$ in period~$2$ for all~$i\in \set{N}\setminus J$. The total weight assigned to the resource in each period is~$T/2$.
	 Hence, after the assignments in each period, the quantity is~$T/2-T/2=0$,
	 and after renewal it is~$\rf(0)=T/2=\ub$. Thus, the assignment is sustainable. Moreover, every agent is assigned to the approved resource, and hence no agent envies another. Therefore, this is a solution to~$I'$.

	\LD{}
	Let~$\pi=(\pi_1,\pi_2)$ be an envy-free sustainable assignment, and let~$A_1$ be the set of agents assigned to~$\r$ in period~$1$. Let $J=\{i\in\set{N}\mid a_i\in A_1\}$. We claim that~$J$ is a solution for~$I$. Suppose first that~$\sum_{i\in J}x_i<T/2$. Then, after renewal following period~$1$,
	the quantity is $T/2-\sum_{i\in J}x_i+T/2>T/2=\ub$,
	contradicting the capacity constraint.
	On the other hand,
	if~$\sum_{i\in J}x_i>T/2$,
	then
	$T/2-\sum_{i\in J}x_i<0$,
	contradicting no over-extraction.
	Hence, $\sum_{i\in J}x_i=T/2$, and therefore~$J$ is a solution for~$I$.
\end{proof}
}

\begin{theorem}[\appref{thm:whardmtau}]
	\label{thm:whardmtau}
	For unary encoding,
	\fish{} is \W{1}-hard
	\begin{inparaenum}[(i)]
		\item \wpb{}~$m$,
		even if~$\tau=1$;
		\item \wpb{}~$\tau$,
		even if~$m=1$.
	\end{inparaenum}
	In particular,
	it is \NP-hard even if~$m=1$ or~$\tau=1$.
\end{theorem}

Each of the two reductions is from \prob{Exact Unary Bin Packing},
where,
given a multiset~$X=\{x_1,\dots,x_N\}$ of integers,
a bin size~$B\in \N$,
and a number~$k\in\N$ of bins,
all encoded in unary
and
such that~$\sum_{x_i\in X} x_i=k\cdot B$,
the question is whether there is a $k$-partition $(X_1,\dots,X_k)$
of~$X$ such that~$\sum_{x_i\in X_j} x_i = B$ for every~$j\in\set{k}$.
\prob{Exact Unary Bin Packing} is \NP-hard and \W{1}-hard \wpb{}~$k$~\cite{JansenKMS13}.

\begin{proof}[Proof idea of \cref{thm:whardmtau}]
 For each number~$x_i$ in the input instance,
 construct an agent~$a_i$ of weight~$x_i$.
 Now,
 for (i),
 construct $k$ resources,
 each of which must be harvested by an amount of exactly~$B$.
 For (ii),
 construct one resource that must be harvested by an amount of exactly~$B$ in each of the $k$ periods.
 Since all agents approve all resources,
 envy-freeness ensures an actual $k$-partition with each part summing to~$B$ in either case.
\end{proof}

\appendixproof{thm:whardmtau}
{
\begin{construction}
  \label{constr:whardmtau}
	(i)
	For each~$i\in\set{k}$,
	a resource~$\r_i$ represents a bin
	with~$\iq_i=B$ and~$\lb_i=1$.
	For each number~$x_i\in X$,
	there is an agent~$a_i$ with weight~$x_i$.
	Each agent approves each resource.
	Exactly one resource has~$\ub_i=\iq_i$, the others are uncritical
	with~$\ub_j=\iq_j+1$.
	Each renewal function is~$\rf(x)=x+1$.

	(ii)
	There is one resource~$\r_1$ with~$\iq_1=B$,
	$\lb_1=1$,
	$\ub_1=B$,
	and~$\rf_1(x)=x+B$.
	Moreover,
	there are~$N$ agents,
	where agent~$a_i$ has weight~$x_i$.
	All agents approve~$\r_1$.
	We have~$\tau=k$ periods.
\end{construction}

\begin{proof}[Proof of \cref{thm:whardmtau}]
 Let~$I'$ be the instance of \fish{}
 obtained from an instance~$I$ of \prob{Exact Unary Bin Packing}
 using \cref{constr:whardmtau}(i) or (ii).
 We claim that~$I$ is a \yes-instance if and only if~$I'$ is a \yes-instance.

 (i):
 A partition~$(X_1,\dots,X_k)$ induces an assignment~$\pi=(\pi_1)$ and vice versa through the following:
 $x_i\in X_j$ if and only if~$\pi_1(a_i)=\r_j$.
 For the reverse direction,
 note that there is at least one agent assigned,
 and due to envy-freeness,
 this implies all agents are assigned.
Moreover,
 terminal viability gives~$\rf_j(B-h_j^\pi(1))=B-h_j^\pi(1)+1\geq\lb_j=1$,
 that is,
 $h_j^\pi(1)\leq B$ for every~$j\in\set{k}$.
 Since the total weight of all agents is~$k\cdot B$,
 it follows that~$h_j^\pi(1)=B$ for every~$j\in\set{k}$.
 Now we have~$\sum_{x_i\in X_j} x_i = h_j^\pi(1)$ for every~$j\in\set{k}$,
 and hence,
 one side equals~$B$ if and only if the other does.

 (ii):
 A partition~$(X_1,\dots,X_k)$ induces an assignment~$\pi$ and vice versa through the following:
 $x_i\in X_j$ if and only if~$\pi_j(a_i)=\r_1$.
 For the reverse direction,
 note that due to the upper capacity and no over-extraction,
 $\r_1$ must be harvested by a total weight of exactly~$B$ in every period,
 leading to total harvest of exactly~$k\cdot B$.
 By envy-freeness,
 this implies every agent is assigned at least once,
and thus~$k\cdot B=\sum_{j=1}^k h_1^\pi(j)\geq \sum_{i=1}^N x_i = k\cdot B$.
 Since all~$x_i$ are positive,
 every agent is assigned exactly once.
 Hence,
 $X_j\ceq\{x_i\mid \pi_j(a_i)=\r_1\}$ for~$j\in\set{k}$ defines a $k$-partition of~$X$
 with~$\sum_{x_i\in X_j}x_i=h_1^\pi(j)=B$.
\end{proof}
}
Complementing this hardness result, the problem is in \XP{} for unary weights.

\begin{theorem}[\appref{thm:mtau:unary:XP}]
 \label{thm:mtau:unary:XP}
 For unary weights,
 \fish{} is in \XP{} \wpb{}~$m+\tau$.
\end{theorem}

\begin{proof}[Proof idea of \cref{thm:mtau:unary:XP}]
 Consider a sustainable envy-free \xassignment[$\tau$]~$\pi$.
 On the one hand,
 $\pi$ induces a \emph{harvest matrix}
 $H=(h_{i,t})_{i\in\set{m},t\in\set{\tau}} \in \{0,\dots,W\}^{m\times \tau}$,
 where $W=\sum_{a\in A} w_a$ and $h_{i,t}\in\{0,\dots,W\}$ corresponds
 to the weight assigned to resource~$\r_i$ in period~$t$.
 On the other hand,
 $\pi$ induces for each agent~$a$
 a sequence $\sigma=(\sigma_1,\dots,\sigma_\tau)$
 with~$\sigma_t=\pi_t(a)\in(\rSet\cup\{\emptyset\})$.
 We check each of the at most~$(W+1)^{m\tau}$
 possible harvest matrices.
 Since the weights are encoded in unary,
 this is in~\XP{} with respect to~$m+\tau$.
For each guess~$H$,
we can compute the resource quantities of every period
and verify sustainability in polynomial time.
 It remains to find an assignment of sequences to agents
that induces~$H$.
 We check each of the at most~$2^{(m+1)^\tau}$ possible sets of sequences.
 For each such set~$S^*$ of sequences,
 we can now use dynamic programming,
 iterating over the agents,
 to check whether there is an assignment of sequences to agents resulting in~$H$.
 To ensure envy-freeness,
 we proceed as follows:
 First,
 we require that every schedule from the set is assigned to at least one agent.
 Second,
 we restrict each agent to \nonwasteful{} sequences that maximize her utility over all of~$S^*$.
 Combining these two ingredients ensures envy-freeness.
 Thus,
 correctness of our algorithm follows.
 Finally,
 the algorithm runs in~\XP{} time regarding~$m+\tau$,
 since iterating over all harvest matrices~$H$
 and sets~$S^*$ of sequences is in~\XP{} time regarding~$m+\tau$,
 and the dynamic program iterates over all agents and checks entries for harvest matrices bounded entrywise by~$H$
 and subsets of~$S^*$.
\end{proof}

\appendixproof{thm:mtau:unary:XP}
{%
Let~$W=\sum_{a\in A} w_a$.
Define a harvest matrix as~$H=(h_{i,t})_{i\in\set{m},t\in\set{\tau}}$ with~$h_{i,t}\in\{0,\dots,W\}$ corresponding to the weight assigned to resource~$\r_i$ in period~$t$.
There are~$(W+1)^{m\cdot\tau}$ such matrices.
For each candidate matrix~$H^*$, we start from the resources' initial quantities and proceed period by period, subtracting the harvest specified by~$H^*$
and applying the corresponding renewal function.
This allows us to check all sustainability constraints in polynomial time.
For each agent~$a$, write~$\sigma=(\sigma_1,\dots,\sigma_\tau)$
for her assignments over all periods,
where~$\sigma_t\in\rSet\cup\{\emptyset\}$.
Let~$\calS=(\rSet\cup\{\emptyset\})^\tau$
be the set of all~$(m+1)^\tau$ such sequences.
For each agent~$a$, let~$\calS_a$ be the set of
approval-restricted sequences, namely those~$\sigma\in\calS$
such that~$\sigma_t\in\rSet\implies\val_a(\sigma_t)=1$
for every~$t\in\set{\tau}$.
We guess a subset~$S^*\subseteq\calS$ from at most
$2^{(m+1)^\tau}$ possibilities and require every sequence
in~$S^*$ to be assigned to at least one agent.
To ensure envy-freeness, we restrict each agent to approval-restricted
sequences that maximize her utility over all of~$S^*$.
Define~$S_a^*$ to be the set of sequences~$\sigma\in S^*\cap\calS_a$
such that~$\sum_{t=1}^\tau \val_a(\sigma_t)\geq \sum_{t=1}^\tau \val_a(\sigma_t')$ for every sequence~$\sigma'=(\sigma_1',\dots,\sigma_\tau')\in S^*$. If~$S_a^*=\emptyset$ for any agent~$a$, we discard~$S^*$.
Moreover, a sequence~$\sigma$ for agent~$a_i$ induces a harvest
matrix~$H^{i,\sigma}$ with~$H^{i,\sigma}_{\ell,t}=w_{a_i}$
if~$\sigma_t=\r_\ell$, and zero otherwise.

We define a dynamic program that iterates over the agents~$a_1,\dots,a_n$,
tests all sequences in~$S_{a_i}^*$ for each agent~$a_i$,
and updates the aggregate harvest matrix and the set of used sequences.
So we have~$D[i,H,S]\in\{\true,\false\}$, where~$i$ is the number of processed agents,
$H$ is their aggregate harvest matrix, and~$S\subseteq S^*$ is the set of sequences used.
Initially,
we set $D[0,\mathbf{0},\emptyset]=\true$.
For the transition step,
we set~$D[i,H,S]=\true$ if there is~$D[i-1,H',S']=\true$ and
a sequence~$\sigma\in S^*_{a_i}$ such that~$H=H'+H^{i,\sigma}$
with~$H\leq H^*$,
defined componentwise as
$H_{j,t}\leq H_{j,t}^*$ for all~$j,t$,
and~$S=S'\cup\{\sigma\}$.
Such an update requires checking at most~$|S_{a_i}^*|$
possible sequences for agent~$a_i$ and all~$m\cdot\tau$ entries of the resulting harvest matrix for each sequence.
Finally, if~$D[n,H^*,S^*]=\true$, we return $\yes$ for~$I$.
\cref{alg:xp-affine-fish} provides full details.

\begin{algorithm}[t]
	\caption{Our algorithm in the proof of \cref{thm:mtau:unary:XP}.
	Input is an instance~$I$ with an $n$-sized agent set~$A$, $m$-sized resource set~$\rSet$, and~$\tau$ periods.}
	\label{alg:xp-affine-fish}

  Set $W \gets \sum_{a\in A} w_a$ and
	$\calS \gets (\rSet\cup\{\emptyset\})^\tau$\;

	\ForEach{$H^* \in \{0,\ldots,W\}^{m\tau}$}{
			compute the resource quantities over all periods induced by~$H^*$\;
			\If{$H^*$ is sustainable}{
				\ForEach{$S^* \subseteq \calS$}{
						\ForEach{$a\in A$}{
								$S_a^* \gets
								\calS_a\cap \arg\max_{\sigma\in S^*} \sum_{t=1}^\tau \val_a(\sigma_t)$\;
						}
						\If{$\forall a\in A: S_a^*\neq \emptyset$}{
							initialize $D[0,(0)_{m\times\tau},\emptyset]\gets\true$ and
							all other DP entries to $\false$\;

							\For{$i\gets 1$ \KwTo $n$}{
									\ForEach{$(H',S')$ with $D[i-1,H',S']=\true$ \textbf{and each} $\sigma\in S_{a_i}^*$}{
													Set $H \gets H' + H^{i,\sigma}$ and
													$S \gets S'\cup\{\sigma\}$\;
													\If{$H \leq H^*$}
													{
															$D[i,H,S] \gets \true$\;
													}
									}
							}

							\If{$D[n,H^*,S^*]=\true$}{
									\Return{\yes}\;
							}
					 }
				}
			}
	}
	\Return{\no}\;
\end{algorithm}

\begin{lemma}
 \label{lem:mtau:unary:XP}
 $D[i,H,S]=\true$ if and only if
 there exists an assignment of sequences to agents $a_1,\dots,a_i$ such that
(i) every $a_j$ receives a sequence in $S_{a_j}^*$,
(ii) their aggregate harvest matrix equals $H\leq H^*$, and
(iii) $S$ is exactly the set of sequences used by these $i$ agents.
\end{lemma}

\begin{proof}
 We prove the statement via induction on~$i$.
 For~$i=0$,
 the statement is trivially correct.
 Now let~$i>0$.

 \RD{}
 Let~$D[i,H,S]$ be set to true from~$D[i-1,H',S']$,
 which hence is true itself,
 and let~$\sigma\in S_{a_i}^*$ be the corresponding sequence.
 By induction,
 we have that
 there exists an assignment of sequences to agents $a_1,\dots,a_{i-1}$ such that
 (i) every $a_j$ receives a sequence in $S_{a_j}^*$,
 (ii) their aggregate harvest matrix equals $H'\leq H^*$, and
 (iii) $S'$ is exactly the set of sequences used by these $i-1$ agents.
 With~$H=H'+H^{i,\sigma}$ and~$S=S'\cup\{\sigma\}$,
 we have that
 there exists an assignment of sequences to agents $a_1,\dots,a_i$ such that
 (i) every $a_j$ receives a sequence in $S_{a_j}^*$,
 (ii) their aggregate harvest matrix equals $H\leq H^*$, and
 (iii) $S$ is exactly the set of sequences used by these $i$ agents.

 \LD{}
 Let~$\pi$ be an assignment of sequences to agents $a_1,\dots,a_i$ such that
 (i) every $a_j$ receives a sequence in $S_{a_j}^*$,
 (ii) their aggregate harvest matrix equals $H\leq H^*$, and
 (iii) $S$ is exactly the set of sequences used by these $i$ agents.
 Let~$\sigma$ be~$a_i$'s sequence,
 and define~$H'=H-H^{i,\sigma}$ and~$S'$ to be the sequences used by~$a_1,\dots,a_{i-1}$.
 By induction,
 we then have $D[i-1,H',S']=\true$.
 Since the algorithm checks all possibilities exhaustively,
 it follows that $D[i,H,S]$ is set to true from $D[i-1,H',S']$.
\end{proof}

\begin{proof}[Proof of \cref{thm:mtau:unary:XP}]
 We prove that~\cref{alg:xp-affine-fish}
 is sound,
 complete,
 and runs in \XP{} time.

 \emph{Soundness}:
 Suppose that \cref{alg:xp-affine-fish} returns~$\yes$ on $D[n,H^*,S^*]=\true$.
 Due to~\cref{lem:mtau:unary:XP},
 there exists an assignment of sequences to agents $a_1,\dots,a_n$ such that
 (i) every $a_j$ receives a sequence in $S_{a_j}^*$,
 (ii) their aggregate harvest matrix equals $H^*$, and
 (iii) $S^*$ is exactly the set of sequences used by these $n$ agents.
Since only sustainable harvest matrices are passed to the dynamic program, $H^*$ is sustainable.
Since exactly the sequences in~$S^*$ are used
and each agent~$a_j$ receives a utility-maximizing sequence
from~$S_{a_j}^*$,
the assignment is envy-free.

 \emph{Completeness}:
 Let~$\pi$ be a solution,
 i.e.,
 a sustainable envy-free \xassignment[$\tau$].
The assignment~$\pi$ induces a harvest matrix~$H^\pi$
with~$h_{j,t}^\pi=h_j^\pi(t)$.
 In addition,
 it defines each agent~$a_i$'s sequence~$\sigma^{\pi,i}=(\sigma_1^{\pi,i},\dots,\sigma_\tau^{\pi,i})$ with
 $\sigma_t^{\pi,i}=\pi_t(a_i)$. Let~$S^\pi$ denote the set of all these sequences.
Since~\cref{alg:xp-affine-fish} iterates over all possible harvest matrices,
it will eventually consider~$H^\pi$.
 Since~$\pi$ is a solution,
 $H^\pi$ passes the subsequent sustainability check.
 Next,
 since all subsets of~$\calS$ are considered,
the algorithm eventually guesses~$S^*=S^\pi$.
For this guess, since~$\pi$ is envy-free,
$\sigma^{\pi,i}\in S_{a_i}^*$ for every agent~$a_i$.
 The correctness now follows from \cref{lem:mtau:unary:XP}.
 Thus,
 the algorithm returns \yes{} for~$D[n,H^\pi,S^\pi]$ being true.
 \emph{Running time}:
 Initializing and computing each DP-table takes
 $n\cdot (W+1)^{m\tau}\cdot (m+1)^{2\tau}\cdot 2^{(m+1)^\tau} \poly(|I|)$ time.
 Computing all~$S_a^*$
 takes up to~$n\cdot 2^{(m+1)^\tau}\cdot \tau\cdot \poly(|I|)$ time.
 Checking whether~$H^*$ is sustainable can be done in~$\poly(|I|)$ time.
 Iterating over~$2^{(m+1)^\tau}$ choices for~$S^*$
 and $(W+1)^{m\cdot \tau}$ for~$H^*$,
 we arrive at a final running time in
 $\poly(|I|)\cdot (W+1)^{2 m\cdot \tau} \cdot 2^{2(m+1)^\tau} \cdot (m+1)^{2\tau}$,
 which is \XP{} wrt.~$m+\tau$ since~$W\leq |I|$ by the unary encoding.
\end{proof}
}%

\section{The Agents: Unit Weights or Few Agents Suffice for Hardness}
We next consider restrictions on the agents.
We show \NP-hardness both for unit weights, for only one agent under binary encoding and for three agents under unary encoding.

\begin{theorem}[\appref{thm:unaryweight}]
	\label{thm:unaryweight}
	\fish{} is \NP-hard,
	even for unary encoding and every agent having weight one.
\end{theorem}

We give a polynomial-time many-one reduction from \prob{\NonWasteful{} Envy-Free Item Allocation},
where,
given a set~$X=\{x_1,\dots,x_M\}$ of items,
a set~$A=\{a_1,\dots,a_N\}$ where each agent~$a_i$ has a binary valuation~$v_i:X\to \{0,1\}$,
one asks whether there is a partition~$(X_1,\dots,X_N)$ of~$X$
such that
for every~$i\in\set{N}$ we have~$v_i(x_j)=1$ for all~$x_j\in X_i$
and
for every distinct~$i,j\in\set{N}$,
$\sum_{x_p\in X_i} v_i(x_p) \geq \sum_{x_p\in X_j} v_i(x_p)$.
Hosseini~\etal~\cite{HosseiniSVWX20} give an alternative proof that the problem without the \nonwasteful{} requirement is \NP-hard~\cite{AzizGMW15}.
In their hardness reduction,
the forward direction constructs an allocation that assigns every item only to an agent who approves it,
while the reverse direction applies to arbitrary envy-free allocations.
Hence,
the same reduction shows that \prob{\NonWasteful{} Envy-Free Item Allocation} is \NP-hard.

\begin{proof}[Proof idea of \cref{thm:unaryweight}]
 For each item in the input instance, we introduce a resource and a period.
 Each input agent forms an agent in our target instance with weight one approving exactly the resources corresponding to her approved items.
 Now, each resource~$\r_j$ becomes available and urgent with a quantity of one exactly in period~$j$. Sustainability therefore requires exactly one agent to harvest it.
 After being harvested in period~$j$ by exactly one agent,
 its renewal rate is chosen sufficiently low such that it cannot be harvested again.
 Thus,
 any sustainable envy-free \xassignment[$M$]~$\pi$ induces an envy-free partition of resources,
 and hence of the items.
\end{proof}

\appendixproof{thm:unaryweight}
{%
\begin{construction}
	\label{constr:unaryweight}
		Let~$I=(X,A)$ be an instance of \prob{\NonWasteful{} Envy-Free Item Allocation}, where $X=\{x_1,\dots,x_M\}$ is the set of items and $A=\{a_1,\dots,a_N\}$ is the set of agents with binary and additive valuations.
		We construct an instance $I'=(\rSet,A',\tau)$ of \fish{} as follows. Let~$A'=A$, with~$w_a=1$ for all $a \in A'$ and~$\tau=M$.
		Let~$\widehat{M}\coloneqq 1/(M+1)$.
		For every $j\in\set{M}$,
		we add one resource $\r_j$ with
		\[
		\ub_j=1,\quad
		\lb_j=\widehat{M},\quad
		\iq_j=1-(j-1)\cdot\widehat{M},\quad
		\rf_j(z)=z+\widehat{M}.
		\]
		Thus, all resources have the same affine renewal function with
		$\alpha_j=1$ and $\beta_j=\widehat{M}$. Each agent~$a$ approves~$\r_j$ if and only if she approves~$x_j$ in~$I$.
	\end{construction}

\begin{proof}[Proof of \cref{thm:unaryweight}]
 Let~$I'$ be the instance of \fish{}
 obtained from an instance~$I=(X,A)$ of  \prob{\NonWasteful{} Envy-Free Item Allocation}
 using~\cref{constr:unaryweight}.
 Since all occurring numerical values~$a/b$ have $a,b\in O(M)$,
 they are polynomially bounded even under unary encoding.
 We prove that~$I$ is a \yes-instance if and only if~$I'$ is a \yes-instance.

 \RD{}
 Let~$\calX=(X_1,\dots,X_N)$ be a solution to~$I$.
 We claim that~$\pi=(\pi_1,\dots,\pi_M)$ with
 $\pi_j(a_i)=\r_j$ if and only if~$x_j\in X_i$
 is a solution to~$I'$.
 First,
 note that since~$\calX$ is an envy-free partition of~$X$,
 $\pi$ is envy-free and to each resource exactly one agent is assigned.
 It remains to show that~$\pi$ is sustainable.
For~$1\leq t\leq j$, we have
$\q_j^\pi(t)=1-(j-t)\widehat{M}\in[\lb_j,\ub_j]$.
In particular, $\q_j^\pi(j)=1$.
 Since agent~$a_i$ has weight one,
 there is no over-extraction.
 Since~$a_i$ harvests~$r_j$ in period~$j$,
 we have for period~$j+t$ that~$\q_j^{\pi}(j+t)=t\cdot \widehat{M}$.
 In particular,
 $\lb_j\leq \q_j^{\pi}(j+t)\leq \ub_j$
 for every~$t\in\set{\tau+1-j}$,
 where the second inequality follows from the definition of~$\widehat{M}=1/(M+1)$.

 \LD{}
 Let~$\pi=(\pi_1,\dots,\pi_M)$ be a solution to~$I'$.
 We first claim that for every~$r_j\in\rSet$
 there is exactly one~$t\in\set{M}$ and exactly one~$i\in\set{N}$
 such that~$\pi_t(a_i)=\r_j$.
 Note that
 if resource~$r_j$ is not harvested at all,
 we have~$\q_j^{\emptyset}(t)<1$ for all~$t\in \set{j-1}$,
 $\q_j^{\emptyset}(j)=1$,
 and $\q_j^{\emptyset}(t)>1=\ub_j$ for all~$t\in \set[j+1]{\tau+1}$.
 This implies that along~$\pi$,
 resource~$\r_j$ cannot be harvested before period~$j$ due to no-over-extraction
 (each agent has weight exactly one),
 and must be harvested at period~$j$ by exactly one agent~$a_i$ due to capacity and no-over-extraction.
 When~$\r_j$ is harvested in period~$j$ by exactly one agent~$a_i$,
 we have that $\q_j^{\pi}(j+t)=t\cdot \widehat{M}<1$,
 and hence,
 $\r_j$ is not harvested again due to no-over-extraction.
 Thus,
 the claim follows.
 We next claim that~$\calX=(X_1,\dots,X_N)$ with
 $X_i=\{x_j\mid \pi_j(a_i)=\r_j\}$ is a solution to~$I$.
 By the previous claim,
 we know that~$\calX$ is a partition.
 Since~$\pi$ is envy-free,
 so is~$\calX$.
 Hence,
 the claim follows.
\end{proof}
}%

\begin{theorem}[\appref{thm:n:bin:pnph}]
	\label{thm:n:bin:pnph}
	\fish{} is \NP-hard even for one agent.
\end{theorem}

We reduce from the \prob{Separated Position Matching} problem,
where,
given a non-decreasing sequence $D = (d_1,\dots, d_N)$ of positive integers
and two sets $B = \{b_1,\dots, b_N\}$
and
$Z = \{z_1,\dots, z_N\}$,
each of $N$ distinct
positive integers,
such that $\max D= \max B < \min Z$
and all numbers are in~$O(N)$,
the question is whether there is a set~$\calM\subseteq \set{N}\times B\times Z$ of triples
so that every $i \in \set{N}$, $b\in B$,
$z\in Z$ occurs exactly once in $\calM$ and for every triple
$(i, b, z) \in \calM$ it holds that $b \leq d_i$ and $z \leq d_i + b$.
\NP-hardness of the problem follows from~Corollary~4 together with Lemma~13
of Kanellopoulos~\etal~\cite{KanellopoulosPK26arxiv}.

\begin{proof}[Proof idea of \cref{thm:n:bin:pnph}]
 Let~$T\ceq 1+\max\{\max Z,\max B+d_N\}$
and~$\tau=3T-1$.
We construct one agent~$a$ of weight~$W=2^T$.
 Blocker resources force~$a$ to harvest in all periods except
 $T+b$ for~$b\in B$ and~$2T+z$ for~$z\in Z$.
 These thus form two phases with~$N$ available periods each.
 For each index~$i$, we create a candidate resource~$\r_i$
 with renewal function~$\rf_i(x)=x/2+W$.
 Starting from~$3W/2$, its quantity approaches~$2W$.
 Its upper bound~$2W-2^{-d_i}$ forces its first harvest to occur by period~$T+d_i$.
 If this harvest occurs in period~$T+b$,
 the next harvest must occur by period~$2T+b+d_i$.
 Thus, harvesting~$\r_i$ in periods~$T+b$ and~$2T+z$ implies~$b\leq d_i$ and~$z\leq d_i+b$.
 Each of the~$N$ candidates requires a harvest in each phase.
 Since~$a$ can harvest only one resource per period,
 every~$b\in B$ and~$z\in Z$ is used exactly once.
 Hence, every sustainable assignment determines a matching.
\end{proof}

\appendixproof{thm:n:bin:pnph}{
\begin{construction}
 \label{constr:n:bin:pnph}
 Let~$I=(D,B,Z)$ be an instance of \prob{Separated Position Matching}
 with $D = \langle d_1,\dots, d_N\rangle$,
 $B = \{b_1,\dots, b_N\}$,
 and
 $Z = \{z_1,\dots, z_N\}$.
 We construct an instance~$I'=(\rSet,A,\tau)$ of \fish{} as follows.
 Let~$T\ceq 1+\max\{\max Z,\max B+d_N\}$ and the number of periods be~$\tau=3T-1$.
 We have one agent~$a$ with weight~$W\ceq 2^{T}$.
 Note that $\log(W)\in O(N)$.

 Let~$T_B\ceq \{T+b\mid b\in B\}$ and~$T_Z\ceq \{2T+z\mid z\in Z\}$.
 By the choice of~$T$,
 the two sets are disjoint and~$\max T_Z\leq \tau$.
 For every~$t\in\set{\tau}\setminus (T_B\cup T_Z)$,
 there is a blocker resource~$r_t$,
 with~$\iq_t=W-t+1$,
 $\lb_t=1$,
 $\ub_t=W$,
 and~$\rf(x)=x+1$.
 Since~$W\geq\tau$, each blocker must be harvested in its designated
period and cannot be harvested again within the horizon.
 Now, construct the set~$\rSet_D=\{\r_i\mid d_i\in D\}$,
where each resource~$\r_i$ has  \[ \iq_i=2W-\frac{W}{2} ,\, \lb_i=1 ,\, \ub_i=2W-\frac{1}{2^{d_i}} ,\, \rf_i(x)=\frac{1}{2}x+W . \]
 Note that, if~$\r_i$ has not been harvested before period~$t$, then~$\q_i(t)=2W-W/2^t$.
 All resources are approved by $a$.
\end{construction}

\begin{lemma}
 \label{lem:n:bin:pnph}
 Let~$I'$ be a \yes-instance.
 For every resource~$\r_i\in\rSet_D$,
 agent~$a$ is assigned to~$\r_i$ in period~$T+b$
 for some~$b\in B$ with~$b\leq d_i$,
 and in period~$2T+z$
 for some~$z\in Z$ with~$z\leq d_i+b$.
 Moreover, in every period from~$T_B\cup T_Z$,
 agent~$a$ is assigned to a resource from~$\rSet_D$.
\end{lemma}

\begin{proof}
 Let~$\pi$ be a sustainable and envy-free~\xassignment[$\tau$].
 Due to the blocker structure, we know that~$a$ can only be assigned to a resource from~$\rSet_D$ in periods within~$T_B$ and~$T_Z$.
 Now suppose towards a contradiction that there is a resource~$\r_i$ such that~$a$ is not assigned to it in any of the periods from~$T_B$. Since~$\max D= \max B$,
 we have that~$\q_i^\pi(T+d_i+1)=2W-\frac{1}{2^{d_i+1}}>\ub_i$,
 a contradiction to sustainability.
 It follows that~$a$ is assigned to~$\r_i$ in period~$T+b\in T_B$ with~$b\leq d_i$,
 and~$a$ is assigned to each resource from~$\rSet_D$ in a period from~$T_B$.
 Since there are~$N$ periods in~$T_B$ and~$N$ resources in~$\rSet_D$,
 it follows that~$a$ is assigned exactly once to each resource from~$\rSet_D$ in a unique period from~$T_B$.
 Now suppose towards a contradiction that
 there is a resource~$\r_i$ such that~$a$ is not assigned to it in any of the periods from~$T_Z$.
 Let~$a$ be assigned to~$\r_i$ at period~$T+b\in T_B$.
 Then,
 $\q_i^\pi(T+b)=2W-\frac{1}{2^{b}}$
 and~$\q_i^\pi(T+b+1)=\frac{1}{2}(2W-\frac{1}{2^{b}}-W)+W = 2W-\frac{1}{2^{b+1}}-\frac{W}{2}$.
 Thus,
 after~$T+d_i$ further periods,
 we have $\q_i^\pi(T+b+T+d_i)= 2W-\frac{1}{2^{b+T+d_i}}-\frac{1}{2^{d_i}}\leq\ub_i$,
 but $\q_i^\pi(T+b+T+d_i+1)= 2W-\frac{1}{2^{b+T+d_i+1}}-\frac{1}{2^{d_i+1}}>2W-\frac{1}{2^{d_i+1}}-\frac{1}{2^{d_i+1}}=\ub_i$.
 Note that since~$b+d_i\leq T-1$,
 we have~$2T+b+d_i+1\leq \tau+1$.
 It follows that~$a$ is assigned to~$\r_i$ in period~$2T+z\in T_Z$ with~$z\leq d_i+b$,
 and~$a$ is assigned to all resources from~$\rSet_D$ in all periods from~$T_Z$.
\end{proof}

\begin{proof}[Proof of \cref{thm:n:bin:pnph}]
 Let~$I'$ be the instance of \fish{}
 obtained from an instance $I=(D,B,Z)$ of \prob{Separated Position Matching}
 with $D = \langle d_1,\dots, d_N\rangle$,
 $B = \{b_1,\dots, b_N\}$,
 and
 $Z = \{z_1,\dots, z_N\}$,
 using \cref{constr:n:bin:pnph}.
 Since all input numbers are in~$O(N)$,
we have~$T,\tau\in O(N)$.
The construction introduces~$O(N)$ resources,
and every numerical parameter has a binary encoding of size~$O(N)$.
Thus, the construction runs in polynomial time.
 We prove that~$I$ is a \yes-instance if and only if~$I'$ is a \yes-instance.

 \RD{}
 Let~$M$ be a solution to $I$.
 For each~$t\in \set{\tau}\setminus (T_B\cup T_Z)$,
 assign~$a$ to the blocker resource~$r_t$.
 For each~$(i,b,z)\in M$,
 assign~$a$ to~$\r_i$ in periods~$T+b$ and~$2T+z$.
 This is possible since each~$b\in B$ and~$z\in Z$ appears in exactly one triple of~$M$.
 It remains to show that $\pi$ is sustainable.
 Since~$b\leq d_i$,
 we have~$\q_i^\pi(T+b)=2W-\frac{1}{2^{b}}\leq \ub_i$,
 and $\q_i^\pi(T+b+1)=2W-\frac{1}{2^{b+1}}-\frac{W}{2}<\ub_i$.
 Since~$z\leq d_i+b$ and~$z>b$,
 we have~$z=b+y$ with~$y\leq d_i$,
 and $\q_i^\pi(2T+b+y)= 2W-\frac{1}{2^{b+T+y}}-\frac{1}{2^{y}}\leq\ub_i$,
 and $\q_i^\pi(2T+b+y+t)= 2W-\frac{W}{2^t}-\frac{1}{2^{T+b+y+t}}-\frac{1}{2^{y+t}}$.
 By the choice of~$\tau$,
 we consider~$t \leq T-b-y\leq T$.
 In particular, we have
$\q_i^\pi(\tau+1)\leq 2W-\frac{W}{2^T}
=2W-1<\ub_i$.
 Since all resource quantities are at least one in every period
and at~$\tau+1$, viability holds.
Since there is only one agent, the assignment is trivially envy-free.

 \LD{}
 Let~$\pi$ be a sustainable and envy-free \xassignment[$\tau$]. By construction every blocker resource~$r_t$ has to be harvested in period~$t$ for every~$t\in\set{\tau}\setminus (T_B\cup T_Z)$. By~\cref{lem:n:bin:pnph},
let~$T+b_i'\in T_B$ be the first harvesting period of~$\r_i$,
and let~$2T+z_i'\in T_Z$ be its second harvesting period.
We claim that~$\calM\ceq\{(i,b_i',z_i')\mid i\in\set{N}\}$
is a solution to~$I$.
 Due to \cref{lem:n:bin:pnph},
 we know that
 $B=\{b_i'\mid i\in\set{N}\}$ and $Z=\{z_i'\mid i\in\set{N}\}$,
 and hence every number is covered exactly once.
 Also due to \cref{lem:n:bin:pnph},
 we know that~$b_i'\leq d_i$ and~$z_i'\leq d_i+b_i'$.
 Hence,
 $M$ is a solution to~$I$.
\end{proof}}

With three agents, hardness persists even under unary encoding
and when all renewal functions are translations.

\newcommand{\as}{a_{\rm s}}
\newcommand{\at}{a_{\rm t}}
\newcommand{\av}{a_{\rm v}}
\newcommand{\rd}{\overline{d}}
\newcommand{\rb}{\overline{b}}
\newcommand{\rz}{\overline{z}}
\newcommand{\mainLine}[1]{%
	\draw[->,ultra thick] (-0.1,0) -- (7.1,0);

	\foreach \x in {1,...,6}{
		\node (x\x) at (\x,0)[inner sep=0pt]{};
		\draw (\x,0.08) -- (\x,-0.08);
	}

	\node (x0) at (0.25,0)[inner sep=0pt]{};
	\draw (0.25,0.08) -- (0.25,-0.08);
	\node (x7) at (6.75,0)[inner sep=0pt]{};
	\draw (6.75,0.08) -- (6.75,-0.08);

	\foreach \x/\l in {%
		0/{selector},
		1/{cleanup},
		2/{target~$z$},
		3/{cleanup},
		4/{verify $i$},
		5/{cleanup},
		6/{verify $b_j$}
		}{%
		\node[above,align=center,anchor=south,yshift=-1 pt] at (x\x) {\scriptsize\strut\l};
	}

	\ifnum#1>0
		\foreach \x/\l in {%
			0/{$S_i$},
			1/{$c_e$},
			2/{$\theta_e,\ T_z$},
			3/{$c_e^{(1)}$},
			4/{$\mathbb D_i$},
			5/{$c_e^{(2)}$},
			6/{$\mathbb B_j$},
			7/{$\tau$}
			}{%
			\node[above,align=center,anchor=south,yshift=7 pt] at (x\x) {\scriptsize\strut\l};
		}
	\else%
		\node[above,align=center,anchor=south,yshift=-1 pt] at (x7) {\scriptsize\strut$\tau$};
	\fi%

	\node[below] at (x0) {\small~$\approx 3L$};
	\node[below] at (x1) {\small~$\approx 4L$};
	\node[below] at (x2) {\small~$\approx 6L$};
	\node[below] at (x3) {\small~$\approx 8L$};
	\node[below] at (x4) {\small\ \ $\approx 10L$};
	\node[below] at (x5) {\small\ \ $\approx 12L$};
	\node[below] at (x6) {\small\ \ $\approx 14L$};
	\node[below] at (x7) {\small$15L$};
}
\newcommand{\unselLine}{%
	\node (uns) at (1,-0.75)[inner sep=0pt]{};
	\draw[very thick,->,red]
		(x1|-uns) -- ($(x5|-uns)-(0.1,0)$);

	\node[left,align=right,red,xshift=-1.75mm] at (uns)
		{\scriptsize unselected\\[-2mm]\scriptsize candidate};

	\foreach \x in {1,3,5}
		\draw[densely dotted,thick,red]
			(x\x) -- (x\x|-uns);

	\node[vernode,draw=red,text=red] (uv1) at (x1|-uns) {$Q$};
	\node[vernode,draw=red,text=red] (uv2) at (x3|-uns) {$Q$};
	\node[vernode,draw=red,text=red] (uv3) at (x5|-uns) {$Q$};

	\node[below=0.3mm of uv1,align=center,font=\scriptsize,red]
		{$q_e=Q$};

	\node[below=0.3mm of uv2,align=center,font=\scriptsize,red]
		{$q_e=\ub_e$};

	\node[below=0.3mm of uv3,align=center,font=\scriptsize,red]
		{$Q\le q_e\le\ub_e$};
}
\newcommand{\selLine}{%
	\node (sel) at (6,-1.5)[inner sep=0pt]{};
	\draw[very thick,->,green!50!black]
		(x0|-sel) -- ($(x6|-sel)-(0.1,0)$);

	\node[right,align=left,green!50!black,xshift=+2mm] at (sel)
		{\scriptsize selected\\[-2mm]\scriptsize candidate};

	\foreach \x in {0,2,4,6}
		\draw[densely dotted,thick,green!50!black]
			(x\x) -- (x\x|-sel);

	\node[selnode] (ss) at (x0|-sel) {$3Q/4$};

	\node[selnode] (st) at (x2|-sel) {$3Q/4$};

	\node[vernode,draw=green!50!black,text=green!50!black]
		(sd) at (x4|-sel) {$Q$};

	\node[vernode,draw=green!50!black,text=green!50!black]
		(sb) at (x6|-sel) {$Q$};

	\node[novnode] (sc) at (x1|-sel) {};
	\node[novnode] (sl1) at (x3|-sel) {};
	\node[novnode] (sl2) at (x5|-sel) {};

	\node[below=0.3mm of sc,align=center,font=\scriptsize,
				green!50!black]
		{$q_e=Q/4$};

	\node[below=0.3mm of st,align=center,font=\scriptsize,
				green!50!black]
		{$q_e\ge 3Q/4$};

	\node[below=0.3mm of sl1,align=center,font=\scriptsize,
				green!50!black]
		{$q_e<Q$};

	\node[below=0.3mm of sd,align=center,font=\scriptsize,
				green!50!black]
		{$Q\le q_e\le\ub_e$};

	\node[below=0.3mm of sl2,align=center,font=\scriptsize,
				green!50!black]
		{$q_e<Q$};

	\node[below=0.3mm of sb,align=center,font=\scriptsize,
				green!50!black]
		{$Q\le q_e\le\ub_e$};
}
\begin{theorem}[\appref{thm:n:unary:pnph}]
	\label{thm:n:unary:pnph}
	\fish{} is \NP-hard even for three agents,
	unary encoding,
	and all renewal functions being translations.
\end{theorem}

We again reduce from the \prob{Separated Position Matching} problem~\cite{KanellopoulosPK26arxiv}.

\begin{proof}[Proof idea of \cref{thm:n:unary:pnph}]
 We have three agents:
 a selection agent~$\as$ of weight~$3Q/4$ for some large number~$Q$,
 a trigger agent~$a_0$ that enforces with a private resource~$\r_0$
 that~$\as$ has to harvest in each period,
 and a verifier agent~$\av$ with weight~$Q$.
 Agent $\as$ is supposed to first select~$N$ feasible tuples~$e=(i,j)$ such that~$b_j\leq d_i$,
 and then to select a target~$z$ for each selected tuple~$(i,j)$
 such that~$z\leq d_i+b_j$.
 Agent~$\av$ is used to verify that for every two distinct tuples~$(i,j)$ and~$(i',j')$ selected by~$\as$ we have that~$i\neq i'$ and~$j\neq j'$. The two validation periods for a selected candidate~$(i,j)$
depend only on~$i$ and~$j$, respectively.
Thus, two distinct tuples sharing an index would require~$\av$ to harvest two resources in the same period.
 This already captures the requirement of a solution to the input problem.
 Now,
 we have a candidate resource~$\r_e$ for each feasible tuple~$e=(i,j)$.
 A resource is either selected twice by~$\as$,
 and then twice by~$\av$ as described above,
 or it is not selected by~$\as$ and then follows a trajectory for unselected candidates,
 in which~$\av$ harvests it three times
 (see \cref{fig:n:unary:pnph} for an overview).
	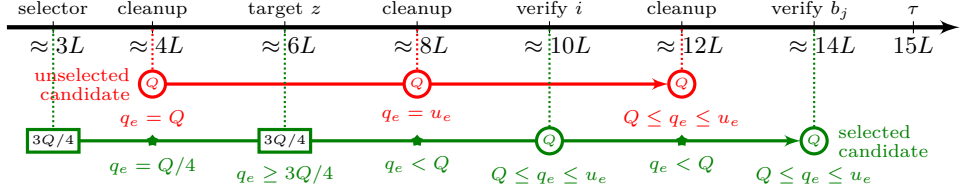
\begin{figure}
	\centering
		\begin{tikzpicture}[
				x=1.75cm,y=1cm,>=latex',
				selnode/.style={
					rectangle,draw=green!50!black,very thick,
					fill=white,
					inner sep=1.75pt,font=\tiny
				},
				tarnode/.style={
					diamond,draw=green!50!black,very thick,
					fill=white,
					inner sep=0.5pt,font=\tiny
				},
				vernode/.style={
					circle,draw,very thick,
					fill=white,
					inner sep=1pt,font=\tiny
				},
				novnode/.style={
					star,draw=green!50!black,very thick,
					fill=green!50!black,
					minimum size=1mm,inner sep=0pt
				}
			]

			\mainLine{0}
			\unselLine{0}
			\selLine{0}
		\end{tikzpicture}
		\caption{Illustration to the proof of \cref{thm:n:unary:pnph}.
		Rectangular
		and circular node shapes correspond to agents~$\as$
		and~$\av$,
		respectively,
		each enclosing the corresponding agent weight.}
		\label{fig:n:unary:pnph}
	\end{figure}
 These two trajectories are parallel,
 but their phases alternate and
 are separated by large gaps of order~$L$,
 ensuring that
 they do not overlap or interfere.
 In particular,
 if~$\as$ selects a candidate resource~$\r_e$,
 $\av$ cannot harvest~$\r_e$ before~$\as$ gets another chance to harvest it.
 To avoid envying~$a_0$, $\as$ must harvest in every unblocked period.
During the target phase, only previously selected candidates
are available to her, and each can be harvested at most once.
Thus, she must harvest every selected candidate a second time.
 Technically,
 $\as$ is blocked by private resources in all time steps except for~$N$ in each of the selection and target phases. The~$N$ target periods correspond to the distinct values~$z\in Z$,
so the~$N$ second harvests use every target exactly once.
 Each candidate resource~$\r_{(i,j)}$ has a translation as a renewal function,
 whose rate depends on~$i$, $d_i$, and~$b_j$.
This ensures that
 each event of a candidate resource
 (selection, target, validation, or cleanup)
 is unique and well-timed.
 In particular,
 it does not depend on when a previous harvest happened, but
 only on the amount harvested.
 This gives exactly the initial switch between the two trajectories.
 Moreover, the rate is chosen so that~$\av$ cannot be assigned
to candidate resources during the selection phase,
nor to selected candidates during the target phase.
At their cleanup and validation periods,
these resources are urgent and have quantity at least~$Q$,
so~$\av$ must harvest them.
\end{proof}

\appendixproof{thm:n:unary:pnph}
{%
\begin{construction}
 \label{constr:n:unary:pnph}
 Let~$I=(D,B,Z)$ be an instance of \prob{Separated Position Matching}
 with $D = ( d_1,\dots, d_N)$,
 $B = \{b_1,\dots, b_N\}$,
 and
 $Z = \{z_1,\dots, z_N\}$,
 where~$N$ is sufficiently large.
 We construct an instance~$I'$ of \fish{} as follows
 (see~\cref{fig:n:unary:pnph:full} for an overview).
 \begin{figure}
	\centering
		\begin{tikzpicture}[
				x=1.75cm,y=1cm,>=latex',
				selnode/.style={
					rectangle,draw=green!50!black,very thick,
					fill=white,
					inner sep=1.75pt,font=\tiny
				},
				tarnode/.style={
					diamond,draw=green!50!black,very thick,
					fill=white,
					inner sep=0.5pt,font=\tiny
				},
				vernode/.style={
					circle,draw,very thick,
					fill=white,
					inner sep=1pt,font=\tiny
				},
				novnode/.style={
					star,draw=green!50!black,very thick,
					fill=green!50!black,
					minimum size=1mm,inner sep=0pt
				}
			]

			\mainLine{1}
			\unselLine{1}
			\selLine{1}
		\end{tikzpicture}
		\caption{Illustration to the proof of \cref{thm:n:unary:pnph}.
		Rectangle
		and circle node shapes correspond to agents~$\as$
		and~$\av$,
		respectively,
		each enclosing the corresponding agent weight.}
		\label{fig:n:unary:pnph:full}
	\end{figure}
 Define
 \[ M = 4\cdot (1 + \max Z)^2,\ \
	  L = M^2,\ \
	  Q = 8\cdot (5\cdot L + 1)^2,\ \
    \text{and }
    \tau = 15 L.
 \]
 By this choice,
 we have
 $N\ll M$,
 $M\cdot \max Z\ll L\ll Q$,
 and~$L<\tau\ll L^2 < Q$.
 We rescale all input numbers by~$M$ so that the index offsets defined below
can distinguish the entries~$d_i,d_{i'}$ from~$D$ with~$i\neq i'$
even when~$d_i=d_{i'}$;
for every~$i,j\in\set{N}$ and~$z\in Z$, let
 \[ \rd_i \ceq M\cdot d_i, \qquad  \rb_j \ceq M\cdot b_j, \qquad \text{and}\quad  \rz \ceq M\cdot z .\]
 We have three agents:
 a trigger agent~$a_0$ with weight~$1$,
 a selection agent~$\as$ with weight~$3Q/4$,
 and a verifier agent~$\av$ with weight~$Q$.
 Let $\calE\ceq \{(i,j)\in\set{N}\times\set{N}\mid b_j\leq d_i\}$ be the set of all feasible candidate pairs among~$D$ and~$B$.
 We will construct a resource~$\r_e$ for each of these candidates,
 resulting in the set~$\rSet_\calE\ceq\{\r_e\mid e\in \calE\}$ of such resources.
The number~$L$ is large enough to separate the periods
into non-overlapping phases for handling the candidates.
 We will choose the renewal translations so that after roughly~$4L$ periods,
 a quantity of~$Q$ is regained.
 We have two parallel trajectories for the candidates:
 If a candidate~$(i,j)$ is not selected,
 it will be reset repeatedly roughly at~$4L$, $8L$, and~$12L$,
 each time by~$\av$.
 If a candidate~$(i,j)$ is selected roughly at~$3L$ by~$\as$,
 a feasible target~$z\in Z$ is selected roughly at~$6L$ by~$\as$,
 and then the uniqueness of~$i$ and~$j$ is verified by~$\av$ roughly at $10L$ and~$14L$,
 respectively.

 For the selection phase,
 for all $k\in\set{N}$ define
 \[ S_k \ceq 3L + k. \]
 These are the only periods in which~$\as$ is supposed to select a candidate resource from~$\rSet_\calE$.
 If a resource~$\r_e$ with~$e=(i,j)\in\calE$ is selected by~$\as$,
 it should be available to~$\as$ at~$\theta_e$ or later
 after roughly another $3L$ periods,
 corresponding to a quantity regain of roughly~$3Q/4$ which equals the weight of~$\as$.
To encode feasibility, we require~$\r_e$ to be available to~$\as$
at a target period~$T_z$ exactly when~$z\leq d_i+b_j$.
Thus, we require
$T_z\geq\theta_e$ if and only if~$z\leq d_i+b_j$.
 Thus,
 for each~$e=(i,j)\in\calE$ and~$z\in Z$,
 let
 \[ \theta_e \ceq 6L - (\rd_i+\rb_j)+4i \ \text{ and }\  T_z \ceq 6L - \rz + 4\cdot N ,\]
 resulting in~$T_z - \theta_e = \rd_i+\rb_j - \rz + 4(N-i)$,
 which yields the requirement by the choice and rescaling of~$M$.
 Note that since~$d_i=d_{i'}$ is possible for~$i\neq i'$,
 $\theta_e$ carries the unique index in form of~$4i$ to differentiate,
 and~$T_z$ compensates for it with~$4N$.
 Now we build all periods in their phases around~$\theta_e$.
 We define~$\Gamma_e$ depending on~$i$, $d_i$, and~$b_j$ of roughly~$2L$ such that
 candidate~$e=(i,j)$ is validated at~$\mathbb{D}_i \ceq\theta_e+2\Gamma_e$ for~$i$ and~$d_i$,
 and at~$\mathbb{B}_j \ceq\theta_e+4\Gamma_e$ for~$b_j$.
 Thus,
 to ensure a unique dependence in either case,
 let
 \[ \Gamma_e \ceq 2L+\frac{1}{4}\rd_i + \frac{1}{2}\rb_j-i. \]
 Then,
 \[ \mathbb{D}_i = 10L-\frac{1}{2}\rd_i+2i
  \ \text{ and }\
  \mathbb{B}_j = 14L+\rb_j
 \]
 denote designated periods for which~$d_i$ and~$b_j$ will be verified,
 respectively.
 Note that for every~$e=(i,j)$,
 it holds that~$\mathbb{D}_i = \theta_e + 2\Gamma_e$ and~$\mathbb{B}_j = \theta_e + 4\Gamma_e$.
 By the choice of~$M$ against~$L$,
 we have that~$\mathbb{B}_j<\tau$ for every~$j\in\set{N}$.
 For the unselected trajectory,
 let
 \[ c_e\ceq \theta_e-\Gamma_e -1,\quad
 c_e^{(1)} \ceq c_e + 2\Gamma_e +1
 ,\quad \text{ and } \quad
 c_e^{(2)} \ceq c_e^{(1)} + 2\Gamma_e,
 \]

 Note that~$c_e^{(1)}= \theta_e+\Gamma_e$ and $c_e^{(2)} = \theta_e+3\Gamma_e$.
 Now we are set to provide all resource details.
 Construct a resource~$\r_0$ with~$\iq_0=2$, $\lb_0=1$,
 $\ub_0=2$, and~$\rf_0(x)=x+1$.
 For each~$e\in\calE$,
 let
 \[
	\beta_e\ceq \left\lfloor \frac{Q}{2\Gamma_e + 1/4} \right\rfloor
 \]
 and construct a resource~$\r_e$ with
 \[ \iq_e = Q - (c_e-1)\cdot \beta_e,\ \
		\lb_e = 1,\ \
		\ub_e = (2\Gamma_e+1)\cdot \beta_e,\ \
		\rf_e(x) = x + \beta_e.
 \]
 The rate~$\beta_e$ is chosen carefully such that~$Q\in(2\Gamma_e\cdot\beta_e,\ub_e)$,
 ensuring that,
 intuitively,
 $\r_e$ repeatedly enters a single-period window in which it is urgent and can be harvested by~$\av$.
 A \emph{blocker} for an agent~$a$ at period~$i$ is a resource~$\r_{a,i}$
 with~$\iq = w_a - i + 1$, $\lb = 1$, $\ub = w_a$ and~$\rf(x)=x+1$
 only approved by~$a$.
 Construct a blocker for~$\as$ for each period~$t\in\set{\tau}\setminus (\{S_k\mid k\in \set{N}\} \cup \{T_z\mid z\in Z\})$.
Since the weight of~$\as$ is larger than~$\tau$,
a blocker does not become urgent again within the time horizon
after being harvested in its designated period.
 Finally,
 only~$a_0$ and~$\as$ approve~$\r_0$,
 and only~$\as$ and~$\av$ approve~$\r_e$ for all~$e\in\calE$.
\end{construction}

\begin{observation}
 \label{obs:n:unary:pnph:ST}
 For every~$e\in \calE$,
 $k\in\set{N}$,
 and~$z\in Z$,
 it holds true that
 \[ c_e-S_k < \Gamma_e/2 \qquad \text{ and }\qquad T_z-c_e<3\Gamma_e/2. \]
\end{observation}

\begin{proof}
 Note that~$c_e = 4L - \frac{5}{4} \rd_i - \frac{3}{2}\rb_j + 5i - 1 \in (3L,4L)$
 and $\Gamma_e > 2L$.
 Hence,
 since~$S_k>3L$,
 we have~$c_e - S_k < L < \Gamma_e/2$.
 Moreover,
 since~$T_z<6L$,
 we have
 $T_z-c_e < 3L < 3\Gamma_e/2$.
\end{proof}

Next,
we show that for a resource~$\r_e$ the number of periods to add a quantity of~$Q$
lies in a small window between $2\Gamma_e$ and~$2\Gamma_e+1$.

\begin{observation}
 \label{obs:n:unary:pnph}
 For every~$e\in \calE$,
 we have
 \[ 2\Gamma_e+\frac{1}{4} \leq \frac{Q}{\beta_e} < 2\Gamma_e + \frac{2}{5} .  \]
 Moreover,
 it holds that $\left\lceil (x\cdot Q)/(4\cdot \beta_e) \right\rceil = x\cdot\Gamma_e/2+1$
 for all~$x\in\{2,4,6,8,10\}$.
\end{observation}

\begin{proof}
 Let~$Y\ceq 2\Gamma_e+1/4$.
 By the definition of~$\beta_e$,
 we have~$Y\leq Q/\beta_e$.
 Since~$\beta_e>Q/Y-1$,
 we have
 \begin{align}
		Y \leq \frac{Q}{\beta_e} &< \frac{Q}{Q/Y-1} =  \frac{YQ}{Q-Y}
		=  \frac{YQ-Y^2+Y^2}{Q-Y} =  Y + \frac{Y^2}{Q-Y}.
		\label{eq:obs:n:unary:pnph:1}
 \end{align}
 With~$Y<5L+1$,
 we have
 \begin{align}
		\frac{Y^2}{Q-Y} < \frac{(5L+1)^2}{Q-(5L+1)} = \frac{1}{8-(5L+1)^{-1}} < \frac{1}{7}.
		\label{eq:obs:n:unary:pnph:2}
 \end{align}
 Combining \eqref{eq:obs:n:unary:pnph:1} with \eqref{eq:obs:n:unary:pnph:2},
 we get
 \[ \frac{Q}{\beta_e} < Y + \frac{1}{7} < 2\Gamma_e+\frac{2}{5}. \]
For the second statement,
 the first one yields~$\frac{xQ}{4\beta_e}\in\left[\frac{x\Gamma_e}{2}+\frac{x}{16},\frac{x\Gamma_e}{2}+\frac{x}{10}\right)$,
 where~$\frac{x\Gamma_e}{2}\in\N$ for even~$x$ and~$0<\frac{x}{16}<\frac{x}{10}\leq 1$ for~$x\leq 10$.
\end{proof}

\begin{lemma}
 \label{lem:n:unary:pnph:unsel}
 Let~$I'$ be a \yes-instance.
 If resource~$\r_e$ for some~$e\in\calE$ is not harvested by~$\as$
 at any~$S_k$,
 then it is harvested exactly in the periods
 $c_e$,
 $c_e^{(1)}$,
 and $c_e^{(2)}$,
 each time by agent~$\av$.
\end{lemma}

\begin{proof}
 If $\r_e$ is not harvested by~$\as$ at any~$S_k$,
 it is also not harvested by~$\av$ before period~$c_e$.
 At~$c_e$,
 its quantity is~$\q_e(c_e) = Q - (c_e-1)\beta_e + (c_e-1)\beta_e = Q$.
 Note that~$Q < u_e < Q+\beta_e$,
 due to \cref{obs:n:unary:pnph} and
 $\frac{u_e}{\beta_e} = 2\Gamma_e + 1 < \frac{Q}{\beta_e} + 1$.
 Hence,
 $\r_e$ must be harvested at period~$c_e$
 and only~$\av$ is free to do so.
 Note after this period~$c_e$,
 we have~$\q_e(c_e+t)=t\cdot \beta_e$
 until the next harvest happens.
 Combining~\cref{obs:n:unary:pnph:ST} and \cref{obs:n:unary:pnph},
we have
$T_z-c_e<3\Gamma_e/2<3Q/(4\beta_e)$
for every~$z\in Z$.
Hence, $\q_e(T_z)<3Q/4$,
and~$\as$ cannot harvest~$\r_e$ in this period.
 By~\cref{obs:n:unary:pnph},
 we know that an amount of~$Q$ becomes available after~$2\Gamma_e+1$ periods.
 Thus,
 the quantity then is~$\q_e(c_e^{(1)}) = (2\Gamma_e+1)\cdot \beta_e = \ub_e$.
 So,
 $\av$ must harvest~$\r_e$.
 We find the next period in which~$\r_e$ must and can be harvested
 (recall that for the remainder of periods,
 $\as$ is blocked).
 After the harvest,
 we have~$\q_e(c_e^{(1)}+t) = (2\Gamma_e+1+t)\cdot\beta_e - Q$.
 So,
 after~$2\Gamma_e-1$ more periods,
 we have~$\q_e(c_e^{(1)}+2\Gamma_e-1) = 4\cdot \Gamma_e\cdot \beta_e  - Q < Q$,
 where the latter follows from~$2\Gamma_e < Q/\beta_e$ due to~\cref{obs:n:unary:pnph}.
 Hence,
 no harvest by~$\av$ can happen until this period.
 Yet,
 after $2\Gamma_e$ more periods,
 we have~$\q_e(c_e^{(2)}) = \q_e(c_e^{(1)}+2\Gamma_e) =2\cdot (2\Gamma_e+1/2)\cdot \beta_e - Q \in (Q,\ub_e)$,
 where the latter follows from~$2\Gamma_e+1/2> 2\Gamma_e+2/5 > Q/\beta_e$.
 Thus,
 $\r_e$ can be harvested in this period.
 In fact,
 it has to since
 $\q_e(c_e^{(2)}+1) = 2(2\Gamma_e+1)\cdot \beta_e - Q > (2\Gamma_e+1)\cdot \beta_e = \ub_e$,
 where the inequality follows from~$Q/\beta_e<2\Gamma_e+1$.
After the third harvest, the resource becomes available to~$\av$
again no earlier than~$c_e+6\Gamma_e+1>15L=\tau$.
\end{proof}

\begin{lemma}
 \label{lem:n:unary:pnph:sel}
 Let~$I'$ be a \yes-instance.
 Then,
 in every solution,
 we have the following.
 Agent~$\as$ harvests exactly~$N$ different resources from~$\rSet_\calE$,
 and each of them twice.
 If resource~$\r_e$ for some~$e=(i,j)\in\calE$ is harvested by~$\as$ in some period~$S_k$,
 then it is harvested by~$\as$ in some period~$T_z$ with~$z\leq b_j+d_i$,
 and then harvested by~$\av$ exactly in the periods~$\mathbb{D}_i$ and~$\mathbb{B}_j$.
\end{lemma}

\begin{proof}
 When~$\as$ harvests~$\r_e$ at some period~$S_k$ for the first time,
 we have for~$S_k<t\leq c_e$
 that~$\q_e(t)=Q/4 + (t-c_e)\cdot\beta_e$.
 In particular,
 $\q_e(c_e)=Q/4$,
 and hence~$\av$ cannot harvest~$\r_e$ at this period
 and~$\as$ cannot harvest~$\r_e$ again in another period~$S_{k'}$.
 Together with envy-freeness,
 it follows that~$\as$ harvests exactly~$N$ distinct resources from~$\rSet_\calE$.
 Moreover, after the first harvest by~$\as$,
even without further harvests, the quantity remains below~$Q$
until the end of the target phase:
$Q/4+(T_z-c_e)\beta_e
<Q/4+3\Gamma_e\beta_e/2<Q$
for every~$z\in Z$,
by~\cref{obs:n:unary:pnph:ST,obs:n:unary:pnph}.
Thus, $\av$ cannot harvest~$\r_e$ before the target phase ends.
 When~$\q_e(t)\geq 3Q/4$,
 agent~$\as$ could harvest it.
 This happens not yet at period~$\theta_e - 1 = c_e+ \Gamma_e$,
 since~$\Gamma_e\cdot \beta_e< Q/2$,
 but at period~$\theta_e= c_e+\Gamma_e+1$,
 since~$(\Gamma_e+1)\cdot \beta_e> Q/2$.
 Since~$\as$ is blocked after~$S_N$ at all periods except for~$\{T_z\mid z\in Z\}$,
 if~$\as$ harvests~$\r_e$,
 then we have that~$T_z\geq \theta_e$,
 which is the case if and only if~$z\leq d_i+b_j$.

Note that~$\as$ cannot harvest during the target phase
a resource~$\r_{e'}$ that she did not harvest during the selection phase.
Due to~\cref{lem:n:unary:pnph:unsel},
we know that~$\r_{e'}$ was harvested by~$\av$ in period~$c_{e'}$.
Hence, until its next harvest at~$c_{e'}^{(1)}$,
its quantity is~$\q_{e'}(c_{e'}+t)=t\cdot\beta_{e'}$.
So, harvesting by~$\as$ requires
$t\geq 3Q/(4\beta_{e'})\geq \frac{3}{2}\Gamma_{e'}$.
However,
$c_{e'}+\frac{3}{2}\Gamma_{e'}
=\theta_{e'}+\frac{1}{2}\Gamma_{e'}-1>6L>T_z$
for every~$z\in Z$.
Since~$\as$ does not envy~$a_0$,
she must harvest a previously selected resource
in each of the~$N$ target periods.

 Thus,
 let~$\r_e$ be harvested at time step~$T_z$ with $z\leq d_i+b_j$
 by agent~$\as$.
 So,
 at any period~$t>T_z$ (before being harvested),
 its quantity is~$q_e(t)=Q-(c_e-t)\beta_e - \frac{3Q}{2} = (t-c_e)\beta_e-\frac{Q}{2}$.
 Since~$T_{z'}-c_e\leq 3\Gamma_e/2$,
 we have that~$q_e(T_{z'}) < Q-Q/2 < 3Q/4$,
 and hence~$\as$ cannot harvest~$\r_e$ again in another period~$T_{z'}$.
 Together with envy-freeness,
 it follows that~$\as$ harvests exactly~$N$ distinct resources from~$\rSet_\calE$,
 and that these are exactly those harvested by~$\as$ before.
 Note that
 \[ \mathbb{D}_i = c_e + 3\Gamma_e + 1 \text{ and } \mathbb{B}_j = c_e + 5\Gamma_e + 1 . \]
 So,
 at~$\mathbb{D}_i-1$,
 we have~$(3\Gamma_e)\beta_e - Q/2 < Q$,
 since~$2\Gamma_e < Q/\beta_e$.
 Yet,
 at $\mathbb{D}_i$,
 we have~$(3\Gamma_e+1)\beta_e - Q/2 \geq Q$,
 since~$2\Gamma_e + 2/3 \geq Q/\beta_e$.
 Finally,
 at $\mathbb{D}_i+1$,
 we have~$(3\Gamma_e+2)\beta_e - Q/2 > \ub_e$,
 since~$3\Gamma_e + 2 - Q/(2\beta_e) > 3\Gamma_e + 2 - (\Gamma_e+1/2) = 2\Gamma_e+3/2$.
 It follows that~$\av$ must harvest~$\r_e$ exactly in period~$\mathbb{D}_i$.

 Similarly,
 at any period~$t>\mathbb{D}_i$ (before being harvested),
 its quantity is~$q_e(t)=Q-(c_e-t)\beta_e - \frac{3Q}{2} - Q = (t-c_e)\beta_e-\frac{3Q}{2}$.
 Along a similar argument line,
 exactly at~$\mathbb{B}_j$ agent~$\av$ must harvest~$\r_e$.
 So,
 at~$\mathbb{B}_j-1$,
 we have~$(5\Gamma_e)\beta_e - 3Q/2 < Q$,
 since~$2\Gamma_e < Q/\beta_e$.
 Yet,
 at~$\mathbb{B}_j$,
 we have~$(5\Gamma_e+1)\beta_e - 3Q/2 > Q$,
 since~$2\Gamma_e + 2/5 > Q/\beta_e$.
 Finally,
 at~$\mathbb{B}_j+1$,
 we have~$(5\Gamma_e+2)\beta_e - 3Q/2 > \ub_e$,
 since~$5\Gamma_e + 2 - 3Q/2\beta_e > 5\Gamma_e + 2 - (3\Gamma_e + 1) = 2\Gamma_e + 1$.
\end{proof}

\begin{proof}[Proof of \cref{thm:n:unary:pnph}]
 Let~$I'$ be the instance of \fish{}
 obtained from an instance $I=(D,B,Z)$ of \prob{Separated Position Matching}
 with $D = ( d_1,\dots, d_N)$,
 $B = \{b_1,\dots, b_N\}$,
 and
 $Z = \{z_1,\dots, z_N\}$.
Note that all numbers in~$I'$ are in~$O(Q)=O(N^8)$,
since all numbers in~$I$ are in~$O(N)$.
Hence, their unary encodings have polynomial size.
Moreover,
the reduction runs in~$\poly(|I|)$ time.
We prove that~$I$ is a \yes-instance if and only if~$I'$ is a \yes-instance.

 \RD{}
 Let~$\calM$ be a solution to~$I$.
 We construct a \xassignment[$\tau$]~$\pi$ as follows.
 Let~$a_0$ harvest~$\r_0$ in every period.
 For each triple~$(i,b_j,z)\in\calM$,
 let~$e=(i,j)$
 and
 assign the agents as follows.

 Assign $\as$ to~$\r_e$ in period~$S_i$.
 This is possible since
 by \cref{obs:n:unary:pnph:ST},
 we have~$c_e - S_i <\Gamma_e/2 < Q/(4\beta_e)$,
 and thus,
 $Q-(c_e-S_i)\beta_e > 3Q/4$.
 Conversely,
 since~$c_e>S_i$,
 we have that~$Q-(c_e-S_i)\beta_e<Q$.
 Assign $\as$ to~$\r_e$ in period~$T_z$.
 Since~$\calM$ is a solution,
 we have that~$\q_e(T_z)=Q/4 + (T_z-c_e)\beta_e$ and $z\leq d_i+b_j$.
 Thus,
 on the one hand,
 we have that~$T_z\geq \theta_e=c_e+\Gamma_e + 1 = c_e + \lceil Q/(2\beta_e)\rceil$
 and hence~$\q_e(T_z)\geq 3Q/4$.
 On the other hand,
 we have by~\cref{obs:n:unary:pnph:ST} that~$T_z-c_e < 3\Gamma_e/2$
 and by~\cref{obs:n:unary:pnph} that~$Q/4 \leq (\Gamma_e/2+1)\beta_e$.
 Hence,
 $\q_e(T_z) < (\Gamma_e/2+1)\beta_e + 3\Gamma_e \beta_e/2 = (2\Gamma_e+1)\beta_e = \ub_e$.
 Finally,
 assign~$\as$ to all its blockers.
 Since the weight of~$\as$ is larger than~$\tau$,
 all blocker resources sustain.

 Assign~$\av$ to~$\r_e$ in periods~$\mathbb{D}_i$ and~$\mathbb{B}_j$.
 For each~$e=(i,j)\in \calE\setminus \calE^*$,
 with~$\calE^*$ containing all~$(i,j)\in\calE$ such that~$(i,b_j,z)\in\calM$,
 assign~$\av$ to~$\r_e$ in~$c_e$,
 $c_e^{(1)}$,
 and~$c_e^{(2)}$.
 Note that by choice of~$c_e$ and~$\Gamma_e$,
 these are disjoint for any two distinct~$e,e'\in\calE$.
 Within each phase, the residues modulo~$M/4$ distinguish
different indices~$i$, since~$M/4>5N$.
For fixed~$i$, distinct values~$b_j$ give distinct periods.
The verification periods are also distinct,
since the selected pairs have distinct indices~$i$ and~$j$.
The quantity calculations in the proofs of
\cref{lem:n:unary:pnph:unsel,lem:n:unary:pnph:sel}
show that the prescribed harvests are feasible
and that the candidate resources respect capacity
and terminal viability.
Together with~$\r_0$ and the blockers,
this proves sustainability.

 Note that there is no envy from~$a_0$,
 and also not from~$\as$ since she is assigned in each period.
 Consider~$\av$.
 Recall that~$\av$ does not approve~$\r_0$,
 and hence does not envy~$a_0$.
 She values the assignment of~$\as$ by~$2N$,
 yet since she harvests the same resources as~$\as$ twice,
 already from these resources she gets a utility of~$2N$.
 Hence,
 $\av$ is also not envious.

 \LD{}
 Let~$\pi$ be a sustainable envy-free \xassignment[$\tau$].
 Due to \cref{lem:n:unary:pnph:sel},
 we know that agent~$\as$ harvests exactly~$N$ different resources from~$\rSet_\calE$,
 and each of these twice.
 If resource~$\r_e$ for some~$e=(i,j)\in\calE$ is harvested by~$\as$ in some period~$S_k$,
 then it is harvested by~$\as$ in some period~$T_z$ with~$z\leq b_j+d_i$,
 and then harvested by~$\av$ exactly in the periods~$\mathbb{D}_i$ and~$\mathbb{B}_j$.
 Let~$\calM$ be the set of triples~$(i,b_j,z)$ such that~$\as$ harvests resource~$\r_{i,j}$
 in some period~$S_k$ and again in period~$T_z$.
 By construction of~$\calE$,
 we know that~$b_j\leq d_i$,
 and by the lemma,
 we know that $z\leq b_j+d_i$.
 It remains to show that each~$i\in\set{N}$,
 each~$b_j\in B$,
 and each~$z\in Z$ is represented in~$\calM$ exactly once.
 For the latter,
 note that~$\as$ harvests exactly~$N$ distinct resources from~$\rSet_\calE$,
 each twice,
 where the second harvest is only possible in the $N$ periods~$T_z$, $z\in Z$.
 Now,
 note that we have exactly~$N$ pairs~$(i,j)\in\calE$ being represented in~$\calM$.
 Suppose towards a contradiction that there are two distinct pairs~$(i,j)$ and~$(i',j')$ such that
 $\as$ harvests~$\r_{i,j}$ and~$\r_{i',j'}$ and~$i=i'$ or~$j=j'$.
 Due to the lemma,
 if~$i=i'$,
 then we know that~$\av$ harvests~$\r_{i,j}$ and~$\r_{i',j'}$ in period~$\mathbb{D}_i$;
 a contradiction to~$\pi$ being an assignment;
 if~$j=j'$,
 then we know that~$\av$ harvests~$\r_{i,j}$ and~$\r_{i',j'}$ in period~$\mathbb{B}_j$;
 a contradiction to~$\pi$ being an assignment;
 Thus,
 we have that for every two distinct~$(i,b_j,z),(i',b_{j'},z')\in\calM$ that~$i\neq i'$,
 $b_j\neq b_{j'}$,
 and~$z\neq z'$.
 It follows that~$\calM$ is a solution to~$I$.
\end{proof}
}%

\section{The Size of an Assignment: Tractability but Incompressibility}

We investigate the upper bound~$n\cdot\tau$
on the number of agent assignments over all periods.
For~$n+\tau$, our problem is \FPT{} yet admits no polynomial kernel.

\begin{theorem}
	\label{thm:ntau-fpt}
	\fish{} is FPT \wpb{} $n+\tau$.
\end{theorem}

First,
we show that XP-time is possible where the base only depends on~$m$.
Then
we show that~$m$ can be bounded by a function depending only on~$n+\tau$,
yielding \FPT{} time.

\begin{observation}
 \label{obs:ntau:xp}
 Every instance~$I$ of \fish{}
 with~$n$ agents,
 $m$ resources,
 and $\tau$ periods,
 can be solved in $(m+1)^{n\cdot\tau}\cdot \poly(|I|)$~time.
\end{observation}

\begin{proof}
	Guess the \xassignment[$\tau$]~$\pi \in ([m]\cup\{\emptyset\})^n \times\dots\times([m]\cup\{\emptyset\})^n$.
	There are at most~$(m+1)^{n\cdot\tau}$ many.
	Since~$\tau$ is unary-encoded,
	each can be verified in polynomial time.
\end{proof}

Too many critical resources reveal a \no-instance as follows.

\begin{rrule}[\appref{rrule:critical}]
  \label{rrule:critical}
	If there are more than~$n\cdot \tau$ many critical resources,
	then return~\no.
\end{rrule}
\appendixproof{rrule:critical}
{%
\begin{proof}
 If a resource is critical because $\lb_i> \q_{i}^{\pi}(\tau+1)$ or~$\q_i^{\pi}(t)<0$ for some~$t\in\set{\tau}$ without
 harvesting it,
 then the instance is trivially \no.
 Assume that each of the more than~$n\cdot\tau$ resources is critical because of $\q_{i}^{\pi}(t) > \ub_i$ for some~$t\in\set{\tau+1}$.
 Every such critical resource requires at least one agent to be assigned at least once to it.
 Since there are at most~$n\cdot\tau$ possible assignments of agents,
 at least one critical resource remains unassigned.
 Hence,
 no assignment can be sustainable.
\end{proof}
}%

The \emph{fingerprint} of a resource~$\r_i$ is the subset of agents who give positive value to~$\r_i$,
i.e.,
$\r_i$'s neighborhood in the approval graph.
Note that there are at most~$2^n$ many fingerprints.
The \emph{marking set} of $\r_i$ is the set of all sequences~$\calA=(A_1,\dots,A_\tau)$
with~$A_i\in 2^A$
of assigned agents to~$\r_i$ which are sustainable for~$\r_i$.
Note that computing the marking set of a resource can take $2^{n\cdot \tau}\poly(|I|)$ time.
Thus,
the following reduction rule runs in polynomial time only if we assume that the marking sets are given in the input.

\begin{rrule}[\appref{rrule:uncritical}]
  \label{rrule:uncritical}
	If for a fingerprint and a marking set there are more than~$n\cdot\tau$ many uncritical resources,
	then delete one of them.
\end{rrule}
\appendixproof{rrule:uncritical}
{%
\begin{proof}
 Let~$I'$ be the instance obtained from an input instance~$I$
 where resource~$\r'$
 is deleted
 and let~$\r_{i_1},\dots,\r_{i_N}$ be~$N\ceq n\cdot \tau$ resources with the same fingerprint and marking set as~$\r'$.
 Since~$\r'$ is uncritical,
 if~$I'$ is a \yes-instance,
 so is~$I$.
 Conversely,
 consider a sustainable envy-free \xassignment[$\tau$]~$\pi$ for~$I$.
 If~$\pi$ assigns no agent to~$\r'$,
 then~$\pi$ is a solution to~$I'$.
 Hence,
 we consider the case where~$\pi$ assigns an agent to~$\r'$.
 Let~$\calA=(A_1,\dots,A_\tau)$ be the sequence of agent sets with at least one non-empty set
 that~$\pi$ assigns to~$\r'$.
 Since~$\pi$ uses at most $N$ distinct resources,
 no agent is assigned to at least one of $\r_{i_1},\dots,\r_{i_N}$,
 say~$\r_{i_j}$.
 Let~$\pi'$ be the assignment obtained from~$\pi$ by
 reassigning in each period
 all agents assigned to~$\r'$ to $\r_{i_j}$,
 i.e.,
 for all~$t\in\set{\tau}$ and every agent~$a$ it holds
 \[ \pi_t'(a)  = \begin{cases}
                  \pi_t(a), &\cif{\pi_t(a)\neq \r'},
                  \\
                  \r_{i_j}, &\cif{\pi_t(a)= \r'}.
                 \end{cases}
 \]
 Note that~$(\pi_t')^{-1}(\r')=\emptyset$ for all~$t\in\set{\tau}$.
 Since~$\r_{i_j}$ has the same fingerprint as~$\r'$,
 we know that each agent in each~$A_k$ approves~$\r_{i_j}$.
 Thus,
 $\pi'$ is \nonwasteful{} and envy-free.
 Since~$\r_{i_j}$ has the same marking set as~$\r'$,
 we know that~$\calA$ is also in the marking set of~$\r_{i_j}$.
 Thus,
 $\calA$ is sustainable for~$\r_{i_j}$,
 and since
 $\r'$ is uncritical and no other assignment changed from~$\pi$,
 it follows that~$\pi'$ is sustainable.
 Thus,
 $\pi'$ is a solution to~$I$ that assigns no agent to~$\r'$,
 and hence,
 $\pi'$ is a solution to~$I'$.
\end{proof}
}%

This implies a doubly exponential upper bound on the number of resources that only depends on~$n+\tau$.

\begin{proof}[Proof of \cref{thm:ntau-fpt}]
 Let~$I$ be an instance of \fish{}
 with~$n$ agents,
 $m$ resources,
 and $\tau$ periods.
 First,
 determine for each resource whether it is critical or uncritical.
 This can be done in polynomial time,
 as we only need to emulate the empty \xassignment[$\tau$]
 and~$\tau$ is unary-encoded.
 Then,
 apply \cref{rrule:critical}:
 if it returns a trivial \no-instance,
 output~$\no$.
 Now,
 assume we have at most~$n\cdot \tau$ many critical resources.
 Compute,
 in polynomial time,
 the fingerprint of each resource.
 Also,
 in~$2^{n\cdot \tau}\poly(|I|)$ time,
 compute the marking set of each.
 There are at most~$2^n\cdot 2^{2^{n\cdot \tau}}$ different fingerprint-marking-set combinations.
 After exhaustively applying~\cref{rrule:uncritical},
 there are at most~$n\cdot \tau$ resources for each such combination.
 In total,
 there are at most~$n\cdot \tau\cdot ( 1 + 2^n\cdot 2^{2^{n\cdot \tau}})$ many resources.
 Hence,
 \cref{obs:ntau:xp} gives the final part of the \FPT-time algorithm.
\end{proof}

\begin{theorem}[\appref{thm:ntau:nopk}]
 \label{thm:ntau:nopk}
 \UnlessPK,
 \fish{} admits no problem kernel of size polynomial in~$n+\tau$,
 even for unary encoding and weights in~$\{1,2,3\}$.
\end{theorem}

We give a polynomial-parameter transformation~\cite{BodlaenderThomasseYeo2011} from \prob{Set Cover},
where,
given a universe~$X$,
a set~$\calF$ of subsets of~$X$,
and an integer~$k\in\N$ with~$k\leq |\calF|$,
the question is whether there are at most~$k$ subsets in~$\calF$ such that their union equals~$X$.
\UnlessPK,
\prob{Set Cover} admits no problem kernel of size polynomial in~$|X|+k$
\cite{DomLokshtanovSaurabh2014}.

\begin{proof}[Proof idea of \cref{thm:ntau:nopk}]
 We have~$k$ periods.
 For each element~$x\in X$,
 there is an agent~$a_x$ and a resource~$\r_x$.
 Resource~$\r_x$ is only approved by~$a_x$,
 and can only be harvested at most~$k-1$ times by her.
 For each set~$F\in\calF$,
 there is a resource~$\r_F$, that, in a sustainable assignment,
 can be harvested only by a selection agent~$a^*$,
 and at most once.
Agent~$a^*$ approves all these resources and is forced
 by an additional agent-resource pair to harvest in each period.
 Now,
 $a_x$ approves~$\r_F$ if and only if~$x$ is \emph{not} an element of~$F$.
Hence,
 if~$a^*$ selects~$k$ sets none of which contains~$x$,
 $a_x$ values~$a^*$'s assignment at~$k$.
 Since~$a_x$ can obtain utility at most~$k-1$
 from%
 ~$\r_x$,
 she envies~$a^*$.
 Thus,
 any solution must ensure that at least one set is selected that contains~$x$.
 As this holds for all elements,
 the selected sets form a set cover.
\end{proof}

\appendixproof{thm:ntau:nopk}
{%
\begin{construction}
 \label{constr:ntau:nopk}
 Let~$I=(X,\calF,k)$ be an instance of~\prob{Set Cover}.
 Construct an instance~$I'$ of \fish{} as follows.
 Let~$\tau=k$.
 Set~$A=\{a_0\}\cup\{a^*\}\cup A_X$
 with~$A_X=\{a_x\mid x\in X\}$,
 such that~$a_0$ has weight 1,
 $a^*$ has weight 2,
 and~$a_x$ has weight 3 for every~$x\in X$.
 We have a resource~$\r_0$ with~$\iq_0=\ub_0=2$,
 $\lb_0=1$, and~$\rf_0(z)=2z$.
 Only~$a_0$ and~$a^*$ approve~$\r_0$.
 We have a set~$\rSet_\calF=\{\r_F\mid F\in\calF\}$,
 where every resource~$\r_F$ has~$\iq_F=\ub_F=3$,
 $\lb_F=1$, and~$\rf_F(z)=z$.
 Resource~$\r_F$ is approved by $a^*$ and all~$a_x$ with~$x\not\in F$.
 Finally,
 for each~$x\in X$,
 there is a resource~$\r_x$ with $\iq_x=\ub_x=3(k-1)+1$,
 $\lb_x=1$, and~$\rf_x(z)=z$.
 Resource~$\r_x$ is only approved by~$a_x$.
\end{construction}

\cref{constr:ntau:nopk} is polynomial-time executable and
since~$n+\tau=|X|+2+k$,
it is also parameter-preserving from~$|X|+k$ to~$n+\tau$.

\begin{observation}
 \label{obs:ntau:nopk:unavoidables}
 Let~$I'$ be a \yes-instance.
 In every solution,
 in every period,
 $a_0$ is assigned to~$\r_0$
 and
 $a^*$ is assigned to a resource from~$\rSet_\calF$.
 Moreover,
 if agent~$a_x$ is assigned,
 then only to~$\r_x$,
 and only in at most~$k-1$ periods.
\end{observation}

\begin{proof}
 The upper capacity and terminal viability constraints together with the fact that~$a_0$ is the only agent of weight one force $a_0$ to be assigned to~$\r_0$ every period.
 Since $a^*$ may envy~$a_0$ for each assignment to~$\r_0$,
 this forces $a^*$ to be assigned to a resource from~$\rSet_\calF$ in every period.
 Since the weight of each~$a_x$ is three,
 being assigned to any~$\r_F$ would violate terminal viability.
 Thus,
 if $a_x$ is assigned,
 then only to~$\r_x$,
 which permits at most $k-1$ harvests of weight three
 before violating terminal viability.
\end{proof}

\begin{proof}[Proof of \cref{thm:ntau:nopk}]
 Let~$I'$ be the instance of \fish{}
 obtained from an instance~$I=(X,\calF,k)$ of~\prob{Set Cover}
 using \cref{constr:ntau:nopk}.
 We prove that~$I$ is a \yes-instance if and only if~$I'$ is a \yes-instance.

 \RD{}
 Let~$\calF'\subseteq \calF$ be a solution to~$I$ of size exactly~$k$
 (which we can assume to exist by a standard padding argument).
 We construct an~\xassignment[$\tau$]~$\pi$ as follows.
 Let~$a_0$ be assigned to~$\r_0$ in every period.
Moreover, for each~$F\in\calF'$,
assign~$a^*$ to~$\r_F$ in a distinct period
 (hence, $a^*$ is assigned to a resource~$\r_F$ in every period).
 Let~$\gamma_x$ be the number of sets~$F\in \calF'$ such that~$x\not\in F$.
 Since~$\calF'$ is a set cover,
 we know that~$\gamma_x<k$.
 Finally,
 for every~$x\in X$,
 assign~$a_x$ to~$\r_x$ in exactly~$\gamma_x$ periods.
 We show that~$\pi$ is sustainable and envy-free,
 starting with the latter.
 Observe that~$a_0$ envies no one,
 and~$a^*$ can only envy~$a_0$,
 but both are assigned to an approved resource in each period.
 Consider~$a_x$.
 Agent~$a_x$ can only envy~$a^*$.
 Agent~$a^*$ is assigned in exactly~$\gamma_x$ periods to a resource they both approve,
 and~$a_x$ is assigned to~$\r_x$ in exactly~$\gamma_x$ periods,
 so no envy here either.
 Note that~$a_x$ and~$a_{x'}$ do not envy each other.
 Finally,
 we show that~$\pi$ is sustainable.
 Resource~$\r_0$ is harvested in each period by a weight of one,
 hence drops down to one and renews to two in the next period.
 Note that the quantities of all other resources are non-increasing over time,
 and hence only viability and no-over-extraction remain to be proved.
 Since only~$a^*$ with weight two is assigned to resources from~$\rSet_\calF$,
 and is assigned to no resource twice,
 all resources from~$\rSet_\calF$ have quantity between one and three in every period.
 Finally,
 since~$a_x$ is assigned only~$\gamma_x<k$ times to~$\r_x$,
 the quantity at $\tau+1$ is $3(k-1)+1-3\gamma_x\geq 3(k-1)+1-3(k-1)=1=\lb_x$.

 \LD{}
 Let~$\pi$ be a sustainable envy-free \xassignment[$\tau$].
 By~\cref{obs:ntau:nopk:unavoidables},
 we have that~$\sum_{t=1}^\tau v_{a_x}(\pi_t(a_x)) \leq k-1$.
 Moreover,
 $a^*$ is assigned in every period,
 and let~$\calF'\subseteq \calF$ be the set of all~$F\in \calF$ such that~$a^*$ is assigned to~$\r_F$ in some period.
 Since~$\pi$ is envy-free,
 $a^*$ must be assigned in at least one period to a resource~$\r_F$
 that~$a_x$ does not approve,
 that is,
 with~$x\in F$.
 Thus,
 $x$ is covered by~$F\in\calF'$.
Since this holds for every~$x\in X$ and~$|\calF'|\leq k$,
$\calF'$ is a solution to~$I$.
\end{proof}
}

\section{Few Agents and Resources: The Role of the Renewal Functions}

We next consider the combined parameter~$n+m$ and show that the complexity of \fish{} depends on the renewal functions and the encoding.
For arbitrary translations, \fish{} is in~\XP{} with respect to~$n+m$, even under binary encoding.

\begin{theorem}[\appref{thm:nm:translations}]
 \label{thm:nm:translations}
 If all renewal functions are translations,
 then \fish{} is \XP{} \wpb{} $n+m$.
\end{theorem}

\begin{proof}[Proof idea for \cref{thm:nm:translations}]
 Every sustainable envy-free \xassignment[$\tau$]~$\pi$ induces an \emph{access matrix}~$X\in\{0,1,\dots,\tau\}^{n\times m}$,
 where~$X_{i,j}$ represents how often agent~$a_i$ harvested resource~$\r_j$.
 Our goal is to guess~$X$ induced by a sustainable envy-free \xassignment[$\tau$].
 If all renewal functions are translations,
 then
 any resource's quantity at period~$t$ only depends on the cumulative harvest up to period~$t$,
 and not on the exact periods in which it is harvested.
 Thus,
 $X$ and a period~$t$ suffice to determine the current resource quantities,
 assuming that no constraints were violated so far.
We ensure the latter through dynamic programming:
a feasible access matrix for the previous period determines
the quantity of each resource,
so we can check all sustainability constraints for the next period.
Finally, $X$ also suffices to check envy-freeness,
since it records how often each agent harvested each resource.
\end{proof}

\appendixproof{thm:nm:translations}
{%
Let~$D[t,X]\in\{\true,\false\}$ with period~$t\in\{0,1,\dots,\tau\}$ and access matrix~$X\in\{0,1,\dots,\tau\}^{n\times m}$,
where~$X_{i,j}$ represents how often agent~$a_i$ harvested resource~$\r_j$.
Note that every~\xassignment[$t$]~$\pi=(\pi_1,\dots,\pi_t)$ \emph{induces} an access matrix~$X^\pi$ with~$X^{\pi}_{i,j}=|\{t'\in\set{t}\mid \pi_{t'}(a_i)=\r_j\}|$.
For translations,
$X$ and~$t$ suffice to derive a resource's quantity
assuming that no constraints were violated so far:
\[ q_j^X(t+1) = \iq_j + t\cdot \beta_j - \sum_{i=1}^n X_{i,j}\cdot w_{a_i} . \]
We call an access matrix~$X$ feasible at period~$t$ if and only if there is a \xassignment[$t$] inducing~$X$ that respects the constraints \eqref{model:noe}, \eqref{model:uc}, and \eqref{model:vp}.
Note that an access matrix can be induced by several such assignments,
yet the crucial insight is that for each~$\pi$ of those it holds that~$q_j^X(t+1)=q_j^{\pi}(t+1)$ for every~$j\in\set{m}$.
We set only~$D[0,\mathbf{0}]=\true$ and all other entries for~$t=0$ to~$\false$.
Set~$D[t,X]=\true$ if and only if there is an approval-restricted assignment~$\pi_t\colon A\to \rSet\cup\{\emptyset\}$ and an entry~$D[t-1,X']=\true$ such that
$\sum_{a\in \pi_t^{-1}(\r_j)} w_a \leq q_j^{X'}(t)$,
$q_j^{X'}(t)< \lb_j \implies \pi_t^{-1}(\r_j) = \emptyset$,
$q_j^{X}(t+1)\leq \ub_j$,
and
$X_{i,j} = X_{i,j}' + \Indic_{\pi_t(a_i)=\r_j}$,
where~$\Indic_{A}\in \{0,1\}$ denotes the indicator function such that~$\Indic_{A}=1$ if and only if boolean statement~$A$ is true.

\begin{algorithm}[t]
	\caption{Input is an instance~$I$ with an $n$-sized agent set~$A$, $m$-sized resource set~$\rSet$, and~$\tau$ periods.}
	\label{alg:nm:translations}
	Let~$T\gets\{0,1,\dots,\tau\}$\;
	initialize $D[0,(0)_{n\times m}]\gets\true$\;
	initialize all other DP entries to $\false$\;

	\For{$t\gets 1$ \KwTo $\tau$}{
			\ForEach{\nonwasteful{} $\pi_t\colon A \to \rSet\cup\{\emptyset\}$ \textbf{and} $X'\in T^{n\times m}$ with $D[t-1,X']=\true$}{
				Let~$X\in T^{n\times m}$ such that $X_{i,j} = X_{i,j}' + \Indic_{\pi_t(a_i)=\r_j}$ for all~$i\in\set{n},j\in\set{m}$\;
				\If{for all~$j\in\set{m}$ it holds that $\sum_{a\in \pi_t^{-1}(\r_j)} w_a \leq q_j^{X'}(t)$
						\textbf{and}
						$q_j^{X'}(t)< \lb_j \implies \pi_t^{-1}(\r_j) = \emptyset$
						\textbf{and}
						$q_j^{X}(t+1)\leq \ub_j$}
				{$D[t,X]\gets\true$}
		 }
	}

	\ForEach{$X\in T^{n\times m}$ with $D[\tau,X]=\true$}{
		\If{$q^X_j(\tau+1)\geq \lb_j \forall j\in\set{m}$ \textbf{and} $\forall\, i,i'\in\set{n}: \sum_{j=1}^m (X_{i,j}-X_{i',j})\cdot \val_{a_i}(\r_j) \geq 0$}{
			\Return{\yes}\;
		}
	}
	\Return{\no}\;
\end{algorithm}

\begin{lemma}
	\label{lem:nm:translations}
	For every~$t\in\{0,1,\dots,\tau\}$,
	it holds true that $D[t,X]=\true$ if and only if~$X$ is a feasible matrix at period~$t$.
\end{lemma}

\begin{proof}
 We prove the statement via induction on~$t$.
 By convention,
 the statement is true for~$t=0$:
 The empty $0$-period assignment respects all constraints
 and induces the all-zero matrix.
 Conversely,
 only the all-zero matrix is feasible for period~$0$.
 Now assume the statement holds true for~$t-1\geq 0$.

 \RD{}
 Let~$D[t,X]=\true$.
 Then,
 there is an approval-restricted assignment~$\pi_t\colon A\to \rSet\cup\{\emptyset\}$ and an entry~$D[t-1,X']=\true$ such that
 $\sum_{a\in \pi_t^{-1}(\r_j)} w_a \leq q_j^{X'}(t)$,
 $q_j^{X'}(t)< \lb_j \implies \pi_t^{-1}(\r_j) = \emptyset$,
 $q_j^{X}(t+1)\leq \ub_j$,
 and
 $X_{i,j} = X_{i,j}' + \Indic_{\pi_t(a_i)=\r_j}$.
 By induction,
 $X'$ is feasible at period~$t-1$ and hence,
 there is a \xassignment[$(t-1)$] $\pi'=(\pi_1,\dots,\pi_{t-1})$ that respects the constraints \eqref{model:noe}--\eqref{model:vp}
 and induces~$X'$.
 We claim that $\pi=(\pi_1,\dots,\pi_{t-1},\pi_t)$ is a \xassignment[$t$]
 that respects the constraints \eqref{model:noe}--\eqref{model:vp}
 and induces~$X$.
 Due to $\sum_{a\in \pi_t^{-1}(\r_j)} w_a \leq q_j^{X'}(t)$,
 there is no over-extraction.
 Due to $q_j^{X}(t+1)\leq \ub_j$,
 upper capacities are respected.
 Due to $q_j^{X'}(t) < \lb_j \implies \pi_t^{-1}(\r_j) = \emptyset$,
 viability protection is respected.
 Finally,
 due to $X_{i,j} = X_{i,j}' + \Indic_{\pi_t(a_i)=\r_j}$,
 $\pi$ induces~$X$.

 \LD{}
 Let~$X$ be a feasible matrix at period~$t$.
 By definition,
 there is a \xassignment[$t$]~$\pi=(\pi_1,\dots,\pi_t)$ inducing~$X$ that respects the constraints \eqref{model:noe}--\eqref{model:vp}.
 In particular,
 we have that~$q_j^{X}(t+1) = q_j^{\pi}(t+1)\leq \ub_j$.
 Note that the \xassignment[$(t-1)$]~$\pi'=(\pi_1,\dots,\pi_{t-1})$ also respects the constraints \eqref{model:noe}--\eqref{model:vp}.
 Let~$X'$ be the access matrix induced by~$\pi'$,
 and note that~$X_{i,j}' = X_{i,j} - \Indic_{\pi_t(a_i)=\r_j}$ for all~$i\in\set{n}$, $j\in\set{m}$.
 By induction,
 we have that~$D[t-1,X']=\true$.
 By the definition of~$\pi$, $X$, $\pi'$, and~$X'$,
 we have that
 $\sum_{a\in \pi_t^{-1}(\r_j)} w_a \leq q_j^{\pi}(t) = q_j^{\pi'}(t) = q_j^{X'}(t)$,
 as well as
 $q_j^{X'}(t) = q_j^{\pi'}(t) = q_j^{\pi}(t) < \lb_j \implies \pi_t^{-1}(\r_j) = \emptyset$.
 By convention,
 it follows that~$D[t,X]=\true$
 due to~$\pi_t$ and~$D[t-1,X']$.
\end{proof}

Finally,
we only need to check whether there is an access matrix~$X$ such that~$D[\tau,X]=\true$,
$\sum_{j=1}^m X_{i,j}\cdot \val_{a_i}(\r_j) \geq \sum_{j=1}^m X_{i',j}\cdot \val_{a_i}(\r_j)$ for all~$i,i'\in\set{n}$,
and $q_j^{X}(\tau+1)\geq \lb_j$ for all~$j\in\set{m}$.
The algorithm outputs \yes{} only in this case.
\cref{alg:nm:translations} outlines the complete algorithm.

\begin{proof}[Proof of \cref{thm:nm:translations}]
 We prove soundness,
 completeness,
 and \XP{} running time of \cref{alg:nm:translations}.

 \emph{Soundness}:
 Assume that the algorithm outputs~\yes.
 Then,
 $D[\tau,X]=\true$ for some~$X$ with
 $\sum_{j=1}^m X_{i,j}\cdot \val_{a_i}(\r_j) \geq \sum_{j=1}^m X_{i',j}\cdot \val_{a_i}(\r_j)$ for all~$i,i'\in\set{n}$,
 and $q_j^{X}(\tau+1)\geq \lb_j$ for all~$j\in\set{m}$.
 By \cref{lem:nm:translations},
 $X$ is a feasible matrix,
 and hence induces a \xassignment[$\tau$]~$\pi$
 that respects the constraints \eqref{model:noe}--\eqref{model:vp}.
 Since for~$\pi$ it holds that~$q_j^{\pi}(\tau+1)=q_j^{X}(\tau+1)\geq \lb_j$ for all~$j\in\set{m}$,
 also terminal viability is respected.
 It follows that~$\pi$ is sustainable.
 Finally,
 since
 $\sum_{t=1}^\tau \val_{a_i}(\pi_t(a_i)) = \sum_{j=1}^m X_{i,j}\cdot \val_{a_i}(\r_j) \geq \sum_{j=1}^m X_{i',j}\cdot \val_{a_i}(\r_j) = \sum_{t=1}^\tau \val_{a_i}(\pi_t(a_{i'}))$,
 it follows that~$\pi$ is envy-free.

 \emph{Completeness}:
 Let~$\pi=(\pi_1,\dots,\pi_\tau)$ be a sustainable envy-free \xassignment[$\tau$].
 Let~$X$ be the matrix induced by~$\pi$.
 By definition,
 we know that~$X$ is feasible at period~$\tau$.
 By \cref{lem:nm:translations},
 we know that~$D[\tau,X]=\true$.
 Since~$\pi$ is sustainable,
 we have that
 $q_j^{X}(\tau+1)=q_j^{\pi}(\tau+1)\geq \lb_j$ for all~$j\in\set{m}$.
 Finally,
 since $\pi$ is envy-free,
 we have that
 $\sum_{j=1}^m X_{i,j}\cdot \val_{a_i}(\r_j) = \sum_{t=1}^\tau \val_{a_i}(\pi_t(a_i)) \geq \sum_{t=1}^\tau \val_{a_i}(\pi_t(a_{i'})) = \sum_{j=1}^m X_{i',j}\cdot \val_{a_i}(\r_j)$.
 Thus,
 the algorithm outputs~\yes.

 \emph{Running time}:
 Table~$D$ has~$(\tau+1)\cdot (\tau+1)^{n\cdot m}$ entries.
Updating an entry requires computing
all possible~$(m+1)^n$ assignments
as well as corresponding
quantities,
which can be done in time polynomial in~$|I|$
even with binary encoding,
and compare it with at most~$(\tau+1)^{n\cdot m}$ entries.
The final check for each entry~$D[\tau,X]$ can also be done in time polynomial in~$|I|$.
It follows that the algorithm runs in \XP-time regarding~$n+m$.
\end{proof}
}%

However, for binary encoding, allowing multiplicative renewal already yields \NP-hardness for five agents and three resources, even when all agents approve all resources.

\begin{theorem}[\appref{thm:nm:pnph}]
\label{thm:nm:pnph}
\fish{} is \NP-hard even for five agents,
three resources,
and all renewal functions having integral~slope.
\end{theorem}

\begin{proof}[Proof idea for \cref{thm:nm:pnph}]
We reduce from an instance~$(G,k)$ of \prob{Vertex Cover}~\cite{DBLP:books/fm/GareyJ79}.
The idea is to encode the graph in the agent weights and use their assignments over time to select vertices. The same encoding agents are used repeatedly, with a designated period for each vertex.
The key is that renewal makes the effect of a harvest depend on its timing. The two resources used to encode have renewal function~$\rf(x)=Bx$ for an integer~$B>1$.
Hence, a unit harvested in period~$t$ reduces the final quantity by~$B^{\tau+1-t}$:
harvesting one period earlier costs a factor~$B$ more.
Thus, assignment times shift fixed digit patterns in base~$B$.

There are four encoding agents, labelled~$+,-,0,1$,
whose weights are scaled versions of encoding numbers~$d_+,d_-,d_0,d_1$. The agents~$+$ and~$-$ encode which vertices are selected,
while~$0$ and~$1$ provide padding to satisfy the remaining digit constraints.
Assigning agent~$a_\alpha$ to a target resource~$\r_{\mathrm{tar}}$ in period~$\tau-p$
contributes~$B^p d_\alpha$, up to a common scaling factor, to the encoded sum.
The goal is to make the sum of all contributions equal a target number~$T$,
whose digits encode the vertex-cover constraints.
The digit pattern~$P$ contained in~$d_+$
encodes the contributions of selected vertices to the corresponding checks.
The digit pattern~$T$ is illustrated in~\cref{fig:nm-t}.
\begin{figure*}[t]
\centering
\begin{tikzpicture}[
    x=0.825cm,y=1cm,>=latex,font=\footnotesize,
    digit/.style={
        rectangle,draw,very thick,fill=white,
        minimum width=0.36cm,
        minimum height=0.5cm,
        inner sep=1pt,font=\small
    },
    paddigit/.style={digit,fill=black!8},
    countdigit/.style={digit,fill=blue!12},
    vardigit/.style={digit}
]

\newcommand{\targetblock}[5]{%
    \begin{scope}[shift={(#1,0.56)}]
        \node[paddigit] at (-0.37,0) {$#2$};
        \node[countdigit] at (0,0) {$#3$};
        \node[vardigit] at (0.37,0) {$#4$};
    \end{scope}
    \node[below,font=\scriptsize] at (#1,-0.16) {$#5$};
}

\foreach \x/\name in {
    0.9/{vertex},
    3.2/{other},
    5.55/{size check},
    8.2/{edge $e_1$},
    11.25/{edge $e_2$}
}{
    \node at (\x,1.29) {\name};
    \draw[gray]
        (\x-0.72,0.96) -- (\x-0.72,1.06) --
        (\x+0.72,1.06) -- (\x+0.72,0.96);
}

\draw[ultra thick] (-0.05,0) -- (12.5,0);

\foreach \x in {0.9,3.2,5.55,8.2,11.25}{
    \draw[very thick] (\x,0.12) -- (\x,-0.12);
    \draw[densely dotted,thick,gray]
        (\x,0.12) -- (\x,0.34);
}

\node[anchor=east] at (-0.18,0.56) {$T$};

\targetblock{0.9}{1}{1}{1}{p_i}
\targetblock{3.2}{1}{1}{0}{p}
\targetblock{5.55}{0}{k}{0}{q_0}
\targetblock{8.2}{1}{2}{0}{q_1}
\targetblock{11.25}{1}{2}{0}{q_2}

\foreach \x in {2.05,4.35,6.95,9.7,12.3}{
    \node[text=gray,font=\small] at (\x,0.56) {$\cdots$};
}

\node[anchor=west,align=left,font=\scriptsize] at (13.1,1.15) {Digits \\ (left to right):};

\node[paddigit,anchor=west,minimum width=0.25cm,minimum height=0.3cm,label=0:{\scriptsize padding}] at (13.4,0.65) {};

\node[countdigit,anchor=west,minimum width=0.25cm,minimum height=0.3cm,label=0:{\scriptsize counting}] at (13.4,0.2) {};

\node[vardigit,anchor=west,minimum width=0.25cm,minimum height=0.3cm,label=0:{\scriptsize vertex control}] at (13.4,-0.25) {};

\end{tikzpicture}
\caption{Schematic representation of the target number~$T$.
Each base-$B$ digit consists of three base-$b$ digits:
padding, counting, and vertex control.
Dots indicate omitted positions.}
\label{fig:nm-t}
\end{figure*}
For each vertex~$v_i$, we reserve a position~$p_i$.
We also reserve a position~$q_0$ for checking the size of the selected set
and a position~$q_j$ for each edge~$e_j$.
Control digits in~$T$ require exactly one of~$a_+$ and~$a_-$
to be assigned to~$\r_{\mathrm{tar}}$ in each vertex period~$\tau-p_i$,
and neither agent to be assigned to~$\r_{\mathrm{tar}}$
in any other period.
The digit pattern~$P$ is illustrated in~\cref{fig:nm-p}.
\begin{figure*}[t]
\centering
\begin{tikzpicture}[
    x=1cm,y=1cm,>=latex,font=\footnotesize,
    digit/.style={rectangle,draw,very thick,fill=white,
        inner sep=1.75pt,font=\small},
    zerodigit/.style={digit,draw=gray!55,text=gray}
]
    \node[circle,draw,fill=white,inner sep=2pt] (v1) at (0.20-0.2,0.20-0.1) {$v_1$};
    \node[circle,draw,fill=white,inner sep=2pt] (v2) at (0.75-0.2,1.25-0.1) {$v_2$};
    \node[circle,draw,fill=white,inner sep=2pt] (v3) at (1.30-0.2,0.20-0.1) {$v_3$};
    \draw (v1) -- node[above left] {$e_1$} (v2);
    \draw (v2) -- node[above right] {$e_2$} (v3);
    \begin{scope}[xshift=2.65cm,x=0.825cm]
    \node at (1.8,1.39) {size check};
    \node at (6.66,1.39) {edge $e_1$};
    \node at (11.52,1.39) {edge $e_2$};
    \foreach \l/\r in {0.35/3.25,5.21/8.11,10.07/12.97}{
        \draw[gray] (\l,1.06) -- (\l,1.16) -- (\r,1.16) -- (\r,1.06);
    }
    \draw[-,ultra thick] (-0.05,0) -- (13.1,0);
    \foreach \x in {0.6,1.86,5.46,6.72,10.32,11.58}
        \draw (\x,0.08) -- (\x,-0.08);
    \foreach \x in {3,7.86,12.72}
        \draw[very thick] (\x,0.14) -- (\x,-0.14);
    \foreach \x in {0.6,1.86,3,5.46,6.72,7.86,10.32,11.58,12.72}
        \draw[densely dotted,thick,gray] (\x,0.08) -- (\x,0.44);
    \node[anchor=east] at (-0.22,0.56) {$P$};
    \foreach \x in {0.12,1.23,2.43,6.09,7.29,10.95,12.15}{
        \node[text=gray,font=\small] at (\x,0.56) {$\cdots$};
    }
    \foreach \x in {4.23,9.09}
        \node[text=gray,font=\small] at (\x,0.56) {$0\cdots0$};
    \foreach \x/\v in {0.6/3,1.86/2,3/1,6.72/2,7.86/1,10.32/3,11.58/2}{
        \node[digit] at (\x,0.56) {$1$};
        \node[above] at (\x,0.77) {$v_\v$};
    }
    \foreach \x/\v in {5.46/3,12.72/1}{
        \node[zerodigit] at (\x,0.56) {$0$};
        \node[above,text=gray] at (\x,0.77) {$v_\v$};
    }
    \foreach \x/\p in {
        0.6/{q_0{-}p_3},1.86/{q_0{-}p_2},3/{q_0{-}p_1},
        5.46/{q_1{-}p_3},6.72/{q_1{-}p_2},7.86/{q_1{-}p_1},
        10.32/{q_2{-}p_3},11.58/{q_2{-}p_2},12.72/{q_2{-}p_1}}{
        \node[below,font=\scriptsize] at (\x,-0.16) {$\p$};
    }
    \end{scope}
\end{tikzpicture}
\caption{The base-$B$ encoding~$P$ for the graph on the left.
An entry at~$q_0-p_i$ counts~$v_i$,
while an entry at~$q_j-p_i$ records that~$v_i$ is incident to~$e_j$.
Dots represent runs of zero digits.}
\label{fig:nm-p}
\end{figure*}
For each vertex~$v_i$, $P$ contains a~$1$ at~$q_0-p_i$ and at~$q_j-p_i$ for each incident edge~$e_j$.
Assigning~$a_+$ to~$\r_{\mathrm{tar}}$ in period~$\tau-p_i$
selects~$v_i$ and shifts~$P$ by~$p_i$ digit positions.
Since~$(q_j-p_i)+p_i=q_j$, the entries corresponding to~$v_i$ reach their check positions,
contributing one to the size check and to each incident edge check.
For example, selecting~$v_2$ in~\cref{fig:nm-p}
moves the three corresponding entries to~$q_0,q_1,q_2$.
Assigning~$a_-$ to~$\r_{\mathrm{tar}}$ in period~$\tau-p_i$
leaves~$v_i$ unselected and makes no contribution to these checks.
The target requires the size-check digit to equal~$k$,
with no padding allowed there.
Each edge-check digit must equal two.
Assigning~$a_0$ or~$a_1$ to~$\r_{\mathrm{tar}}$
in period~$\tau-q_j$ adds zero or one, respectively,
to the edge-check counting digit.
Thus, this digit can reach two exactly when at least one endpoint is selected.
Assigning~$a_+$ to~$\r_{\mathrm{tar}}$ in period~$\tau-p_k$
also shifts entries associated with other vertices~$v_i$
to digit positions~$q_j+(p_k-p_i)$.
We choose the vertex positions so that every nonzero difference~$p_k-p_i$ is unique, and separate the regions around the different checks.
Thus, these unintended contributions never reach a check position, and at most one reaches any other position. At each such position, the target counting digit is one, so assigning~$a_0$ or~$a_1$ in the corresponding period always completes it.
A sufficiently large base prevents carries.
Thus, the padding agents can be assigned so that the encoded harvest sum equals~$T$ exactly when the selected vertices form a vertex cover of size~$k$.

To enforce this equality through sustainability, we introduce two further
resources~$\r_{\mathrm{rest}}$ and~$\r_*$, and a fifth agent.
Only this agent can harvest~$\r_*$.
The resource starts at its capacity, and this agent's harvest restores
it to that capacity after renewal.
Without harvesting, renewal would exceed the capacity.
Thus, the fifth agent must be assigned to~$\r_*$
in every period.
Since the four encoding agents also approve~$\r_*$,
envy-freeness forces each of them to be assigned in every period.
They must therefore split between~$\r_{\mathrm{tar}}$
and~$\r_{\mathrm{rest}}$.

Let~$S_p$ contain the labels of the agents assigned to~$\r_{\mathrm{tar}}$ in period~$\tau-p$.
After scaling the weights, complementary initial quantities ensure that
terminal viability of~$\r_{\mathrm{tar}}$ requires
$\sum_{p=0}^{\tau-1}B^p\sum_{\alpha\in S_p}d_\alpha\leq T$,
while terminal viability of~$\r_{\mathrm{rest}}$ requires the reverse inequality.
Thus, sustainability forces equality and hence a vertex cover.
\end{proof}

\appendixproof{thm:nm:pnph}{ We present a reduction from \textsc{Vertex Cover} \cite{DBLP:books/fm/GareyJ79}.
The idea is to encode the graph in the numerical parameters and use the multi-period allocation to represent a vertex cover. The key is that renewal makes the effect of a harvest depend on its timing. For the two encoding resources, we use renewal function~$\rf(x)=Bx$ for an integer~$B>1$. Hence a unit harvested in period~$t$ reduces the final quantity by $B^{\tau+1-t}$, i.e. harvesting one period earlier costs a factor~$B$ more. We therefore read period $\tau-p$ as digit position~$p$ of a base-$B$ number.
Each encoding agent represents a fixed digit pattern.
Assigning it to our target resource in period~$\tau-p$ shifts this pattern
by~$p$ positions in the encoded sum.
The construction fixes a target number~$T$ such that the encoded sum can equal~$T$ if and only if $G$ has a vertex cover of size~$k$. %

Informally, at each vertex position~$p_i$, the labels $+$ and $-$ determine whether~$v_i$ is selected.
The position~$q_0$ checks that exactly~$k$ vertices are selected, while each position~$q_j$ checks that edge~$e_j$ has at least one selected endpoint. The padding labels complete the remaining digits whenever these conditions hold.
The encoding is then translated into the allocation instance.
In period~$\tau-p$, the agents assigned to~$\r_{\mathrm{tar}}$ determine the set~$S_p$, while the remaining agents use~$\r_{\mathrm{rest}}$. Resource~$\r_*$ together with envy-freeness forces all agents to be assigned in every period.
The initial quantities of $\r_{\mathrm{tar}}$ and $\r_{\mathrm{rest}}$ are chosen so that sustainability forces
$\sum_{p=0}^{\tau-1} B^p\sum_{\alpha\in S_p}d_\alpha=T$.
By the encoding, this holds exactly when the selected vertices form a vertex cover of size at most~$k$.

\begin{construction}\label{constr:nm}
Let~$I = (G=(V,E),k)$ be an instance of \prob{Vertex Cover}, where
$V=\{v_1,\ldots,v_N\}$ and $E=\{e_1,\ldots,e_M\}$. We may assume that $k<N$.

For every vertex~$v_i$, define $p_i\coloneqq 2N^2(i-1)+(i-1)^2$. Let $\tau\coloneqq (M+2)(2p_N+1)$. Since~$M\ge0$ and~$p_N\ge0$, we have~$\tau\ge2$.
We refer to the $p_i$ as \emph{positions}, where a position is an integer $p\in\{0,\dots,\tau-1\}$ that is the exponent in~$B^p$ and corresponds to period~$\tau-p$.
The positions~$p_i$ will encode whether the corresponding vertices are selected for the vertex cover. They are chosen such that every nonzero difference $p_k-p_i$ uniquely determines $i$ and $k$. Indeed, $p_k-p_i=2N^2(k-i)+(k-i)(k+i-2)$ where $0<|(k-i)(k+i-2)|<2N^2$. Since the second term has absolute value less than~$2N^2$ and has the same sign as~$k-i$, the difference~$p_k-p_i$ uniquely determines~$k-i$. Subtracting~$2N^2(k-i)$ then determines~$(k-i)(k+i-2)$, and hence~$k+i$.
We additionally introduce one position~$q_0$ for checking the size of the selected set and one position~$q_j$ for every edge~$e_j$, such that $q_j\coloneqq (j+1)(2p_N+1)+p_N$ for all $j\in\{0,\ldots,M\}$. This choice ensures that the intervals
$[q_j-p_N,q_j+p_N]$ are pairwise disjoint, lie strictly above all
positions~$p_i$, and are contained in~$\{0,\dots,\tau-1\}$.
Indeed, $q_0-p_N=2p_N+1>p_N$ and $q_M+p_N=\tau-1$.

Let $b\coloneqq N+2M+3$ and $B\coloneqq b^3$. Thus, every base-$B$ digit, which we call a \emph{block}, consists of three base-$b$ digits, where we count digits starting from the least significant one.
We define
\[ P\coloneqq
\sum_{i=1}^N B^{q_0-p_i}
+
\sum_{j=1}^M
\sum_{\substack{i\in[N]\\v_i\in e_j}}
B^{q_j-p_i}.\] The exponent $q_j-p_i$ is chosen so that multiplying by $B^{p_i}$, i.e., selecting $+$ at position $p_i$, moves this term exactly to the check position $q_j$; selecting $+$ at another position $p_k$ moves it to $q_j+(p_k-p_i)\neq q_j$. The first sum thus encodes the number of selected vertices at position~$q_0$, while the second encodes whether every edge has a selected endpoint.

We call $+,-,0,1$ \emph{labels}. For each label~$\alpha$, we define an encoding number
\[
d_+\coloneqq 1+bP,
\quad
d_-\coloneqq 1,
\quad
d_0\coloneqq b^2,
\quad
d_1\coloneqq b^2+b.
\] Each label~$\alpha$ will also have a corresponding agent~$a_\alpha$, which we define below. The labels~$+$ and~$-$ correspond to selecting and not selecting a vertex, respectively, while the labels~$0$ and~$1$ provide padding. More precisely, at every position except~$q_0$, exactly one padding label is used: $d_0$ adding nothing and $d_1$ adding one to the relevant count.
We refer to the three relevant base-$b$ digits as the \emph{variable}, \emph{counting}, and \emph{padding} digits: the first encodes the vertex choice, the second records the contribution to a check, and the third is used to enforce the padding structure.

In terms of base-$b$ digits, $d_-$ sets only the variable digit, $d_0$ only the padding digit, and $d_1$ both the counting and padding digits. When $d_+$ is used at position~$p$ for vertex~$v_i$, it sets the variable digit at position~$p$ and, through~$P$, the counting digits at the positions $p+(q_j-p_i)$ corresponding to the relevant checks.

For every $p\in\{0,\ldots,\tau-1\}$, let
\[
o_p\coloneqq
\begin{cases}
1 & \text{if }p\in\{p_1,\ldots,p_N\},\\
0 & \text{otherwise},
\end{cases}
\quad
c_p\coloneqq
\begin{cases}
k & \text{if }p=q_0,\\
2 & \text{if }p\in\{q_1,\ldots,q_M\},\\
1 & \text{otherwise}.
\end{cases}
\]
Define $T\coloneqq
\sum_{p=0}^{\tau-1}(o_p+bc_p+b^2)B^p-b^2B^{q_0}$, where the subtracted term sets the padding digit at~$q_0$ to~$0$.
Thus, at~$q_0$, the padding digit is~$0$, while it is~$1$ everywhere else.
See \cref{fig:nm-t} for a schematic representation.

We next define the agents and resources.
Let
$D\coloneqq d_++d_-+d_0+d_1$ and
$\lambda\coloneqq B^{\tau-1}$.
There are five agents
$A\coloneqq\{a_+,a_-,a_0,a_1,a_*\}$
with weights
\[
w_{a_+}=\lambda d_+,\quad
w_{a_-}=\lambda d_-,\quad
w_{a_z}=\lambda d_z \ (z\in\{0,1\}),\quad
w_{a_*}=1.
\] We refer to $a_+,a_-,a_0,a_1$ as the encoding agents.
Every agent approves every resource.

There are three resources
$\rSet=\{\r_{\mathrm{tar}},\r_{\mathrm{rest}},\r_*\}$.
The resources~$\r_{\mathrm{tar}}$ and~$\r_{\mathrm{rest}}$ have
renewal function~$\rf(x)=Bx$, while~$\r_*$ has renewal function~$\rf(x)=2x$.
The initial quantities and bounds are
\begin{align*}
\begin{array}{c|ccc}
    \text{Resource} & \iq & \lb & \ub\\\hline
    \r_{\mathrm{tar}}
        & T
        & 0
        & 2\lambda D\\
    \r_{\mathrm{rest}}
        & D\sum_{p=0}^{\tau-1}B^p-T
        & 0
        & 2\lambda D\\
    \r_*
        & 2
        & 0
        & 2
\end{array}
\end{align*}
The initial quantities are nonnegative and respect the upper bounds.

The factor~$B^{\tau-1}$ accounts for the~$\tau-1$ renewals between the first and last harvest: harvesting~$\lambda d_\alpha$ in the last period requires~$d_\alpha$ units of initial stock.
The same harvest~$p$ periods earlier requires~$B^p d_\alpha$ units, corresponding to a contribution of~$B^p d_\alpha$ to~$T$.
Since~$\tau\ge2$, $B>2$, and every~$d_\alpha\ge1$, each encoding agent has weight~$\lambda d_\alpha\ge B>2$, so only~$a_*$ can be assigned to~$\r_*$.

As shown below, every sustainable envy-free assignment forces the four encoding agents
to harvest a total weight of~$\lambda D$ from~$\r_{\mathrm{tar}}$ and~$\r_{\mathrm{rest}}$ in every period. The initial quantity of~$\r_{\mathrm{rest}}$ is chosen as the complement of that of~$\r_{\mathrm{tar}}$, turning the two sustainability inequalities into an equality. The resource~$\r_*$ is urgent in every period: without harvesting, its quantity would grow from~$2$ to~$4>2$.
\end{construction}

Since $p_N=O(N^3)$ and $\tau=O((M+2)\cdot N^3)$, all exponents are polynomially bounded, and all constructed numbers have polynomial binary encoding length. In particular, $\tau$ itself is polynomially bounded and thus has polynomial unary encoding length. Only the constructed weights and resource quantities may be exponentially large in value, which is where binary encoding is required. Hence, the construction can be computed in polynomial time.

\begin{lemma}\label{lem:nm:arithmetic}
For the numbers defined in \Cref{constr:nm}, $G$ has a vertex cover of size at most~$k$
if and only if there exist label subsets
$S_0,\ldots,S_{\tau-1}\subseteq\{+,-,0,1\}$ such that
$\sum_{p=0}^{\tau-1}
B^p\sum_{\alpha\in S_p}d_\alpha=T$. For $\alpha\in S_p$ we then say that label $\alpha$ is \emph{selected} at position $p$.
\end{lemma}

\begin{proof}
Selecting label~$\alpha$ at position~$p$ adds a value of~$B^pd_\alpha$. Hence, a term~$B^{q_j-p_i}$ in~$P$ contributes at position $q_j-p_i+p$. In particular, if~$+$ is selected at position~$p_i$, this term contributes exactly at~$q_j$.

We first observe that no carries occur, even for arbitrary
choices of the subsets $S_p$.
Indeed, the first and third digits receive at most two units (the first from $+$ and $-$, the third from $0$ and $1$).
Moreover, $P$ contains~$N+2M$ terms, and every such term contributes to any position at most once, since a term $B^{q_j-p_i}$ is shifted to position $q_j-p_i+p$ for $+\in S_p$, and different values of $p$ produce different positions.
Together with the padding contribution, the second digit is therefore
at most~$N+2M+1<b$.
Consequently, $\sum_{p=0}^{\tau-1} B^p \sum_{\alpha\in S_p} d_\alpha = T$ if and only if the total contribution at each base-$b$ position equals the corresponding digit of $T$.

If the sum equals $T$, the variable-control digits force exactly one of $+$ and $-$ to belong to $S_{p_i}$ for each vertex $v_i$,
and neither to belong to $S_p$ at any other position.
We interpret $+\in S_{p_i}$ as selecting $v_i$ and $-\in S_{p_i}$ as not selecting it.
The padding-control digits force exactly one label $z_p\in\{0,1\}$ to belong to each $S_p$ except for $p=q_0$,
where neither padding label can be selected. Its contribution to the counting digit is $z_p$.

We next examine where the contributions of selected vertices appear.
Recall that the encoding of~$+$ contains a term at position $q_0-p_i$ for every vertex~$v_i$, and a term at position $q_j-p_i$ whenever $v_i$ is incident to edge~$e_j$.
If $+$ is selected at position~$p_k$, then $B^{p_k}$ shifts such a term
to position
$q_j-p_i+p_k=q_j+(p_k-p_i)$.
Intuitively, $p_i$ is the position at which vertex~$v_i$ is selected,
while $q_0$ checks the number of selected vertices and $q_j$ checks
whether edge~$e_j$ is covered.
In particular, a term corresponding to~$v_i$ contributes at its
intended check position precisely when $+$ is selected at~$p_i$.

All contributions associated with a check position~$q_j$ lie in $[q_j-p_N,q_j+p_N]$.
By the choice of the positions~$q_0,\ldots,q_M$, these ranges are pairwise disjoint and contained in $[0,\tau-1]$.

Consider first the size-check position~$q_0$.
Every selected vertex contributes exactly one unit at~$q_0$.
Since no padding label can be selected at this position, the counting
digit equals~$k$ if and only if exactly~$k$ vertices
are selected.

Next, consider the position~$q_j$ corresponding to an edge~$e_j$.
A contribution occurs at~$q_j$ exactly for each selected endpoint of
$e_j$.
Thus, the contribution at~$q_j$ equals the number
$t_j\in\{0,1,2\}$ of selected endpoints of~$e_j$.
The counting digit equals~$2$ exactly when $t_j+z_{q_j}=2$.
Since $z_{q_j}\in\{0,1\}$, this equality is possible if and only if
$t_j\geq1$, that is, if and only if~$e_j$ has a selected endpoint.

It remains to consider the remaining positions in these ranges.
A contribution at position~$q_j+\delta$ with~$\delta\neq 0$ arises from a pair~$(i,k)$ satisfying $p_k-p_i=\delta$ and~$+\in S_{p_k}$.
By uniqueness of the nonzero differences~$p_k-p_i$, there is at most one such pair and hence at most one contribution.
Its value is therefore either zero or one.
As the target counting digit at such a position is one, it can always
be completed by choosing padding one or zero, respectively.

Suppose now that
$\sum_{p=0}^{\tau-1} B^p
    \sum_{\alpha\in S_p} d_\alpha = T$.
The variable-control digits define a set of exactly~$k$ selected
vertices, and the equality $t_j+z_{q_j}=2$ implies that every edge has
a selected endpoint.
Hence, the selected vertices form a vertex cover of~$G$.

Conversely, suppose that~$G$ has a vertex cover~$K$ of size at most~$k$.
Since $k<N$, we can add arbitrary vertices until $|K|=k$.
At each position~$p_i$, select~$+$ if $v_i\in K$ and~$-$ otherwise;
select neither at other positions.
At~$q_0$, select no padding label.
At each edge position~$q_j$, select the padding label
$z_{q_j}=2-t_j$,
which belongs to $\{0,1\}$ since~$K$ is a vertex cover.
At every other position, select padding zero if there is a contribution
and padding one otherwise.
All three types of digits then agree with~$T$.
Moreover, all contributions lie within $[0,\tau-1]$.
Thus,
$\sum_{p=0}^{\tau-1} B^p
    \sum_{\alpha\in S_p} d_\alpha = T$.
\end{proof}

With \cref{lem:nm:arithmetic} we can now prove \cref{thm:nm:pnph}. The main idea is that the two resources $\r_{\mathrm{tar}}$ and
$\r_{\mathrm{rest}}$ encode a partition of the four agents
$a_+,a_-,a_0,a_1$ in every period.
Resource~$\r_*$ forces agent~$a_*$ to be assigned in every period, and
envy-freeness then forces all remaining agents to be assigned as well.
The initial quantities of $\r_{\mathrm{tar}}$ and
$\r_{\mathrm{rest}}$ are chosen so that their combined quantity is
exactly sufficient for all these assignments over the whole horizon.
Hence, any sustainable assignment determines subsets~$S_p$ whose
weighted sum must equal~$T$ exactly.
By the encoding, this equality holds precisely when the
choices encoded by the agents correspond to a vertex cover of size at
most~$k$.

\begin{proof}[Proof of \cref{thm:nm:pnph}]
First, suppose that the constructed instance admits a sustainable envy-free assignment $\pi$.
Agent~$a_*$ must be assigned to~$\r_*$ in every period.
Indeed, $\r_*$ has initial quantity and upper bound~$2$ and renewal function~$\rf(x)=2x$.
Whenever its current quantity is~$2$, sustainability requires a harvest~$h$ satisfying~$2(2-h)\le2$, equivalently~$h\ge1$.
Since~$w_{a_*}=1$, while every encoding agent has weight greater than~$2$, only~$a_*$ can provide such a harvest without violating no over-extraction.
Assigning~$a_*$ leaves quantity~$1$ before renewal and hence restores the resource to~$2$ for the next period.
Starting from its initial quantity~$2$, induction therefore forces~$a_*$ to be assigned to~$\r_*$ in every period.
Thus, $a_*$ must be assigned to~$\r_*$ in every period.

Since every agent approves every resource and agent $a_*$ is assigned $\tau$ times, envy-freeness now implies that every other agent must also be assigned in every period; otherwise, that agent would receive fewer than~$\tau$ approved assignments and would envy~$a_*$.
Consequently, each of $a_+,a_-,a_0,a_1$ is assigned in every period to exactly one of $\r_{\mathrm{tar}}$ and $\r_{\mathrm{rest}}$.

For $p\in\{0,\ldots,\tau-1\}$, let $S_p\subseteq\{+,-,0,1\}$ contain exactly the labels of the agents assigned to~$\r_{\mathrm{tar}}$ in period~$\tau-p$. The total assigned weight of~$\r_{\mathrm{tar}}$ in this period is therefore $\lambda\sum_{\alpha\in S_p}d_\alpha$. Since $\q^\pi_{\mathrm{tar}}(t+1)=B(\q^\pi_{\mathrm{tar}}(t)-h^\pi_{\mathrm{tar}}(t))$ this gives \[\q^\pi_{\mathrm{tar}}(\tau+1)=B^\tau\iq_{\mathrm{tar}}-\sum_{t=1}^\tau B^{\tau+1-t}h^\pi_{\mathrm{tar}}(t)\] Substituting $h^\pi_{\mathrm{tar}}(\tau-p)=\lambda\sum_{\alpha\in S_p}d_\alpha$, $\lambda=B^{\tau-1}$ and $\iq_{\mathrm{tar}}=T$ we get \[ \q^\pi_{\mathrm{tar}}(\tau+1)=B^\tau\left(T-\sum_{p=0}^{\tau-1}B^p\sum_{\alpha\in S_p}d_\alpha\right).\] Hence, terminal viability of~$\r_{\mathrm{tar}}$, that is, $\q^\pi_{\mathrm{tar}}(\tau+1)\geq 0$, implies $\sum_{p=0}^{\tau-1}B^p\sum_{\alpha\in S_p}d_\alpha\leq T$.
Similarly, since all agents in $\{a_+,a_-,a_0,a_1\}$ not assigned to~$\r_{\mathrm{tar}}$ are assigned to~$\r_{\mathrm{rest}}$, the total assigned weight of~$\r_{\mathrm{rest}}$ in period~$\tau-p$ is $\lambda\left(D-\sum_{\alpha\in S_p}d_\alpha\right)$. Applying the same calculation to~$\r_{\mathrm{rest}}$ gives $\sum_{p=0}^{\tau-1}B^p\left(D-\sum_{\alpha\in S_p}d_\alpha\right)\leq D\sum_{p=0}^{\tau-1}B^p-T$.
The two left-hand sides sum to $D\sum_{p=0}^{\tau-1}B^p$, as do the two right-hand sides. Hence, both inequalities hold with equality, and therefore $\sum_{p=0}^{\tau-1}B^p\sum_{\alpha\in S_p}d_\alpha=T$. By \Cref{lem:nm:arithmetic}, $G$ has a vertex cover of size at most~$k$.

Conversely, suppose that $G$ has a vertex cover of size at most~$k$.
By \Cref{lem:nm:arithmetic}, there exist subsets
$S_0,\ldots,S_{\tau-1}\subseteq\{+,-,0,1\}$ such that
$\sum_{p=0}^{\tau-1}
B^p\sum_{\alpha\in S_p}d_\alpha=T$.

In period~$\tau-p$, assign to~$\r_{\mathrm{tar}}$ exactly the agents
whose labels belong to~$S_p$, and assign the remaining agents in
$\{a_+,a_-,a_0,a_1\}$ to~$\r_{\mathrm{rest}}$.
Assign~$a_*$ to~$\r_*$ in every period.

The equality above and the same calculation as in the forward direction give $\q^\pi_{\mathrm{tar}}(\tau+1)=0$.
Moreover,
$\sum_{p=0}^{\tau-1}
B^p\left(D-\sum_{\alpha\in S_p}d_\alpha\right)
=
D\sum_{p=0}^{\tau-1}B^p-T$, so~$\q^\pi_{\mathrm{rest}}(\tau+1)=0$ as well.

For either index~$i\in\{\mathrm{tar},\mathrm{rest}\}$, rewriting the renewal recurrence gives
$\q_i^\pi(t)=h_i^\pi(t)+\q_i^\pi(t+1)/B$.
Since~$\q_i^\pi(\tau+1)=0$, backward induction gives
$\q_i^\pi(t)
=
\sum_{s=t}^{\tau}B^{t-s}h_i^\pi(s)$
for every~$t\in\{1,\ldots,\tau+1\}$, where an empty sum equals zero.
Thus, the quantity available in each period covers that period's harvest and all later harvests, discounted by their intervening renewals.
In particular, for every~$t\in\{1,\ldots,\tau\}$,
$\q_i^\pi(t)-h_i^\pi(t)
=
\sum_{s=t+1}^{\tau}B^{t-s}h_i^\pi(s)
\ge 0$.
Hence, no over-extraction occurs in any period.
Furthermore, since~$h_i^\pi(s)\le\lambda D$, the same formula gives
\[
\begin{aligned}
\q_i^\pi(t)
\le\lambda D\sum_{j=0}^{\tau-1}B^{-j}
<\lambda D\frac{B}{B-1}<2\lambda D=\ub_i
\end{aligned}
\]
for every~$t\in\{1,\ldots,\tau+1\}$, including the initial quantity.
Both lower bounds are zero, so viability protection and terminal viability also hold.
For~$\r_*$, assigning~$a_*$ leaves quantity~$1$, which renews to~$2$
for the next period.
Thus, the assignment is sustainable.

Finally, every agent is assigned an approved resource in every period,
so the assignment is envy-free.
\end{proof} }

For unary encoding and integral slopes, we obtain an~\XP{} algorithm for~$n+m$.

\begin{theorem}[\appref{thm:nm:integralslope}]
 \label{thm:nm:integralslope}
 With unary encoding and all renewal functions having integral slope,
 \fish{} is \XP{} \wpb{} $n+m$.
\end{theorem}

\begin{proof}[Proof idea of \cref{thm:nm:integralslope}]
 Consider an instance~$I$ with unary encoding and all renewal functions having integral slope.
 First,
 we prove that~$I$ can be transformed in \XP-time regarding~$n+m$ into an equivalent instance~$I'$ with all numbers being integral and bounded by~$|I|^{O(n+m)}$.
 This transformation multiplies all agent weights, initial quantities,
bounds, and renewal intercepts by the least common multiple
of their denominators, leaving the slopes unchanged.
Since there are~$O(n+m)$ such denominators,
each at most~$|I|$ under unary encoding,
the scaled numbers have absolute values bounded by~$|I|^{O(n+m)}$.

We then show that \fish{} is \XP{} \wpb{} $n+m$
when all input numbers are integral and encoded in unary.
 With all input numbers being integral,
 the quantity of each resource~$\r_j$ in any sustainable assignment is integral and in~$\{0,1,\dots,\ub_j\}$ in each period.
 Due to the unary encoding,
 we can hence store all possible resource quantities for each period.
 In particular,
 a state at a period consists of all resources' quantities and
 an \emph{envy matrix} from $\{-\tau,\dots,\tau\}^{n^2}$,
 where for each pair of agents~$(a,a')$, we store the difference between how~$a$ values her own assignments and those of~$a'$ so far.
We use dynamic programming over all periods and states.
 For each period and valid state,
 we can go through all~$(m+1)^n$ possible next assignments.
 Whenever the assignment respects the sustainability conditions,
 we can
 update all quantities
 and valuation differences
 for the next period.
 Finally, after~$\tau$ periods, we accept if a reachable state
satisfies terminal viability and all entries of its envy matrix
are nonnegative.
 Thus,
 $I'$ is solvable in~\XP-time regarding~$n+m$,
 and so is~$I$.
\end{proof}

\appendixproof{thm:nm:integralslope}
{%
The proof of \cref{thm:nm:integralslope} uses two main ingredients.
First, we give an~\XP{} algorithm parameterized by~$n+m$
when all numbers in the input instance are integral and encoded in unary.
Second, we transform a unary-encoded instance with integral slopes
into an equivalent integral instance of unary encoding length
$|I|^{O(n+m)}$.

\begin{proposition}
 \label{thm:nm:integral}
 If all input numbers are integral and encoded in unary,
 then \fish{} is solvable in~$|I|^{f(n,m)}$ time
 for some function~$f$ only depending on~$n$ and~$m$.
\end{proposition}

Let~$D[t, \q_1,\dots,\q_m,(\eps_{a,a'})_{a,a'\in A}] \in\{\true,\false\}$
with entries from $\{0,\dots,\tau\}\times
\{0,\dots,u_1\}\times\dots\times\{0,\dots,u_m\}\times
\{-\tau,\dots,\tau\}^{n^2}$. We show that it is true if and only if
there is a \xassignment[$t$]~$\pi$ that results in envy-differences~$\eps_{a,a'}=\sum_{t'=1}^t (v_a(\pi_{t'}(a)) - v_a(\pi_{t'}(a')))$
and quantities~$\q_1,\dots,\q_m$ \emph{after} period~$t$.
We set~$D[t,\q_1,\dots,\q_m,(\eps_{a,a'})_{a,a'\in A}]=\true$
if and only if
there is an approval-restricted assignment~$\pi\colon A\to \rSet\cup\{\emptyset\}$
and an entry
$D[t-1,\q_1',\dots,\q_m',(\eps_{a,a'}')_{a,a'\in A}]=\true$
such that
for all~$i\in\set{m}$,
$q_i'\geq \sum_{a\in \pi^{-1}(\r_i)} w_a$,
$\q_i=\rf_i(q_i'-\sum_{a\in \pi^{-1}(\r_i)} w_a)$,
and
$\sum_{a\in \pi^{-1}(\r_i)} w_a>0$ only if~$q_i'\geq \lb_i$,
and for all~$a,a'\in A$,
$\eps_{a,a'}=\eps_{a,a'}' + v_a(\pi(a)) - v_a(\pi(a'))$.
Initially,
we set only~$D[0,\iq_1,\dots,\iq_m,\vec{0}]=\true$ and all remaining entries with~$t=0$ to~$\false$.

\begin{lemma}
 \label{lem:nm:integral}
 $D[t, \q_1,\dots,\q_m,(\eps_{a,a'})_{a,a'\in A}]=\true$ if and only if
 there is a \xassignment[$t$]~$\pi$ that respects \eqref{model:noe}--\eqref{model:vp},
 results in envy-differences~$\eps_{a,a'}=\sum_{t'=1}^t (v_a(\pi_{t'}(a)) - v_a(\pi_{t'}(a')))$,
 and yields quantities~$\q_1,\dots,\q_m$ after period~$t$.
\end{lemma}

\begin{proof}
 We prove the statement via induction on~$t$.
 For~$t=0$,
 the statement is trivially clear.
 So assume that for~$t-1\geq 0$, the statement holds.
 Consider~$D[t, \q_1,\dots,\q_m,(\eps_{a,a'})_{a,a'\in A}]$.

 \RD{}
 Let~$D[t,\q_1,\dots,\q_m,(\eps_{a,a'})_{a,a'\in A}]=\true$.
 Then, there is an \nonwasteful{} assignment~$\pi_t\colon A\to\rSet\cup\{\emptyset\}$
 and an entry
 $D[t-1,\q_1',\dots,\q_m',(\eps_{a,a'}')_{a,a'\in A}]=\true$
 such that
 for all~$i\in\set{m}$,
 $\q_i'\geq \sum_{a\in \pi^{-1}_t(\r_i)} w_a$,
 $\q_i=\rf_i(q_i'-\sum_{a\in \pi^{-1}_t(\r_i)} w_a)$,
 and
 $\sum_{a\in \pi^{-1}_t(\r_i)} w_a>0$ only if~$q_i'\geq \lb_i$,
 and for all~$a,a'\in A$,
 $\eps_{a,a'}=\eps_{a,a'}' + v_a(\pi_t(a)) - v_a(\pi_t(a'))$.
 By induction hypothesis,
 $D[t-1,\q_1',\dots,\q_m',(\eps_{a,a'}')_{a,a'\in A}]=\true$
 implies that
 there is a \xassignment[$(t-1)$]~$\pi'=(\pi_1,\dots,\pi_{t-1})$ that
 respects \eqref{model:noe}--\eqref{model:vp},
 results in envy-differences~$\eps_{a,a'}'=\sum_{t'=1}^{t-1} (v_a(\pi_{t'}(a)) - v_a(\pi_{t'}(a')))$,
 and yields quantities~$\q_1',\dots,\q_m'$ after period~$t-1$.
 Thus,
 $\pi=(\pi_1,\dots,\pi_{t-1},\pi_t)$ is a \xassignment[$t$]
 that results in envy-differences~$\eps_{a,a'}=\sum_{t'=1}^t (v_a(\pi_{t'}(a)) - v_a(\pi_{t'}(a')))$
 and quantities~$\q_1,\dots,\q_m$ after period~$t$.
 It remains to show that~$\pi$ respects \eqref{model:noe}--\eqref{model:vp}.
 Due to
 $\q_i'\geq \sum_{a\in \pi^{-1}_t(\r_i)} w_a$,
 $\pi$ respects \eqref{model:noe}.
 By the domain restrictions in~$\{0,1,\dots,\ub_i\}$ for each~$\q_i$,
 $\pi$ respects \eqref{model:uc}.
 Finally,
 due to
 $\sum_{a\in \pi^{-1}_t(\r_i)} w_a>0$ only if~$q_i'\geq \lb_i$,
 $\pi$ respects \eqref{model:vp}.

 \LD{}
 Let $\pi=(\pi_1,\dots,\pi_t)$ be a \xassignment[$t$] that respects \eqref{model:noe}--\eqref{model:vp},
 results in envy-differences~$\eps_{a,a'}=\sum_{t'=1}^t (v_a(\pi_{t'}(a)) - v_a(\pi_{t'}(a')))$,
 and yields quantities~$\q_i=q_i^{\pi}(t+1)$ for all~$i\in\set{m}$ after period~$t$.
 Then,
 $\pi'=(\pi_1,\dots,\pi_{t-1})$ is a \xassignment[$(t-1)$] that respects \eqref{model:noe}--\eqref{model:vp},
 results in envy-differences~$\eps_{a,a'}'=\eps_{a,a'}-(v_a(\pi_{t}(a)) - v_a(\pi_{t}(a')))$,
 and quantities~$q_i'=q_i^{\pi'}(t)$.
 By no-over-extraction,
$q_i'\geq \sum_{a\in\pi_t^{-1}(\r_i)}w_a\geq0$.
Together with integrality and capacity,
this gives~$q_i'\in\{0,1,\dots,\ub_i\}$.
 $q_i'\geq \sum_{a\in \pi^{-1}_t(\r_i)} w_a$,
 $\sum_{a\in \pi^{-1}_t(\r_i)} w_a > 0$ only if~$q_i'\geq \lb_i$,
 and~$q_i = \rf_i(q_i'-\sum_{a\in \pi^{-1}_t(\r_i)} w_a)$.
 By induction,
 we know that~$D[t-1,\q_1',\dots,\q_m',(\eps_{a,a'}')_{a,a'\in A}]=\true$.
 Together with~$\pi_t$,
 it follows that~$D[t,\q_1,\dots,\q_m,(\eps_{a,a'})_{a,a'\in A}]=\true$.
\end{proof}

Finally,
we output~\yes{}
if there is an entry~$D[\tau,\q_1^*,\dots,\q_m^*,(\eps_{a,a'}^*)_{a,a'\in A}]=\true$
such that
$\lb_i\leq \q_i^*\leq \ub_i$ for all~$i\in\set{m}$
and~$\eps_{a,a'}^*\geq 0$ for all~$a,a'\in A$.

\begin{proof}[Proof of~\cref{thm:nm:integral}]
 We prove soundness,
 completeness,
 and \XP{} running time of the algorithm presented above.

 \emph{Soundness}:
 Let the algorithm output~$\yes$.
 Let $D[\tau,\q_1^*,\dots,\q_m^*,(\eps_{a,a'}^*)_{a,a'\in A}]=\true$
 such that
 $\lb_i\leq \q_i^*\leq \ub_i$ for all~$i\in\set{m}$
 and~$\eps_{a,a'}^*\geq 0$ for all~$a,a'\in A$.
 Due to~\cref{lem:nm:integral},
 we know that
 there is a \xassignment[$\tau$]~$\pi=(\pi_1,\dots,\pi_\tau)$ that respects \eqref{model:noe}--\eqref{model:vp},
 results in envy-differences~$\sum_{t=1}^\tau (v_a(\pi_{t}(a)) - v_a(\pi_{t}(a')))=\eps_{a,a'}^*$,
 and yields quantities~$\q_1^*,\dots,\q_m^*$ after period~$\tau$.
 Since $\lb_i\leq \q_i^*\leq \ub_i$ for all~$i\in\set{m}$,
 we know that also \eqref{model:tv} is respected,
 and thus,
 $\pi$ is sustainable.
 Since $\eps_{a,a'}^*\geq 0$ for all~$a,a'\in A$,
 $\pi$ is also envy-free.

 \emph{Completeness}:
 Let~$\pi=(\pi_1,\dots,\pi_\tau)$ be a sustainable envy-free \xassignment[$\tau$].
 Let~$\eps_{a,a'}^*\ceq \sum_{t=1}^\tau (v_a(\pi_{t}(a)) - v_a(\pi_{t}(a')))$.
 Note that due to envy-freeness,
 $\eps_{a,a'}^*\geq 0$ for all~$a,a'\in A$.
 Let~$\q_j^*\ceq q_j^{\pi}(\tau+1)$.
 Due to sustainability,
 we know that $\lb_i\leq \q_i^*\leq \ub_i$ for all~$i\in\set{m}$.
 Due to~\cref{lem:nm:integral},
 we know that $D[\tau,\q_1^*,\dots,\q_m^*,(\eps_{a,a'}^*)_{a,a'\in A}]=\true$.
 It follows that
 the algorithm outputs~\yes.

 \emph{Running time}:
 Table~$D$ has~$(\tau+1)\cdot (2\tau+1)^{n^2}\cdot \prod_{j=1}^m (1+\ub_j)$ entries.
For each entry,
we consider at most~$(m+1)^n$ possible assignments
$\pi\colon A\to\rSet\cup\{\emptyset\}$.
Each assignment determines a successor state,
which can be computed and checked in~$\poly(|I|)$ time.
The final requirement can also be checked in~$\poly(|I|)$ time.
Altogether,
this results in a running time of
$(\tau+1)\cdot (2\tau+1)^{n^2}\cdot (m+1)^n
\cdot (1+\max_{j\in\set{m}}\ub_j)^m
\poly(|I|)$.
\end{proof}

\begin{lemma}%
 \label{lem:nm:rescaling}
 Let~$I$ be an instance of \fish{} with unary encoding,
 $n$ agents,
 $m$ resources,
 and
 renewal functions~$\rf_i(x)=\alpha_i\cdot x+\beta_i$
 such that
 $\alpha_i\in\N$ for all~$i\in\set{m}$.
 Then,
 in~$|I|^{O(n+m)}$ time,
 we can construct an equivalent instance~$I'$
 such that
 (i) all agent weights and resource parameters are integral,
 and
 (ii) every numerical parameter of $I'$ has absolute value at most~$|I|^{O(n+m)}$.
\end{lemma}

\begin{proof}
 Let~$D$ denote the least common multiple of the denominators of
 $\{w_a\mid a\in A\} \cup \{\iq_i,\lb_i,\ub_i,\beta_i\mid i\in\set{m}\}$.
 Note that~$D\leq |I|^{n+4m}$,
 since due to the unary encoding,
 each of the $n+4m$ denominators is at most~$|I|$.
 So,
 obtain the new instance~$I'$ from~$I$
 by rescaling,
 i.e.,
 set~$w_a'\ceq D\cdot w_a$ for all~$a\in A$,
 and for all~$i\in\set{m}$,
 set
 $\iq_i'=D\cdot\iq_i$,
 $\lb_i'=D\cdot\lb_i$,
 $\ub_i'=D\cdot\ub_i$,
 and~$\beta_i'=D\cdot \beta_i$.
 Leaving the~$\alpha_i$'s untouched,
 we have~$\rf_i'(x)=\alpha_i\cdot x + \beta_i'$.
 Now, all numbers are integral,
 and the maximum absolute value is at most~$|I|^{n+4m+1}$.
 We claim that~$I$ is a \yes-instance if and only if~$I'$ is a \yes-instance.
 Note that either way,
 since the approval graph is the same,
 any assignment is envy-free for~$I$ if and only if it is for~$I'$.
 Hence,
 it remains to discuss sustainability.
 Let~$\pi$ be a \xassignment[$\tau$].
 We have that
 $D\cdot \q^{\pi}_{i}(1) = \q^{\prime\pi}_{i}(1)$.
 For~$t>1$,
 we have
 \begin{align*}
  D\cdot \q^{\pi}_{i}(t) &=
	\alpha_i\cdot\left( D\cdot \q^{\pi}_{i}(t - 1) - D\cdot h_i^\pi(t-1) \right) +D\cdot \beta_i
	\\
	&= \alpha_i\cdot\left( \q^{\prime\pi}_{i}(t - 1) - h_i^{\prime\pi}(t-1) \right) +\beta_i'
	= \q^{\prime\pi}_{i}(t),
 \end{align*}
 where in the second equality we used the fact that
 all weights are rescaled by~$D$.
 Thus,
 due to the rescaled resource-parameters,
 we have
 for every $t\in\set{\tau}$ that
 $h_i^\pi(t)\le \q_i^\pi(t)$ if and only if $h_i^{\prime\pi}(t)\le \q_i^{\prime\pi}(t)$;
 for every $t \in \set{\tau+1}$,
 that $\q_i^{\pi}(t)\leq \ub_i$ if and only if $\q_i^{\prime\pi}(t)\leq \ub_i'$,
 and
 $\q_i^{\pi}(t)\geq \lb_i$ if and only if $\q_i^{\prime\pi}(t)\geq \lb_i'$.
\end{proof}

\begin{proof}[Proof of \cref{thm:nm:integralslope}]
 Let~$I$ be an input instance of \fish{}
 with~$n$ agents,
 $m$ resources,
 and~$\tau$ periods such that all renewal functions have integral slope.
 Apply \cref{lem:nm:rescaling}
 to obtain in~$|I|^{O(n+m)}$ time
 an equivalent instance~$I'$
 such that
 all agent weights and resource parameters are integral,
 and
 every numerical parameter of $I'$ has absolute value at most~$|I|^{O(n+m)}$.
 Now,
 we solve~$I'$ using \cref{thm:nm:integral}.
 Since~$|I'|\in |I|^{O(n+m)}$,
 the algorithm behind \cref{thm:nm:integral} runs in~$|I'|^{f(n,m)}=|I|^{g(n,m)}$ time for some function~$g$ only depending on~$n$ and~$m$.
 Altogether,
 we solved~$I$ in \XP-time regarding~$n+m$.
\end{proof}
}%

\section{Outlook}
Our results raise the question of how much of the complexity
stems from sustainability itself and how much from
the additional requirement of envy-freeness.
One direction is to study sustainability alone or variants with
individual utility guarantees, inspired by satisfaction requirements in multistage voting~\cite{BredereckFK22,ijcai} or maximin objectives in repeated matching~\cite{LimNeohTeh26}.
Such variants could require a minimum number of approved assignments per agent or to maximize the minimum total utility.
Comparing these with our problem would help identify which tractability boundaries depend on the fairness criterion.

\section{Use of generative AI}
We used generative AI (ChatGPT/Claude) for literature research,
proofreading,
checking proofs/references,
and improving the presentation. %
All AI output was checked by the authors.

\bibliography{ref}

\clearpage
\newpage
\appendix
\section{Supplementary Material}
\appendixProofText

\end{document}